\documentclass[11pt]{article}
\usepackage{amsmath}
\usepackage{amssymb}
\usepackage{enumerate}
\usepackage{fullpage}
\usepackage{graphicx}
\usepackage{subfig}
\usepackage{bbm}
\usepackage{mathrsfs}
\usepackage{epstopdf}
\usepackage{epsfig}
\usepackage{color}
\usepackage{amsthm}
\def \beq {\begin{equation}}
\def \eeq {\end{equation}}

\newtheorem{theorem}{Theorem}
\newtheorem{lemma}[theorem]{Lemma}
\newtheorem{Condition}[theorem]{Condition}
\newtheorem{assumption}[theorem]{Assumption}
\newtheorem{cor}[theorem]{Corollary}
\newtheorem{prop}[theorem]{Proposition}
\numberwithin{equation}{section}
\newtheorem{corollary}[theorem]{Corollary}
\newtheorem{proposition}[theorem]{Proposition}

\newenvironment{example}{
\bigskip%
\par%
    \refstepcounter{theorem}%
  \noindent \textbf{Example \arabic{theorem}.}%
  \quad
}{%
\hfill$\clubsuit$%
\bigskip%
\par%
}

\begin{document}
\title{Dynamics of Order Positions and Related Queues in a Limit Order Book}
\author{Xin Guo\thanks{
Department of Industrial Engineering and Operations Research, University of California at Berkeley, Berkeley, CA 94720-1777.
Email: xinguo@berkeley.edu. Tel: 1-510-642-3615.} \ \ \ \
Zhao Ruan\thanks{Department of Industrial Engineering and Operations Research,
University of California at Berkeley, Berkeley, CA 94720-1777. Email: zruan@berkeley.edu. }\ \ \ \
Lingjiong Zhu\thanks{School of Mathematics, University of Minnesota, Minneapolis, MN 55455. Email: zhul@umn.edu.}}
\date{\today}
\maketitle

\begin{abstract}
  Order positions are key variables in algorithmic trading. This paper
  studies the limiting behavior of order positions and related queues
  in a limit order book. In addition to the fluid and diffusion limits
  for the processes, fluctuations of order positions and related
  queues around their fluid limits are analyzed.  As a corollary,
  explicit analytical expressions for various quantities of interests
  in a limit order book are derived.
\end{abstract}

\section{Introduction}

In modern financial markets, automatic and electronic order-driven
trading platforms have largely replaced the traditional floor-based
trading; orders arrive at the exchange and wait in the {\it Limit
  Order Book (LOB)} to be executed. There are two types of buy/sell
orders for market participants to post, namely, market orders and
limit orders. A \emph{limit order} is an order to trade a certain
amount of security (stocks, futures, etc.) at a given specified
price. Limit orders are collected and posted in the LOB, which
contains the quantities and the price at each price level for all
limit buy and sell orders.  A \emph{market order} is an order to
buy/sell a certain amount of the equity at the best available price in
the LOB; it is then matched with the best available price and a trade
occurs immediately and the LOB is updated accordingly.  A limit order
stays in the LOB until it is executed against a market order or until
it is canceled; cancellation is allowed at any time without penalty.

The availability of both market orders and limit orders presents
market participants opportunities to manage and balance risk and
profit.  As a result, one of the most rapidly growing research areas
in financial mathematics has been centered around modeling LOB
dynamics and/or minimizing the inventory/execution risk with
consideration of the microstructure of LOB. A few examples include
\iffalse Alfonsi et al.~\cite{alfonsi2010optimale}, Alfonsi
et al.~\cite{alfonsi2012order}, Avellaneda and Stoikov~\cite{avellaneda2008high},
Bayraktar and Ludkovski~\cite{bayraktar2014liquidation},
Cartea and Jaimungal~\cite{cartea2013modelling},
Cartea et al.~\cite{cartea2014buy},
Cont et al.~\cite{cont2010stochastic},
Cont and Kukanov~\cite{cont2013optimal},
Gu\'eant et al.~\cite{gueant2012optimal},
Guilbaud and Pham~\cite{guilbaud2013optimal},
Guo~\cite{guo2013tutorials},
Kirilenko et al.~\cite{kirilenko2013multiscale},
Laruelle et al.~\cite{laruelle2011optimal},
Maglaras et al.~\cite{maglaras2012optimal},
Predoiu et al.~\cite{predoiu2011optimal},
Shreve et al.~\cite{shreve2014diffusion}, and
Veraarta~\cite{veraart2010optimal}.  \fi
%
\cite{alfonsi2010optimale, alfonsi2012order, avellaneda2008high,
  bayraktar2014liquidation, cartea2013modelling, cartea2014buy,
  cont2010stochastic, cont2013optimal, gueant2012optimal,
  guilbaud2013optimal, guo2013tutorials, kirilenko2013multiscale,
  laruelle2011optimal, maglaras2012optimal, predoiu2011optimal,
  shreve2014diffusion, veraart2010optimal}.

At the core of these various optimization problems is the trade-off
between the inventory risk from unexecuted limit orders and the cost
from market orders.  While it is straightforward to calculate the costs
and fees of market orders, it is much harder to assess the inventory
risk from limit orders.  Critical to the analysis is the dynamics of
an order position in an LOB.  Because of the price-time priority
(i.e., best-priced order first and first-in-first-out) in most
exchanges in accordance with regulatory guidelines, a better order
position means less waiting time and a higher probability of the order
being executed.  In practice, reducing low latency in trading and
obtaining good order positions is one of the driving forces behind the
technological race among high-frequency trading firms.  Recent
empirical studies by Moallemi and Yuan~\cite{moallemi2015value} show
that values of order positions (if appropriately defined) have the
same order of magnitude of a half spread.  Indeed, analyzing order
positions is one of the key components for studying algorithmic
trading strategies.
%However, this topic has not been studied much with the exception of
%some limited analysis on the probability of an order being executed as in Hult and Kiessling \cite{hult2010algorithmic} and
%Cont, Stoikov, and Talreja \cite{cont2010stochastic}. % The probability
%that a particular order in a limit order queue is executed depends on
%many factors including the queue length, its position in the LOB, the
%frequency of price changes, the arrival rate of market orders, the
%cancellation of orders, and the distribution of order sizes (see
%Figure 1).
Knowing both the order position and the related queue lengths not only
provides valuable insights into the trading direction for the
``immediate'' future but also provides additional risk assessment for
the order --- if it were good to be in the front of any queue, then it
would be even better to be in the front of a {\it long} queue.
Therefore, it is important to understand and analyze the dynamics of
order positions together with their related queues.  This is the focus
of our work.
\begin{figure}[htb]
\label{Figure LOB}
\begin{center}
\includegraphics[width=3.5in,height=3in]{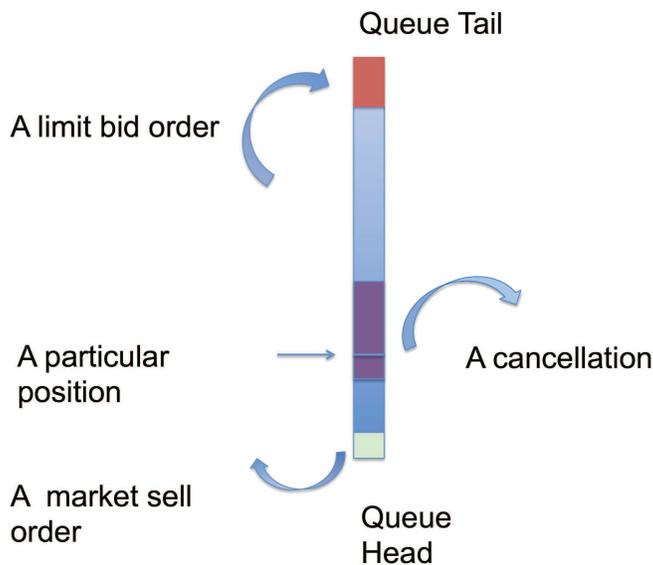}
\end{center}
\caption{Orders happened in the best bid queue.}
\end{figure}

\paragraph{Our contributions.}
The dynamics of an order position in a queue will be affected by both
the market orders and the cancellations, and the dynamics of its
relative position in a queue will be affected by the limit orders as
well (see Figure 1). Without loss of generality, we will focus on an
order position in the best bid queue along with the best bid and ask
queues.  Order positions in other queues will be similar and simpler
because of the absence of market orders.

First, we derive the fluid limit for the order positions and related
best bid and ask queues; in a sense, this is a first order
approximation to the processes.  We show (Theorem~\ref{FluidThm} and
Theorem~\ref{FluidThmGeneralized}) that the rate of the order position
approaching zero is proportional to the mean of order arrival
intensities and to the summation of the average size of market orders
and the ``modified'' average size of cancellation orders in the queue;
this modification depends on different assumptions on order
cancellations.  We also derive the (average) time it takes for the
order position to be executed.  The derivation is via two steps. The
first step is to establish the functional strong law of large numbers
for the related bid/ask queues; this is straightforward.  The second
step is intuitive but requires a delicate analysis involving passing
the convergence relation of stochastic processes in their
corresponding c\`adl\`ag space with the Skorokhod topology to their
integral equations.

Next, we proceed to the second order approximation for order positions
and related queues.  The first step is to establish appropriate forms
of the diffusion limit for the bid and ask queues.  We establish a
multi-variate functional central limit theorem (FCLT) using ideas from
random fields.  Under appropriate technical conditions, we show
(Theorem~\ref{DiffThm}) that the queues are two-dimensional Brownian
motion with mean and covariance structure explicitly given in terms of
the statistics of order sizes and order arrival intensities.  The
second step is to combine the FCLTs and the fluid limit results to
show (Theorem~\ref{Thm fluctuation}) that fluctuations of the order
positions are Gaussian processes with ``mean-reversion''.  The
mean-reverting level is essentially the fluid limit of order position
relative to the queue length modified by the order book net flow,
which is defined as the limit order minus the market order and the
cancellation.  The speed of the mean-reversion is proportional to the
order arrival intensity and the rate of cancellations.

\iffalse
Our analysis builds on techniques and results from classical
probability theory, including the functional central limit theorems by
Jacod and Shiryaev~\cite{jacod1987limit}, Glynn and Ward~\cite{glynn1988ordinary},
and Bulinski and Shashkin~\cite{bulinski2007limit}, the convergence of stochastic processes by
Kurtz and Protter~\cite{kurtz1991weak}, and the sample path
large deviation principle by Dembo and Zajic~\cite{dembo2009large}.
\fi

%We remark here that instead of the fluid limit approach presented in this paper, one may be
%tempted to consider the diffusion limit of order positions. While this is clearly a legitimate
%and mathematical interesting subject, the precise economic meaning of order positions after
%re-centering and rescaling for the second order is less clear compared to its fluid limit
%counterpart. In particular, such a limit after recentering  may not yield meaningful results
%for quantities of interest in LOB such as the probability of an order being executed.
Our results are built on fairly general technical assumptions
(stationarity and ergodicity) on order arrival processes and order
sizes.  For instance, order arrival processes (Section~\ref{Sec application}) can be Poisson processes or Hawkes processes, both of
which have been extensively used in LOB modelings; see for instance,
Abergel and Jedidi~\cite{AJ2015} and Huang, Lehalle, and
Rosenbaum~\cite{HLR2015}.

Practically speaking, studying order positions gives more direct
estimates for the ``value'' of order positions, which is useful in
algorithmic trading.  Indeed, based on the fluid limit, we derive
(Section~\ref{Sec application}) explicit analytical comparisons
between the average time an order is executed and the average time any
related queue is depleted. This is an important piece of information
especially when combined with an estimate on the probability of a
price increase.  The latter is a core quantity for the LOB and has
been studied in Avellaneda and Stoikov~\cite{avellaneda2008high} and
Cont and de Larrard~\cite{cont2013price, cont2012order} for special
cases.  In addition, we derive from the fluctuation analysis explicit
expressions for the first hitting times of the queue depletion, for
the expected order execution time, and for the fluctuations of order
execution time and first hitting times.  Furthermore, by the large
deviations theory, we derive the tail probability that the queues
deviate from their fluid limits.

\paragraph{Related work.}

The main idea behind our analysis is to draw connections between LOBs
and multi-class priority queues, as LOBs with cancellations are
reminiscent of reneging queues; see for instance Ward and
Glynn~\cite{ward2003diffusion, ward2005diffusion}.  In the
mathematical finance literature, there have been a number of papers on
modeling LOB dynamics in a queuing framework and establishing
appropriate diffusion and fluid limits for queue lengths or the order
book prices.  This line of work can be traced back to
Kruk~\cite{kruk2003functional}, who established diffusion and fluid
limits for prices in an auction setting and showed that the best bid
and ask queues converge to reflected two-dimensional Brownian motion
in the first quadrant.  Similar results were later obtained by Cont
and de Larrard~\cite{cont2012order} for the best bid and best ask
queues under heavy traffic conditions, where they also established the
diffusion limit for the price dynamics under the same ``reduced form''
approach with stationary conditions on the queue
lengths~\cite{cont2013price}.  Abergel and
Jedidi~\cite{abergel2013mathematical} modeled the volume of the order
book by a continuous-time Markov chain with independent Poisson order
flow processes and showed that mid price has a diffusion limit and
that the order book is ergodic.  Horst and Paulsen~\cite{horst2015law}
studied the fluid limit for the whole limit order books including both
prices and volumes, under a very general mathematical setting. Their
analysis was further extended in Horst and
Kreher~\cite{horst2015weak}, where the order dynamics could depend on
the state of the LOB.  Under different time and space scalings,
Blanchet and Chen~\cite{blanchet2013continuous} derived a pure jump
limit for the price-per-trade process and a jump diffusion limit for
the price-spread process.

One of our results, Theorem~\ref{DiffThm}, is mostly related to yet
different from the diffusion limit in ~\cite{cont2012order}.  This is a result of a different scaling
approach.  In order for us to analyze the dynamics of the order
positions, we need to differentiate limit orders from market and
cancellation orders, whereas in~\cite{cont2012order} order processes
are aggregated from limit, market, and cancellations orders.  Because
of this aggregation, they could use the main idea from
``heavy-traffic-limit'' in classical queuing theory and assume that
the mean order flow is dominated by the variance. While this
assumption~\cite[Assumption 3.2]{cont2012order} is critical to their
analysis, it does not hold in our setting where each individual order
type is considered.  On the other hand, if we were to impose this
assumption, then our result will be reduced to theirs because the
second term in Eqn.~(\ref{Eqn_DiffThm}) would simply vanish.
%Finally, from a modeling perspective, our focus of
%studying the dynamics of an order position makes it natural to
%consider only the (best) bid/ask queues, without having to impose the
%``random initialization'' assumption as in
%\cite{cont2012order}, thus avoiding the boundary condition of the
%integral type for the analysis.

To the best of our knowledge, the dynamics of order positions and its
relation to the queue lengths, which is the focus of our work, has not
been studied before.  Indeed, classical queuing tends to focus more on
the stability of the entire system, rather than analyzing individual
requests.  Most of the existing modeling approaches in algorithmic
trading have ignored order positions, with very limited efforts on the
probability of it being executed. For instance, such a probability is
either assumed to be a constant as in Cont and de
Larrard~\cite{cont2012order, cont2013price} and Guo, de Larrard and
Ruan~\cite{guo2013optimalp}, or is computed numerically from modeling
the whole LOB as a Markov chain as in Hult and
Kiessling~\cite{hult2010algorithmic}, or is analyzed with a
homogeneous Poisson process for order arrivals and with constant order
sizes as in Cont, Stoikov, and Talreja~\cite{cont2010stochastic}.

%{\it To analyze the appropriate limiting behavior of the dynamics of a particular order, it is
%without loss of generality, unless otherwise specified in this paper, to focus on the queues %on the best bid
%and the best ask. The analysis on the other queues would be similar but simpler because of the
%absence of market orders on other levels of bid and ask.}

\section{Fluid limits of order positions and related queues}

%We will start our analysis with the case  where the mean arrival rate of the order process
%$\lambda$ is constant. We will see that this analysis can be easily extended to allowing
%linear dependence between the mean arrival rate and the trading volume.

%The results from the fluid limit analysis enables us to  compute several quantities of
%interest in LOB,  including the probability distribution of the first time when bid/ask queues %are empty, and the expected time  of a particular order gets executed.

\subsection{Notation}

\newcommand{\obb}{\mathsf{bb}}
\newcommand{\ombb}{\mathsf{mbb}}
\newcommand{\ocbb}{\mathsf{cbb}}
\newcommand{\oba}{\mathsf{ba}}
\newcommand{\omba}{\mathsf{mba}}
\newcommand{\ocba}{\mathsf{cba}}

Without loss of generality, consider the best bid and ask queues. Then
there are six types of orders: best bid orders ($\obb$), market orders at the
best bid ($\ombb$), cancellation at the best bid ($\ocbb$), best ask ($\oba$),
market orders at the best ask ($\omba$), and cancellation at the
best ask ($\ocba$).
Denote the order arrival process by $\mathbf{N}=(N(t), t\ge 0)$ with
the inter-arrival times $\{D_i\}_{i\ge 1}$. Here
 \begin{eqnarray*}
N(t) & = & \max\left\{m: \sum_{i=1}^m D_i \le t\right\}.
\end{eqnarray*}
For simplicity, assume that there are no simultaneous arrivals of
different types of orders.
Consider order arrivals
of any of these six types as a point process, and define a sequence of six-dimensional random vectors
$\{\overrightarrow{V}_i \}_{i \ge 1}$, where for the $i$th order
$$\overrightarrow{V}_i
= (V_i^\obb, V_i^\ombb, V_i^\ocbb, V_i^\oba, V_i^\omba, V_i^\ocba)
:= (V_i^1, V_i^2, \ldots, V_i^6),$$
represents the sizes of the six types of orders; by the assumption,
exactly one entry of $\overrightarrow{V}_i$ is positive.
For instance, $\overrightarrow{V}_5=(0,0,0,4,0,0)$ means
the fifth order is of size $4$ and of type $\oba$, i.e.,
a limit order at the best ask.
%Followings are some basic quantities in our model:
%\begin{itemize}
%\item [$\cdot$] $q^b(t)$ stands for the total volume waiting for execution at the best bid level at time $t$;
%\item [$\cdot$] $q^a(t)$ stands for the total volume waiting for execution at the best ask level at time $t$;
%\item [$\cdot$] $z(t)$ stands for the position of a particular order at time $t$ we placed at the best bid queue at the beginning.
%\end{itemize}
%Naturally, before hitting the level $0$, $(q^b(t), q^a(t))$ could be computed by the following
%\begin{equation}
%(q^b(t), q^a(t))^T=(q^b(0), q^a(0))^T+\left(
%\begin{array}{cccccc}
%1, &-1, &-1, &0, &0, &0 \\
%0, &0, &0, &1, &-1, &-1
%\end{array}
%\right)\sum\limits_{i=1}^{N(t)} \overrightarrow{V}_i
%\end{equation}
%We use bold letters to denote the stochastic process as a whole, and regular letter with $t$ %for its value at time point $t$.
%For example, $\mathbf{N}$ denotes arrival process, while $N(t)$ is the number of orders %arrived by time $t$.
In this paper, we only consider c\`{a}dl\`{a}g processes.

For ease of references in the main text, we will use the following notation.
\begin{itemize}
\item $D[0, T]$ is the space of one-dimensional c\`{a}dl\`{a}g functions on $[0, T]$,
while $D^K[0, T]$ is the space of $K$-dimensional c\`{a}dl\`{a}g functions on $[0, T]$.
Consequently, the convergence in this space is, unless otherwise specified, in the sense of  the weak convergence in $D^K[0, T]$ equipped with $J_1$ topology;
\item
$L_{\infty}[0,T]$ is the space of functions $f:[0,T]\rightarrow\mathbb{R}^{K}$, equipped with the topology of uniform convergence;
\item
$\mathcal{AC}_{0}[0,T]$ is the space of functions
$f:[0,T]\rightarrow\mathbb{R}^{K}$ that are absolutely continuous
and $f(0)=0$;
\item
$\mathcal{AC}_{0}^{+}[0,T]$ is the space of non-decreasing functions $f:[0,T]\rightarrow\mathbb{R}^{K}$ that are absolutely continuous
and $f(0)=0$.\end{itemize}
Similarly, we define $D[0,\infty)$, $D^K[0,\infty)$, $L_{\infty}[0,\infty)$,
$\mathcal{AC}_{0}[0,\infty)$, $\mathcal{AC}_{0}^{+}[0,\infty)$ for $T=\infty.$

%\begin{assumption} \label{zta1}
%$\{D_{i}\}_{i \ge 1}$ is stationary array of positive random variables, and $\{\overrightarrow{V}_{i}\}_{i \ge 1}$
%is a stationary array of random vectors with finite expectation.
%\end{assumption}
%%\begin{assumption} \label{flui}
% %$\{D_{i}\}_{i \ge 1}$  is independent of $\{\overrightarrow{V}_{i}\}_{i \ge 1}$ .
%%\end{assumption}
%\begin{assumption} \label{wlln}
%As $i \to \infty$, there exists $\lambda >0$ and a constant vector $\overrightarrow{\bar{V}} = (\bar{V}^j>0, 1\le j\le 6)$, such that
%\begin{align}
%\frac{D_{1}+D_{2}+...+D_{i}}{i} \rightarrow  \frac{1}{\lambda},\\
%\frac{\overrightarrow{V}_{1}+\overrightarrow{V}_{2}+...+\overrightarrow{V}_{i}}{i} \rightarrow \overrightarrow{\bar{V}}.
%\end{align}
%in probability.
%\end{assumption}

\subsection{Technical assumptions and preliminaries}
In order to study  the fluid limit for the order position and related queues, we will first need to impose some technical assumptions.

\begin{assumption} \label{fluidd}
$\{D_{i}\}_{i \ge 1}$ is a stationary array of positive random variables with
\begin{align*}
\frac{D_{1}+D_{2}+ \cdots +D_{i}}{i} \rightarrow  \frac{1}{\lambda}, \qquad \text{in probability}
\end{align*}
as $i \to \infty$, where $\lambda$ is a positive constant.
\end{assumption}
\begin{assumption} \label{fluidv}
$\{\overrightarrow{V}_{i}\}_{i \ge 1}$
is a stationary array of square-integrable random vectors with
\begin{align*}
\frac{\overrightarrow{V}_{1}+\overrightarrow{V}_{2}+ \cdots +\overrightarrow{V}_{i}}{i} \rightarrow \overrightarrow{\bar{V}}, \qquad \text{in probability}
\end{align*}
as $i \to \infty$, where $\overrightarrow{\bar{V}} = (\bar{V}^1, \bar{V}^2, \ldots, \bar{V}^6)$
is a constant vector.
\end{assumption}

\begin{assumption} \label{flui}
$\{D_i\}_{i \ge 1}$ is independent of $\{\overrightarrow{V}_{i}\}_{i \ge 1}$ .
\end{assumption}

Now,  we define a new process $\overrightarrow{\mathbf{C}}_n$ as follows,
\begin{align} \label{cnfluid}
\overrightarrow{C}_n(t)=\frac{1}{n}  \sum_{i=1}^{N(nt)} \overrightarrow{V}_i.
%= \left(\frac{1}{n} \sum_{i=1}^{N(nt)} V_i^j, 1\le j \le 6\right).
\end{align}
We call such a process $\overrightarrow{\mathbf{C}}_n$ the {\it scaled net order flow process}.
\begin{theorem} \label{convC}
Given Assumptions \ref{fluidd} and \ref{fluidv}, for any $T>0$,
\[
\overrightarrow{\mathbf{C}}_n \Rightarrow \lambda \overrightarrow{\bar{V}}\mathbf{e}, \qquad \text{in $(D^6[0,T], J_1)$ \ \ as $n\rightarrow\infty$},
\]
where $\mathbf{e}$ is the identity function.
\end{theorem}

\begin{proof} \label{proof_convC}

First, we define the scaled processes $\mathbf{S}^D_n$ and $\overrightarrow{\mathbf{S}}^V_n$ by
\begin{align*}
S^D_n(t) &=\frac{1}{n} \sum_{i=1}^{\lfloor nt \rfloor} D_i, \\
\overrightarrow{S}^V_n(t) &=\frac{1}{n} \sum_{i=1}^{\lfloor nt \rfloor} \overrightarrow{V}_i.
%=\left(\frac{1}{n} \sum_{i=1}^{\lfloor nt \rfloor} V_i^j, 1\le j \le 6\right).
\end{align*}
Then by Assumption \ref{fluidd} and according to
Glynn and Whitt~\cite[Theorem 5]{glynn1988ordinary},
the strong Law of Large Numbers (SLLN) also follows, i.e.,
\begin{align*}
\lim\limits_{i \to \infty} \frac{D_{1}+D_{2}+ \cdots +D_{i}}{i} = \frac{1}{\lambda},  \qquad \text{a.s.}
\end{align*}
Then by the equivalence of SLLN and FSLLN~\cite[Theorem 4]{glynn1988ordinary},
it is clear that for any $T>0$,
\begin{align*}
\mathbf{S}_{n}^{D}=\frac{1}{n} \sum_{i=1}^{\lfloor n\cdot \rfloor} D_i \Rightarrow \frac{\mathbf{e}}{\lambda},
\qquad \text{a.s. in $(D[0, T], J_1)$ as $n\rightarrow\infty$} .
\end{align*}
Moreover,  since $\overrightarrow{V}_1$ is square-integrable, it follows that $\mathbb{E}[V^j_1] < \infty$ for $1\le j\le 6$. Note that $\{V^j_i\}_{i \ge 1}$ is stationary and
applying Birkhoff's Ergodic Theorem~\cite[Theorem 6.28]{breiman1968probability}
leads to
\begin{equation*}
\frac{1}{n} \sum_{i=1}^{n}V^j_i \rightarrow \mathbb{E}[V_1^j \mid \mathcal{I}^j], \qquad \text{a.s. as $n\rightarrow\infty$,}
\end{equation*}
where $\mathcal{I}^j$ is the invariant $\sigma$-algebra of $\{V_i^j\}_{i \ge 1}$.
Given the WLLN for $\{V_i^j\}_{i \ge 1}$, it follows that
\begin{align*}
\mathbb{E}[V_1^j \mid \mathcal{I}^j]=\bar{V}^j,
\end{align*}
and
\begin{equation*}
\frac{1}{n} \sum_{i=1}^{n}V^j_i \rightarrow \bar{V}^j, \qquad \text{a.s. as $n\rightarrow\infty$.}
\end{equation*}
Therefore, again by \cite[Theorem 4]{glynn1988ordinary},
\begin{align*}
\overrightarrow{\mathbf{S}}_{n}^{V,j}=\frac{1}{n} \sum_{i=1}^{\lfloor n\cdot \rfloor} V_i^j \Rightarrow \bar{V}^j\mathbf{e},
\qquad \text{a.s. in $(D[0,T], J_1)$ as $n\rightarrow\infty$}
\end{align*}
Since the limit processes for $\{\mathbf{S}_n^{D}\}_{n \ge 1}$ and $\{\mathbf{S}_n^{V,j}\}_{n \ge 1}$, $1\le j\le 6$, are deterministic, then according to \cite[Theorem 11.4.5]{whitt2002stochastic},
\begin{align*}
(\overrightarrow{\mathbf{S}}_{n}^{V}, \mathbf{S}_n^{D})\Rightarrow\left(\overrightarrow{\bar{V}}\mathbf{e},
\frac{\mathbf{e}}{\lambda}\right), \qquad\text{a.s. in $(D^7[0, T], J_1)$ as $n\rightarrow\infty$}.
\end{align*}
Finally, from \cite[Theorem 9.3.4]{whitt2002stochastic},
\begin{align*}
\overrightarrow{\mathbf{C}}_n \Rightarrow \lambda\overrightarrow{\bar{V}} \mathbf{e}, \enspace \mbox{in } (D^6[0,T], J_1) \mbox{ as } n\rightarrow\infty.\qquad\hfill \qedhere
\end{align*}
\end{proof}

Next we proceed to study the order position in the best bid and related queues of the best bid and the best ask.
We further assume
\begin{assumption}\label{cancelprop}
Cancellations are uniformly distributed on every queue.
\end{assumption}
We will see  that this assumption on cancellation is not critical, except for affecting the exact form of the fluid limit for the order position. (See Theorem \ref{FluidThmGeneralized} without this assumption in Section \ref{Sec discussion}.)

Now define the scaled queue lengths with $\mathbf{Q}_n^b$ for the best bid queue and  $\mathbf{Q}_n^a$ for the best ask queue,  and the scaled order position $\mathbf{Z}_n$ by
\begin{equation} \label{scaledq}
\left\{
\begin{aligned}
 Q_n^b(t) &=Q_n^b(0)+ C^1_n(t) - C^2_n(t) -C^3_n(t),\\
 Q_n^a(t) &=Q_n^a(0)+ C^4_n(t) - C^5_n(t) -C^6_n(t), \\
d Z_n(t) &=- dC^2_n(t) -\frac{Z_n(t-)}{Q_n^b(t-)}dC^3_n(t).
\end{aligned}
\right.
\end{equation}
The above equations are straightforward:  bid/ask queue lengths increase with limit orders
and decrease with market orders and cancellations according to their corresponding order flow processes;
an order position will decrease and move towards zero with arrivals of cancellations and market orders;
new limit orders arrivals will not change this particular order position; however,
the arrival of limit orders
may change the speed of the order position approaching zero following Assumption \ref{cancelprop}, hence the factor of  $\frac{Z_n(t-)}{Q_n^b(t-)}$.

Strictly speaking, Eqn.~\eqref{scaledq} only describes the dynamics of the
triple $(Q_n^b(t), Q_n^a(t), Z_n(t))$
{\it before} any of them hits zero: $\mathbf{Z}_n$ hitting zero means that the
order placed has been executed,
while $\mathbf{Q}_n^a$ hitting zero means that the best ask queues is depleted.
Since our primary interest is in the order position,
with little risk we may truncate the processes to avoid unnecessary technical difficulties on the boundary. That is,  define
\begin{equation}
\label{tau_n}
\tau_n = \min\{ \tau_n^z, \tau_n^a, \tau_n^b\},
\end{equation}
with
\begin{align*}
\tau_n^b = \inf\{t\ge 0: Q_n^b(t)\le 0\},  \qquad \tau_n^a = \inf\{t\ge 0: Q_n^a(t)\le 0\},   \qquad \tau_n^z = \inf\{t\ge 0: Z_n(t) \le 0\}.
\end{align*}
%Clearly, from the definition and the expression of $Z_n(t), Q_n(t)$, we have
%\begin{\lemma}
%$\tau^z_n \le \tau_n^b$.
%\end[lemma}
Now, define the truncated processes
\begin{align}
\label{truncated}
\tilde{Q}_n^b(t) = Q_n^b(t \wedge \tau_n), \qquad \tilde{Q}_n^a(t) = Q_n^a(t \wedge \tau_n), \qquad \tilde{Z}_n(t) = Z_n(t \wedge \tau_n).
\end{align}
Still, it is not immediately clear that these truncated processes would be well defined either
since we do not know {\it a priori}
if the term $-\frac{Z_n(t-)}{Q_n^b(t-)}$ is bounded when $\mathbf{Q}_{n}^{b}$ hits zero.
This, however, turns out not to be an issue.
\begin{lemma}\label{zlessthanqb}
Eqn.~\eqref{truncated} is well defined, with $Z_n(t) \le Q_n^b(t)$ for any time $t\le \min(\tau_n^z, \tau_n^a)$. In particular, $\tau_n^z\le \tau_n^b$.

\end{lemma}
\begin{proof}
Note that $\overrightarrow{\mathbf{C}}_n$ is a positive jumping process. Therefore,
when $\delta Z_n(t) = 0$, we have $\delta C^1_n(t) > 0$ and $\delta Q_n^b(t) > 0$;
when $\delta C^2_n(t) > 0$, we have $\delta Q_n^b(t) = \delta Z_n(t)$;
and when $\delta C^3_n(t) > 0$, we have $\frac{\delta Q_n^b(t)}{Q_n^b(t-)} = \frac{\delta Z_n(t)}{Z_n(t-)}$.
Hence, when $0 < Z_n(t-) \le Q_n^b(t-)$, we have $Z_n(t) \le Q_n^b(t)$.
Moreover, the number of order arrivals for any given time horizon is
finite with probability 1.
\end{proof}

%Now, we add $\mathbb{I}_{Z_n(t)>0}$ to the drift term for $\mathbf{Z}_n$ to stop the movement of $Z_n$ once $Q_n^b(t)$ hits zero,
%\begin{align}\label{scaledI2}
%d(Q_n^b(t), Q_n^a(t), Z_n(t))&=\left(
%\begin{array}{cccccc}
%1, &-1, &-1, &0, &0, &0 \\
%0, &0, &0, &1, &-1, &-1  \\
%0, &-1, & -\frac{Z_n^b(t)}{Q_n^b(t)}\mathbb{I}_{Z_n(t)>0}, &0, &0, &0
%\end{array}
%\right)
% \cdot d\mathbf{C}_n^{flu}(t)\\
%(Q_n^b(0), Q_n^a(0), Z_n(0)) &= (q^b(0), q^a(0), z(0)).
%\nonumber
%\end{align}
This lemma, though simple, turns out to play an important role to ensure that  fluid limits of order positions and related queues are well defined after rescaling.
That is, we can extend the definition of $\tilde{\mathbf{Q}}_n^b$, $\tilde{\mathbf{Q}}_n^a$, and $\tilde{\mathbf{Z}}_n$ for any time $t\ge 0$.

For simplicity, for the rest of the paper we will use 
$\mathbf{Q}_n^b$, $\mathbf{Q}_n^a$, and $\mathbf{Z}_n$ instead of $\tilde{\mathbf{Q}}_n^b$, $\tilde{\mathbf{Q}}_n^a$, and $\tilde{\mathbf{Z}}_n$, defined on $t\ge 0$.
The dynamics of the truncated processes could be described in the following matrix form.
\begin{align}\label{scaledII}
d\left(\begin{aligned}
&{Q}_n^b(t)\\
&{Q}_n^a(t)\\
&{Z}_n(t)\end{aligned}\right)
&=\left(
\begin{array}{cccccc}
1 &-1 &-1 &0 &0 &0 \\
0 &0 &0 &1 &-1 &-1  \\
0 &-1 & -\frac{{Z}_n(t-)}{{Q}_n^b(t-)} &0 &0 &0
\end{array}
\right)\mathbb{I}_{{Q}_n^a(t-)>0, Q_n^b(t-)>0, {Z}_n(t-)>0}
 \cdot d\overrightarrow{C}_n(t).
%({Q}_n^b(0), {Q}_n^a(0), {Z}_n(0)) &= (Q_n^b(0), Q_n^a(0), Z_n(0)).
%\nonumber
\end{align}
The modified processes coincide with the original processes before hitting zero,
which implies $\mathbb{I}_{t\le \tau_n} = \mathbb{I}_{{Q}_n^a(t-)>0, Q_n^b(t-)>0, {Z}_n(t-)>0}$.

In order to establish the fluid limit for the joint process  ($\mathbf{Q}_n^b$, $\mathbf{Q}_n^a$ and $\mathbf{Z}_n$),
we see that it is fairly standard to establish
the limit process for $(\mathbf{Q}_n^b, \mathbf{Q}_n^a)$ from classical probability theory where various forms of functional strong law of large numbers exist. However,
checking Eqn.~\eqref{scaledII} for $Z_n(t)$, we see that in order to pass from the fluid limit for $\mathbf{Q}_n^b$ to that for $\mathbf{Z}_n(t)$,  we effectively need to pass the convergence relation
between some  c\`adl\`ag processes $(\mathbf{X}_n, \mathbf{Y}_n)$ to $(\mathbf{X}, \mathbf{Y})$ in the
Skorokhod topology to the convergence relation between $\int X_n dY_n$
to $\int XdY$. That is, consider
a sequence of stochastic processes $\{\mathbf{X}_n\}_{n\ge1}$ defined by a sequence of SDEs
\begin{align} \label{SDEn}
X_n(t)=U_n(t)+\int_0^t F_n(X_n, s-)dY_n(s),
\end{align}
where $\{\mathbf{U}_n\}_{n\ge 1}$,
$\{\mathbf{Y}_n\}_{n\ge 1}$ are two sequences of stochastic processes and $\{F_n\}_{n\ge 1}$
is a sequence of functionals.  Now, suppose that
$\{\mathbf{U}_n, \mathbf{Y}_n, F_n\}_{n\ge1}$ converges to $\{\mathbf{U}, \mathbf{Y}, F\}$ in some way.
Then, would the sequence of the solutions to
Eqn.~\eqref{SDEn} converge to the solution to
\begin{align*}
X(t)=U(t)+\int_0^t F(X, s-)dY(s)?
\end{align*}

It turns out that such a convergence relation is delicate and can easily fail, as shown by the following simple example.
\begin{example}
\label{example1}
Let $\{X_{i}\}_{i\geq 1}$ be a sequence of identically distributed random variables taking values in $\{-1, 1\}$ such that
\begin{align*}
\mathbb{P}(X_1=1)=\mathbb{P}(X_1=-1)=\frac{1}{2},
\quad
\mathbb{P}(X_{i+1}=1 \mid X_{i}=1)=\mathbb{P}(X_{i+1}=-1 \mid X_{i}=-1)=\frac{3}{4} \mbox{ for } i > 1.
\end{align*}
Define $S_n(t)=\frac{1}{\sqrt{n}}\sum_{i=1}^{\lfloor nt \rfloor}X_i$.
Note that $\{X_{i}\}_{i\geq 1}$ is a strictly stationary sequence, with mean zero and  is a Markov Chain
with finite state space $\{-1,1\}$. Since each entry of the transition probability matrix is strictly
between $0$ and $1$, the sequence $\{X_{i}\}_{i\geq 1}$ is $\psi$-mixing, see e.g., \cite{Rosenblatt1971}.
Note that $\psi$-mixing implies $\phi$-mixing, see e.g., \cite{Bradley2005}. By stationarity,
\begin{equation*}
\lim_{n\rightarrow\infty}\frac{1}{n}\mathbb{E}\left[\left(\sum_{i=1}^{n}X_{i}\right)^{2}\right]
=\sigma^{2}=\mathbb{E}[X_{1}^{2}]+2\sum_{i=1}^{\infty}\mathbb{E}[X_{1}X_{i+1}].
\end{equation*}
We can compute by induction that for any $i\geq 1$,
\begin{equation*}
\mathbb{E}[X_{1}X_{i+1}]=\frac{3}{4}\mathbb{E}[X_{1}X_{i}]+\frac{1}{4}\mathbb{E}[X_{1}(-X_{i})]
=\frac{1}{2}\mathbb{E}[X_{1}X_{i}]
=\frac{1}{2^{i}}.
\end{equation*}
Therefore, we have $
\sigma^{2}=1+2\sum_{i=1}^{\infty}\frac{1}{2^{i}}=3.$
For strictly stationary centered $\phi$-mixing sequence with $\mathbb{E}\left[\left(\sum_{i=1}^{n}X_{i}\right)^{2}\right]\rightarrow\infty$
as $n\rightarrow\infty$ and $\mathbb{E}[|X_{1}|^{2+\delta}]<\infty$ for some $\delta>0$,
the invariance principle holds, see e.g., \cite{Ibragimov1975}, i.e., $\mathbf{S}_{n}$ converges
to $\sigma \mathbf{B}$.
Hence $\mathbf{S}_n$ converges to $\sqrt{3}\mathbf{B}$.
Now define a sequence of SDE's $dY_{n}(t)=Y_{n}(t)dS_{n}(t)$ with $Y_n(0)=1.$
Clearly, since $X_{i}\in\{\pm 1\}$ and $|X_{i}|\leq 1$, for sufficiently large $n$,
\begin{equation*}
Y_{n}(t)=\prod_{i=1}^{\lfloor nt \rfloor}\left(1+\frac{X_{i}}{\sqrt{n}}\right)
=e^{\sum_{i=1}^{\lfloor nt \rfloor}\log(1+\frac{X_{i}}{\sqrt{n}})}
=e^{\frac{1}{\sqrt{n}}\sum_{i=1}^{\lfloor nt \rfloor}X_{i}-\frac{1}{2n}\lfloor nt\rfloor+\epsilon_{n}},
\end{equation*}
where $|\epsilon_{n}|\leq\frac{C}{\sqrt{n}}$, where $C>0$ is a constant.
Hence, $\mathbf{Y}_n$ converges to the limiting process described by $\exp\{\sqrt{3}B(t) - \frac{t}{2}\}$, as $n \to \infty$.
However, the solution to
$dY(t)=Y(t)d(\sqrt{3}B(t))$ with  $Y(0)=1$
is given by $Y(t)=\exp\{\sqrt{3}B(t)-\frac{3t}{2}\}$.
\end{example}

Nevertheless, under proper conditions as specified in
Assumptions \ref{fluidd}, \ref{fluidv}, \ref{cancelprop}, one can
establish the desired convergence relation.  Such assumptions prove to
be sufficient using a result of
Kurtz and Protter~\cite[Theorem 5.4]{kurtz1991weak}.  For sake of
completeness, we present this result next, along with the technical
conditions required for the convergence.

\subsection{Detour: Convergence of stochastic processes by
Kurtz and Protter~\cite{kurtz1991weak}}

Define $h_{\delta}(r):[0,\infty) \to [0,\infty)$
by $h_{\delta}(r)=(1-\delta/r)^+$.
Define $J_{\delta}:D^m [0,\infty) \to D^m[0,\infty)$
by
\begin{align*}
J_{\delta}(x)(t)=\sum_{s \le t}h_{\delta}(|x(s)-x(s-)|)(x(s)-x(s-)).
\end{align*}
Let ${Y_{n}}$ be a sequence of stochastic processes adapted to ${\mathcal{F}_{t}}$.
Define $Y_n^\delta=Y_n-J_{\delta}(Y_n)$. Let $Y_n^\delta=M_n^\delta+A_n^\delta$
be a decomposition of $Y_n^\delta$ into an ${\mathcal{F}_{t}}$-local martingale and a process with finite variation.

\begin{Condition} \label{C1}
For each $\alpha > 0 $, there exist stopping times ${\tau_n^\alpha}$ such that $P\{\tau_n^\alpha \le 1\} \le 1/\alpha$
and $\sup_n \mathbb{E}[[M_n^\delta]_{t\leq\tau_n^\alpha}+T(A_n^\delta)_{t\leq\tau_n^\alpha}] < \infty$,
where $[M_n^\delta]_{t\leq\tau_n^\alpha}$ denotes the total quadratic variation of $M_n^\delta$
up to time $\tau_n^\alpha$, and $T(A_n^\delta)_{t\leq\tau_n^\alpha}$ denotes the total variation of $A_n^\delta$ up to time $\tau_n^\alpha$.
\end{Condition}

 Let $T_1[0, \infty)$ denote the collection of
non-decreasing mappings $\lambda$ of $[0, \infty)$ to $[0, \infty)$
(in particular, $\lambda(0)=0$)
such that $\lambda(h+t)-\lambda(t) \le h$ for all $t, h \ge 0$. Let $\mathbb{M}^{km}$ be the space of real-valued $k\times m$ matrices, and $D_{\mathbb{M}^{km}}[0, \infty)$ be the space of c\`{a}dl\`{a}g functions from $[0, \infty)$ to $\mathbb{M}^{km}$.
Assume that there exist mappings $G_n, G : D^k[0, \infty) \times T_1[0, \infty) \to D_{\mathbb{M}^{km}}[0, \infty)$
such that $F_n\circ\lambda =G_n(x\circ \lambda, \lambda)$ and $F(x)\circ \lambda = G(x\circ \lambda, \lambda)$
for $(x, \lambda) \in D^k[0, \infty) \times T_1[0, \infty)$.
\begin{Condition}\label{c2}
(i) For each compact subset $\mathcal{H} \subset D^k[0, \infty)$ and $t>0$,
$\sup_{(x, \lambda) \in \mathcal{H}}\sup_{s \le t}|G_n(x, \lambda, s)-G(x, \lambda, s)| \to 0$;

(ii) For $\{(x_n, \lambda^n)\} \in D^k[0, \infty) \times T_1[0, \infty)$, $\sup_{s\le t}|x_n(s)-x(s)| \to 0$
and $\sup_{s \le t}|\lambda^n(s)-\lambda(s)| \to 0$ for each $t>0$ implies $\sup_{s \le t}|G(x_n, \lambda^n, s)-G(x, \lambda, s)| \to 0$.
\end{Condition}
\begin{theorem} \label{kurtz}
Suppose that $(\mathbf{U}_n, \mathbf{X}_n, \mathbf{Y}_n)$ satisfies
\begin{align*}
X_n(t)=U_n(t)+\int_0^t F_n(X_n, s-)dY_n(s),
\end{align*}
$(\mathbf{U}_n, \mathbf{Y}_n) \Rightarrow (\mathbf{U}, \mathbf{Y})$
in the Skorokhod topology, and that $\{\mathbf{Y}_n\}$ satisfies
Condition \ref{C1} for some $0 < \delta \le \infty$. Assume that $\{F_n\}$ and $F$ have representations
in terms of $\{G_n\}$ and $G$ satisfying Condition \ref{c2}. If there exists a global solution $X$ of
\begin{equation*}
dX(t)=U(t)+\int_0^tF(X,s-)dY(s),
\end{equation*}
and the local uniqueness holds, then
\begin{align*}
(\mathbf{U}_n, \mathbf{X}_n, \mathbf{Y}_n) \Rightarrow (\mathbf{U}, \mathbf{X}, \mathbf{Y}).
\end{align*}
\end{theorem}

\subsection{Fluid limit for order positions and related queues}

We are now ready to establish our first  result.
\begin{theorem}\label{FluidThm}
Given Assumptions \ref{fluidd}, \ref{fluidv}, \ref{flui}, and \ref{cancelprop}, suppose
there exist constants $q^b$, $q^a$, and $z$ such that
\begin{align*}
(Q_n^b(0), Q_n^a(0), Z_n(0)) \Rightarrow (q^b, q^a, z).
\end{align*}
Then, for any $T>0$, %the scaled sequence of process $\tilde{\mathbf{Q}}_n^b$, $\tilde{\mathbf{Q}}_n^a$, and %$\tilde{\mathbf{Z}}_n$ defined in
Eqn.~\eqref{scaledII}
%converges to a limit process  $D^3[0,\infty)$ with $J_1$ topology. That is,
\begin{align*}
&({\mathbf{Q}}_n^b, {\mathbf{Q}}_n^a, {\mathbf{Z}}_n )  \Rightarrow ({\mathbf{Q}}^b, {\mathbf{Q}}^a, {\mathbf{Z}}) \qquad\text{in} \quad (D^3[0,T], J_1),
\end{align*}
where $({\mathbf{Q}}^b, {\mathbf{Q}}^a, {\mathbf{Z}})$ is given by
\begin{align}
&{Q}^{b}(t)=q^{b}-\lambda v^b(t\wedge\tau),\label{Qb}
\\
&{Q}^{a}(t)=q^{a}-\lambda v^a(t\wedge\tau),\label{Qa}
\end{align}
and for $t<\tau$,
\begin{equation} \label{Zd}
\frac{d{Z}(t)}{dt}=-\lambda\left(\bar{V}^{2}+\bar{V}^{3}\frac{{Z}(t-)}{{Q}^{b}(t-)}\right),\qquad {Z}(0)=z.
\end{equation}
Here $\tau =\min\{\tau^a, \tau^b, \tau^z\} $ with
\begin{equation}\label{taub}
\tau^{a}=\frac{q^{a}}{\lambda v^a}, \qquad\qquad \tau^{b}=\frac{q^{b}}{\lambda v^b},
\end{equation}
and
\begin{equation} \label{tauz}
\begin{aligned}
\tau^z=\left\{
\begin{aligned}
&\left(\frac{(1+c)z}{a}+b\right)^{c/(c+1)}b^{1/(c+1)}c^{-1}-b/c \qquad & \qquad c \notin \{-1, 0\}, \\
&b(1-e^{-\frac{z}{ab}}) \qquad & \qquad c=-1, \\
&b\log\left(\frac{z}{ab}+1\right) & \qquad c=0.
\end{aligned}
\right.
\end{aligned}
\end{equation}
Moreover, if $v^b>0$, $v^a>0$, and
$q^{a}/v^a>q^{b}/v^b,$
then $\tau_{n}^{z}\rightarrow \tau^{z}$ a.s. as $n\rightarrow\infty$, where
\begin{equation}\label{abc}
a=\lambda\bar{V}^2,\quad b=q^b/(\lambda\bar{V}^3),\quad
c=-\frac{v^b}{\bar{V}^{3}},
\end{equation}
\begin{equation}
\label{vbva}
v^b=-\bar{V}^{1}+\bar{V}^{2}+\bar{V}^{3}, \ \ v^a=-\bar{V}^{4}+\bar{V}^{5}+\bar{V}^{6}.
\end{equation}
\end{theorem}

\begin{proof}
Note that Eqns.~\eqref{Qb}, \eqref{Qa}, \eqref{Zd} satisfy the following SDE's
\begin{align}\label{fluidnotpro}
d\left(\begin{aligned}
&{Q}^b(t)\\
&{Q}^a(t)\\
&{Z}(t)\end{aligned}\right)
=&\left(
\begin{array}{cccccc}
1 &-1 &-1 &0 &0 &0 \\
0 &0 &0 &1 &-1 &-1 \\
0 &-1 &-\frac{{Z}(t-)}{{Q}^b(t-)}&0&0&0
\end{array}
\right)\mathbb{I}_{{Q}^a(t-)>0, {Q}^b(t-)>0, {Z}(t-)>0}
\lambda\overrightarrow{\bar{V}} dt; \\
({Q}^b(0), {Q}^a(0), {Z}(0))=&(q^b, q^a, z).\notag
\end{align}
Therefore, we first show the convergence to Eqn.~\eqref{fluidnotpro}.
Now,  set $Y_n=\overrightarrow{\mathbf{C}}_n$, $X_n=({\mathbf{Q}}_n^b, {\mathbf{Q}}^a_n, {\mathbf{Z}}_n)$, and
\begin{align*}
F_n(x, s-)=F(x, s-)=\left(\begin{aligned}
1& \quad -1&  -1& \quad 0& \quad 0& \quad 0&\\
0& \qquad 0&   0& \quad 1&  \quad -1& \quad -1&\\
0& \quad -1&  \quad -\frac{x^3(s-)}{x^1(s-)}& \quad 0& \quad 0& \quad 0&
\end{aligned}
\right)\mathbb{I}_{x(s-) > 0}.
\end{align*}
In order to apply Theorem \ref{kurtz}, we need to decompose $Y_n$. Now take $\delta=\infty$, define the filtrations
$\mathcal{F}^n_t := \sigma (\{N(s)\}_{0 \le s \le nt}, \{\overrightarrow{V}_i\}_{1 \le i \le N(nt)})$ and $\mathcal{G}_i
:= \sigma (\{\overrightarrow{V}_k\}_{1 \le k \le i})$,
\begin{align*}
M_n(t)=\frac{1}{n}\sum_{i=1}^{N(nt)} \overrightarrow{V}_i -
\mathbb{E}[\overrightarrow{V}_i \mid \mathcal{G}_{i-1} ],
\end{align*}
and
\begin{align*}
A_n(t)=Y_n(t)-M_n(t).
\end{align*}
We will show that $M_n$ is a martingale with respect to $\mathcal{F}^n_{t}$ and $\{Y_n\}_{n \ge 1}$
satisfies Condition \ref{C1} in Theorem \ref{kurtz}.

For $s \in [0, t)$, it is easy to see that
% $\mathcal{F}_{s}^n \cap (N(ns) < i) \subseteq \mathcal{F}^n_{\frac{1}{n}\sum_{k=1}^{i}D_k-} \cap (N(ns) < i)$.
Assumption \ref{flui} implies that
$\mathbb{E}[\overrightarrow{V}_i \mid
\mathcal{F}^n_{\frac{1}{n}\sum_{k=1}^{i-1}D_k}] =
\mathbb{E}[\overrightarrow{V}_i \mid \mathcal{G}_{i-1}]$. Thus
\begin{align*}
 &\mathbb{E}\left[\mathbb{E}\left[\overrightarrow{V}_i \mid \mathcal{G}_{i-1}\right]\big|\mathcal{F}^n_{s} \cap (N(ns) < i) \right] \notag\\
 &=\mathbb{E}\left[\mathbb{E}\left[\overrightarrow{V}_i \mid \mathcal{F}^n_{\frac{1}{n}\sum_{k=1}^{i-1}D_k}\right]\Big|\mathcal{F}^n_{s}\cap(N(ns) < i)\right]\notag \\
 &= \mathbb{E}\left[ \overrightarrow{V}_i\big| \mathcal{F}^n_{s} \cap ( N(ns) < i )\right].
\end{align*}
Meanwhile, $ \mathcal{F}^n_{\frac{1}{n}\sum_{k=1}^{i}D_k} \cap (N(ns) \ge i) \subseteq \mathcal{F}^n_{s} \cap (N(ns) \ge i)$. Thus
\begin{align*}
&\mathbb{E}\left[\mathbb{E}\left[\overrightarrow{V}_i \mid \mathcal{G}_{i-1}\right]\Big| \mathcal{F}^n_{s}\cap (N(ns) \ge i)\right]
\\
&=\mathbb{E}\left[\mathbb{E}\left[\overrightarrow{V}_i\mid \mathcal{F}^n_{\frac{1}{n}\sum_{k=1}^{i}D_k}\right]\Big|\mathcal{F}^n_{s}\cap (N(ns) \ge i)\right]\nonumber
\\
&=\mathbb{E}\left[ \overrightarrow{V}_i \big| \mathcal{F}^n_{\frac{1}{n}\sum_{k=1}^{i}D_k}\cap (N(ns) \ge i)\right]
\nonumber
\\
&=\mathbb{E}\left[\overrightarrow{V}_i \mid \mathcal{G}_{i-1}\cap (N(ns) \ge i)\right]. \nonumber
\end{align*}
Moreover, $\mathbb{E}\left[ \overrightarrow{V}_i\big| \mathcal{F}^n_{s} \cap (N(ns) \ge i)\right]=\overrightarrow{V}_i $
since $\overrightarrow{V}_i$ is measurable with respect to $\mathcal{F}^n_{s} \cap (N(ns) \ge i)$. Therefore,
\begin{align*}
\mathbb{E}\left[M_n(t)\big|\mathcal{F}^n_{s}\right]&=
\mathbb{E}\left[\sum_{i=1}^{N(nt)}
\frac{\overrightarrow{V}_i-\mathbb{E}[\overrightarrow{V}_i \mid \mathcal{G}_{i-1}]}{n}\bigg| \mathcal{F}^n_{s}\right]
\\
%&= \frac{1}{n} %\mathbb{E}\left[\sum_{i=1}^{N(nt)}\overrightarrow{V}_i-\mathbb{E}[\overrightarrow{V}_i|\mathcal{G}_{i-1}]\bigg| %\mathcal{F}^n_{s}\right] \\
&=\frac{1}{n}\sum_{i=1}^{N(ns)}  \left( \mathbb{E}\left[
    \overrightarrow{V}_i \mid  \mathcal{F}^n_{s} \cap (N(ns) \ge i)\right]
-\mathbb{E}\left[\mathbb{E}[\overrightarrow{V}_i \mid \mathcal{G}_{i-1}\cap (N(ns) \ge i)]\Big|\mathcal{F}^n_{s}\right]\right) \\
&\quad+\frac{1}{n}\mathbb{E}\left[\sum_{i=N(ns)+1}^{N(nt)}
\overrightarrow{V}_i-\mathbb{E}[\overrightarrow{V}_i \mid
\mathcal{G}_{i-1}]\bigg| \mathcal{F}^n_{s}\cap (N(ns) < i)\right] \\
&=\frac{1}{n}\sum_{i=1}^{N(ns)}  \left(  \overrightarrow{V}_i
  -\mathbb{E}[\overrightarrow{V}_i  \mid \mathcal{G}_{i-1} \cap (N(ns)\ge i)]\right)\\
&\quad+\frac{1}{n}\lambda n(t-s) \Big(  \mathbb{E}[\overrightarrow{V}_i\big|\mathcal{F}^n_{s}\cap (N(ns) < i)]
-\mathbb{E}\left[\mathbb{E}[\overrightarrow{V}_i \mid \mathcal{G}_{i-1}]\Big|\mathcal{F}^n_{s}\cap (N(ns) < i)\right] \Big)\\
&=\frac{1}{n}\sum_{i=1}^{N(ns)}  \left(  \overrightarrow{V}_i
  -\mathbb{E}[\overrightarrow{V}_i \mid \mathcal{G}_{i-1} \cap (N(ns) \ge i)]\right)=M_t(s).
\end{align*}
And $\mathbb{E} |M_n(t)| < \infty$ follows directly from Assumption \ref{fluidv}. Hence it follows that $M_n(t)$ is a martingale.
The quadratic variance of $M_n(t)$ is as follows:
\begin{align*}
\mathbb{E}\big[[M_n]_t\big]&= \frac{nt}{n^2} \sum_{j=1}^6 \mathbb{E}\left[ \lambda\left(V^j_i
-\mathbb{E}\left[V^j_i\big|\mathcal{G}_{i-1}\right]\right)^2\right] \\
&=\frac{t}{n}\sum_{j=1}^6 \lambda\mathbb{E} \left[\left(V^j_i\right)^2-2V^j_i
\mathbb{E}\left[V^j_i\big|\mathcal{G}_{i-1}\right]+\left(\mathbb{E}\left[V^j_i\big|\mathcal{G}_{i-1}\right]\right)^2\right]\\
&= \frac{t}{n}\sum_{j=1}^6 \lambda\left(\mathbb{E}\left(V^j_i\right)^2
-\mathbb{E}\left(\mathbb{E}\left[V^j_i\big|\mathcal{G}_{i-1}\right]\right)^2\right)
\le \frac{t}{n}\sum_{j=1}^6 \lambda\mathbb{E}\left(V^j_i\right)^2,
\end{align*}
since
\begin{align*}
\mathbb{E}\left[V^j_i\mathbb{E}\left[V^j_i\big|\mathcal{G}_{i-1}\right]\right]
&= \mathbb{E}\left[\mathbb{E}\left[V^j_i\mathbb{E}\left[V^j_i\big|\mathcal{G}_{i-1}\right]\right]\Big|\mathcal{G}_{i-1}\right]
=\mathbb{E}\left(\mathbb{E}\left[V^j_i\big|\mathcal{G}_{i-1}\right]\right)^2.
\end{align*}
Thus $\mathbb{E}\left[[M_n]_t\right]$ is bounded uniformly in $n$
since $\overrightarrow{V}_i$ is square-integrable.
Let $[T(A_n)]_t$ denote the total variation of $A_n$ up to time $t$. Then $\mathbb{E}\left[[T(A_n)]_t\right]$ is also uniformly bounded in $n$, as
\begin{align}
\mathbb{E}\left[[T(A_n)]_t\right]
&=t\sum_{j=1}^6 \lambda\mathbb{E}\left|\mathbb{E}\left[V^j_i\big|\mathcal{G}_{i-1}\right] \right|
\enspace \leq \enspace t\sum_{j=1}^6 \lambda\mathbb{E}\left[\mathbb{E}\left[|V^j_i|\big|\mathcal{G}_{i-1}\right]\right]
\label{conditionalJensen}
\\
&=t\sum_{j=1}^6 \lambda\mathbb{E}|V_1^j| \enspace < \enspace \infty,\label{L1L2}
\end{align}
where the inequality in
Eqn.~\eqref{conditionalJensen} uses the Jensen's inequality for conditional expectations
and
Eqn.~\eqref{L1L2} follows from the  square-integrability assumption.
Thus, $Y_n$ satisfies Condition \ref{C1} with $\tau_n^\alpha=\alpha+1$.
Moreover, taking $G_n(x\circ\mathbf{e},\mathbf{e})=F_n(x)=F(x)$,
it is easy to see that Condition \ref{c2} is satisfied according to \cite{kurtz1991weak}.

Now ${Q}^{a}(t)=0$ when $t=\tau^a$ as given in Eqn.~\eqref{taub};
$\tau^{a}>0$ if $\bar{V}^{4}-\bar{V}^{5}-\bar{V}^{6}<0$; otherwise ${Q}^{a}(t)$ never hits zero in which
case define $\tau^{a}=\infty$. The case for $\tau^b$ is similar.

 It remains to find the unique solution  for the limit Eqn.~\eqref{fluidnotpro}.
%hence these conditions are naturally satisfied.
%Clearly, Eqn.~\eqref{Qb} and
%Eqn.~\eqref{Qa} are solutions to ${Q}^b(t)$, ${Q}^a(t)$ before hitting 0. Moreover, $Z(t)$ satisfies %Eqn.~\eqref{Zd}.
%Similarly, ${Q}^{b}(t)=0$ when $
%t=\tau^{b}:=\frac{-q^{b}}{\lambda(\bar{V}^{1}-\bar{V}^{2}-\bar{V}^{3})}.$
%If $\tau^{b}>0$, i.e., $\bar{V}^{1}-\bar{V}^{2}-\bar{V}^{3}<0$ and ${Q}^{b}(t)$ never hits %zero otherwise (in which
%case we can define $\tau^{b}:=\infty$).
The equation for $Z(t)$ when ${Z}(t-)>0$
%can be simplified to
%\begin{align*}
%&\frac{d{Z}(t)}{dt}+\frac{1}{b+ct}{Z}(t-) \mathbb{I}_{t < \tau}=-a %\mathbb{I}_{t<\tau},\\
%&Z(0)=z.
%\end{align*}
is a first-order linear ODE with the solution
\begin{align}
\label{Z}
{Z}(t)=\left\{
\begin{aligned}
&-\frac{a}{1+c}(b+c(t\wedge\tau))+\left(z+\frac{ab}{1+c}\right)\left(\frac{b}{b+c(t\wedge\tau)}\right)^{1/c} \qquad & \qquad c \notin\{-1,0\},\\
&(a \log (b-(t\wedge\tau))+z/b-a\log b) \cdot (b-(t\wedge\tau)) \qquad & \qquad c=-1,
\\
&(z+ab)e^{-t/b}-ab & \qquad c=0.
\end{aligned}
\right.
\end{align}
From the solution, we can solve $\tau^z$ explicitly as given in Eqn.~\eqref{tauz}.
Note that the expression for ${Z}(t)$ may not be
monotonic and there might be multiple roots when $c\neq 0$.
Nevertheless, it is easy to check that the solution given in Eqn.~\eqref{tauz} is the smallest positive root.
For instance, when $c\notin\{-1,0\}$, there are two roots $-b/c$ and $\left(\frac{(1+c)z}{a}+b\right)^{c/(c+1)}b^{1/(c+1)}c^{-1}-b/c$
and when $c=-1$, there are two roots $b$ and $b(1-e^{-\frac{z}{ab}})$.  More computations confirm that
%\begin{itemize}
%\item [i] when $c > 0$, $-b/c <0 $ while %$\left(\frac{(1+c)z}{a}+b\right)^{c/(c+1)}b^{1/(c+1)}c^{-1}-b/c > 0$,
%\item [ii] when $c < 0$ and $c\neq -1$, $-b/c > %\left(\frac{(1+c)z}{a}+b\right)^{c/(c+1)}b^{1/(c+1)}c^{-1}-b/c > 0$,
%\item [iii] when $c=-1$, $b>b(1-e^{-\frac{z}{ab}})$.
%\end{itemize}
 indeed the smallest positive roots are $\tau^z = \left(\frac{(1+c)z}{a}+b\right)^{c/(c+1)}b^{1/(c+1)}c^{-1}-b/c$ for $c\notin\{-1,0\}$
and $\tau^z=b(1-e^{-\frac{z}{ab}})$ for $c=-1$.
Moreover, $\tau^z < \tau^b$ from these calculations.
Therefore $\tau=\min\{\tau^a, \tau^z\}$ is well defined and finite.
\end{proof}

%\begin{figure}[here]
% \centering
 %\subfloat[Case of $c \notin \{-1, 0\}$]{\label{fluid1}\includegraphics[width=0.5\textwidth]{fluid1.eps}}
 %\subfloat[Case of $c = %-1$]{\label{fluid2}\includegraphics[width=0.5\textwidth]{fluid2.eps}}\\
 %\subfloat[Case of $c = 0$]{\label{fluid3}\includegraphics[width=0.5\textwidth]{fluid3.eps}}
 %\caption{Illustration of the fluid limit $(Q^{b}(t),Q^{a}(t),Z(t))$}
%\end{figure}
The following figure  illustrates the fluid limits of $(Q^b(t), Q^a(t), Z(t))$
with $Q^b(0)=Q^a(0)=Z(0)=100$, $\lambda=1$, $\bar{V}^1=\bar{V}^4=1$, $\bar{V}^2=0.6, \bar{V}^3=0.8, \bar{V}^5=0.7, \bar{V}^6=0.8$. %parameters.
% Figure \ref{figure1} takes $Q^b(0)=Q^a(0)=Z(0)=100$, $\lambda=1$, $\bar{V}^1=\bar{V}^4=1$, $\bar{V}^2=0.6, %\bar{V}^3=0.8, \bar{V}^5=0.7, \bar{V}^6=0.8$.
%\begin{figure}[htp]
%\centering
%\subfloat[Queue lengths and the order position vs. %time]{\label{fluid1}
%\includegraphics[width=0.45\textwidth]{fluid1.eps}}
%\subfloat[Case of $c = %-1$]{\label{fluid2}\includegraphics[width=0.5\textwidth]{fluid2.eps}}\\
%\subfloat[Case of $c = 0$]{\label{fluid3}\includegraphics[width=0.5\textwidth]{fluid3.eps}}
%\subfloat[$Q^b(\tau) - Z(\tau)$ vs. $\bar{V}^2$ and %$\bar{V}^3$]{\label{fluid4}\includegraphics[width=0.45\textwidth]{dif_V2_V3.eps}}\\
%\subfloat[$Q^b(\tau) - Z(\tau)$ vs. %$\bar{V}^3$]{\label{fluid5}\includegraphics[width=0.45\textwidth]{dif_V3.eps}}
%\subfloat[$Q^b(\tau) - Z(\tau)$ vs. $\bar{V}^2$]{\label{fluid6}\includegraphics[width=0.45\textwidth]{dif_V2.eps}}\\
%\subfloat[$\tau^b - \tau^z$ vs. %$\bar{V}^3$]{\label{fluid7}\includegraphics[width=0.45\textwidth]{tau_V3.eps}}
%\subfloat[$\tau^b - \tau^z$ vs. %$\bar{V}^2$]{\label{fluid8}\includegraphics[width=0.45\textwidth]{tau_V2.eps}}
%\caption{Illustration of the fluid limit $(Q^{b}(t),Q^{a}(t),Z(t))$}
%\end{figure}
%Figure \ref{figure2} takes $\bar{V}^3=1.3$ with $\bar{V}^2$ varying from $1.3$ to $3.3$, and
%Figure \ref{figure3} takes $\bar{V}^2=1.3$ with $\bar{V}^3$ varying from $1.3$ to $3.3$.
\begin{figure}[htp]
\begin{center}
\includegraphics[width=3.5in,height=2.5in]{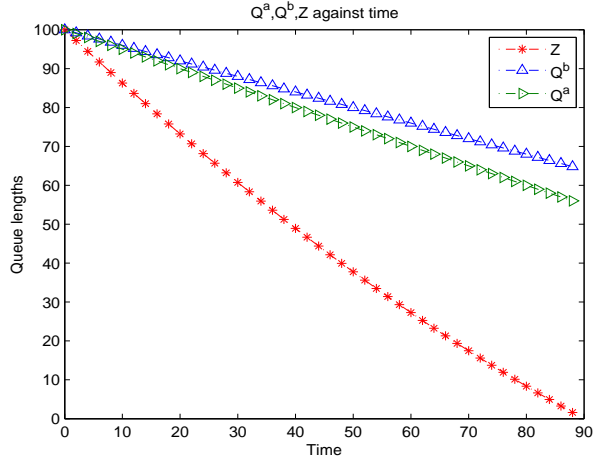}
\end{center}
\caption{Illustration of the fluid limit$(Q^{b}(t),Q^{a}(t),Z(t))$.}
 \label{figure1}
\end{figure}

%\begin{figure}[htp]
%\begin{center}
%\includegraphics[width=3.5in,height=2.5in]{zqratio_V2.eps}
%\end{center}
%\caption{Illustration of the ratio $Z(t)/Q^{b}(t)$ with different $\bar{V}^2$.}\label{figure2}
%\end{figure}

%\begin{figure}[htp]
%\begin{center}
%\includegraphics[width=3.5in,height=2.5in]{zqratio_V3.eps}
%\end{center}
%\caption{Illustration of the ratio $Z(t)/Q^{b}(t)$ with different $\bar{V}^3$.}\label{figure3}
%\end{figure}

\section{Fluctuation analysis}

The fluid limits in the previous section
are essentially  functional strong law of large numbers, and may well be regarded as
the ``first order'' approximation for order positions and related queues.
In this section, we will proceed to obtain a ``second order'' approximation for these processes.
We will first derive appropriate diffusion limits for the queues, and then analyze how these processes ``fluctuate'' around their corresponding fluid limits.
In addition, we will also apply the large deviation principles to compute the probability of the rare events that these processes deviate from their fluid limits.

\subsection{Diffusion limits for the best bid and best ask queues}
We will adopt the same notation for the order arrival processes as in the previous section. However,
we will need stronger assumptions for the diffusion limit analysis.

There is rich literature on multivariate Central Limit Theorems (CLTs) under some mixing conditions, e.g., Tone~\cite{tone2010central}.
However, these are not functional CLTs (FCLTs) with mixing conditions.
In the literature of limit theorems for associated random fields,
FCLTs are derived under some weak dependence conditions with explicit formulas for asymptotic covariance of the limit process.
Here, to establish FCLTs for $\{\overrightarrow{V}_i\}_{i \ge 1}$, we will
follow as in ~\cite{burton1986invariance}.
Readers can find more details in the framework of
Bulinski and Shashkin~\cite[Chapter 5, Theorem 1.5]{bulinski2007limit}.

\begin{assumption} \label{assumptionD}
$\{N(i,i+1]\}_{i\in\mathbb{Z}}$ is a stationary and ergodic sequence, with $\lambda:=\mathbb{E}[N(0,1]]<\infty$, and
\begin{equation}\label{L2Condition}
\sum_{n=1}^{\infty}\Vert\mathbb{E}[N(0,1]-\lambda \mid \mathcal{F}^{-\infty}_{-n}]\Vert_{2}<\infty,
\end{equation}
where $\Vert Y\Vert_{2}=(\mathbb{E}[Y^{2}])^{1/2}$ and $\mathcal{F}^{-\infty}_{-n}:=\sigma(N(i,i+1],i\leq -n)$.
\end{assumption}

%With this stationary, uniform mixing condition on $\{D_i\}_{i \ge 1}$, it follows that,
%\begin{theorem}
%If Assumption \ref{assumptionD} holds, then
%\begin{equation}
%v_d^2=\mbox{Var}(D_1)+2\sum_{i=2}^{\infty}\mbox{Cov}(D_1, D_i) < \infty
%\end{equation}
%and for any $T>0$,
%\begin{align}
%\mathbf{\Phi}_n^D(t) \Rightarrow  v_d \mathbf{W}_1 \qquad \text{a.s. in $(D[0,T], J_1)$ as $n\rightarrow\infty$}
%\end{align}
%where $\mathbf{W}_1$ is a standard Brownian motion.
%\end{theorem}
%\begin{proof}
%With the stationary, uniform mixing condition, thanks to Ward Whitte (\cite[Theorem 4.4.1]{whitt2002stochasticstochastic})
%the sequence $\{D_i\}_{i \ge 1}$ obeys the functional central limit theorem (FCLT), and the following series converges
%\begin{align}
%v_d^2=\mbox{Var}(D_1)+2\sum_{i=2}^{\infty}\mbox{Cov}(D_1, D_{i}).
%\end{align}
%i.e., the desired result holds.
%\end{proof}
%\begin{assumption}
%$\{\overrightarrow{V}_i\}_{i \ge 1}$ is a stationary, uniform mixing array of random vectors satisfying
%\begin{align}
%\mathbb{E}[\overrightarrow{V}_1] = \overrightarrow{\bar{V}}=(\bar{V}^{lb}, \bar{V}^{mb}, \bar{V}^{cb}, \bar{V}^{la}, \bar{V}^{ma}, \bar{V}^{ca})\\
%v_j^2=\mbox{Var}(V_1^j)+2\sum_{i=2}^{\infty} \mbox{Cov}(V_1^j, V_i^j) \qquad \text{for} \qquad j \in \mathbb{O}
%\end{align}
%\end{assumption}

%In addition to Assumption \ref{assumptionD} and Assumption \ref{flui}, we also need to impose %conditions for $\{\overrightarrow{V}_i\}_{i \ge 1}$
%to make the FCLT also hold for it.

\begin{assumption}
\label{assumptionV}
Let $n \in \mathbb{N}$ and $\mathcal{M}(n)$ denote the class of real-valued bounded coordinate-wise non-decreasing Borel functions on $\mathbb{R}^n$.
Let $|I|$ denote the cardinality of $I$ when $I$ is a set, and  $||\cdot||$ denote the $L^{\infty}$-norm. Let $\{\overrightarrow{V}_i\}_{i \ge 1}$
be a stationary sequence of $\mathbb{R}^6$-valued random vectors and for any finite set $I \subset \mathbb{N}$, $J \subset \mathbb{N}$,
and any $f, g\in \mathcal{M}(6|I|)$, one has
\begin{align*}
\text{Cov}(f(\overrightarrow{V}_I), g(\overrightarrow{V}_J)) \ge 0.
\end{align*}
Moreover, for $1\le j \le 6$,
\begin{equation*}
v_j^2 = \mbox{Var}(V_1^j)+2\sum_{i=2}^{\infty} \mbox{Cov}(V_1^j, V_i^j)  < \infty.
\end{equation*}
\end{assumption}

\paragraph{Remark.}
Note that an i.i.d. sequence $\{\overrightarrow{V}_i\}_{i \ge 1}$ clearly satisfies the above assumption if $\overrightarrow{V}_1$ is square-integrable.
It is not difficult to see that Assumption \ref{assumptionD} implies Assumption \ref{fluidd}, and Assumption \ref{assumptionV} implies Assumption \ref{fluidv}. In particular, Theorem \ref{FluidThm} holds under
 Assumptions \ref{assumptionD} and \ref{assumptionV}.

With these assumptions,  we can define
%To avoid confusion about the notation, we use different notation here. Roughly speaking,
%$\Phi$ corresponds to $S$ in the fluid setting, $\overrightarrow{\Psi}$ corresponds to
%$\overrightarrow{C}$ in the fluid setting, and $R$ corresponds to $Q$ in the fluid setting.
%The time scaled processes $\mathbf{\Phi}_{n}^{D}$, $\overrightarrow{\mathbf{\Phi}}_{n}^{V}$ %are defined by
 the centered and scaled net order flow $\overrightarrow{\mathbf{\Psi}}_{n}=(\overrightarrow{\Psi}_{n}(t), t\ge 0)$ by
\begin{align}\label{scaledc}
\overrightarrow{\Psi}_{n}(t)=\frac{1}{\sqrt{n}}\left(\sum_{i=1}^{N(nt)}\overrightarrow{V}_{i}-\lambda\overrightarrow{\bar{V}} nt\right).
%= \frac{1}{\sqrt{n}}\left(\sum_{i=1}^{N(nt)}V^j_{i}-\lambda\bar{V}^j nt, 1\le j \le 6 \right).
\end{align}
Here,
\begin{equation*}
\overrightarrow{\bar{V}}=(\bar{V}^j, 1\le j\le 6)=(\mathbb{E}[V^j_i], 1\le j\le 6),
  \end{equation*}
  is the mean vector of  order sizes.
%  $\overrightarrow{\mathbf{\Psi}}_{n}$ is a six-dimensional process, denoting %the net order flow. %Both the time scaled best bid queue $\mathbf{R}_n^b$
%and time scaled best ask queue $\mathbf{R}_n^a$ are driven by this process. We aim to show the %convergence of $(\mathbf{R}_n^b, \mathbf{R}_n^a)$,

Next, define $\mathbf{R}_n^b$ and $\mathbf{R}_n^a$, the time rescaled queue length for the best bid and best ask respectively, by
\begin{align*}
d R_n^b(t) = d (\Psi^1_n(t)+\lambda \bar{V}^1t) -d(\Psi^2_n(t)+\lambda \bar{V}^2t)-d(\Psi^3_n(t)+\lambda \bar{V}^3t),\\
d R_n^b(t) = d (\Psi^4_n(t)+\lambda \bar{V}^4t) -d(\Psi^5_n(t)+\lambda \bar{V}^5t)-d(\Psi^6_n(t)+\lambda \bar{V}^6t).
\end{align*}
The definition of the above equations is intuitive just as their fluid limit counterparts. The only modification here is that
 the drift terms is added back to the dynamics of the queue lengths because $\overrightarrow{\mathbf{\Psi}}$ has been re-centered.
The  equations can also be written in a more compact matrix form,
\begin{equation} \label{Qn}
d\left(\begin{aligned}
&R_n^b(t)\\ &R_n^a(t)
\end{aligned}\right)
=A\cdot d\left(\overrightarrow{\Psi}_n(t) + \lambda\overrightarrow{\bar{V}}t \right),
% (R_n^b(0), R_n^a(0)) &= (q^b(0), q^a(0)),
\end{equation}
with the linear transformation matrix
\begin{equation}
\label{equation A}
A=\left(
\begin{array}{cccccc}
1 &-1 &-1 &0 &0 &0 \\
0 &0 &0 &1 &-1 &-1
\end{array}
\right).
\end{equation}
However, Eqn.~\eqref{Qn} may not be well defined, unless  $R_n^b (t) >0 $ and $R_n^a (t) > 0$.
%For instance, when one of $R_n^b (t) >0 $ and $R_n^a (t) > 0$ hits zero, meaning the %corresponding queue is depleted, then .
As in the fluid limit analysis, one may truncate the process at the time when one of the queues vanishes. That is, define
\begin{equation} \label{taun}
\iota_n^a=\inf\{t: R_n^a(t)\le 0\}, \qquad \iota_n^b=\inf\{t: R_n^b(t) \le 0\}, \qquad \iota_n = \inf\{\iota_n^a, \iota_n^b \},
\end{equation}
and define the truncated process $({\mathbf{R}}_n^b, {\mathbf{R}}_n^a)$ by
\begin{equation}
\begin{aligned}\label{QnTruncated}
d\left(\begin{aligned}
&{R}_n^b(t)\\ &{R}_n^a(t)
\end{aligned}\right)
&=A \mathbb{I}_{t \le \iota_n}
 \cdot d\left(\overrightarrow{\Psi}_{n}(t) + \lambda\overrightarrow{\bar{V}}t \right) \qquad \text{with}\qquad
 \left(\begin{aligned}
 &{R}_n^b(0)\\ &{R}_n^a(0)
 \end{aligned}\right) &= \left(\begin{aligned}
 &R_n^b(0)\\ & R_n^a(0)
 \end{aligned}\right).
 \end{aligned}
\end{equation}
Now, we will show
\begin{theorem}\label{DiffThm}
Given Assumptions \ref{flui}, \ref{assumptionD}, and \ref{assumptionV}, for any $T>0$,
%Suppose the sequence $\{D_i\}_{i \ge 1}$ satisfies Assumption \ref{assumptionD}, and %$\{\overrightarrow{V}_{i}\}_{i \ge 1}$
%satisfies Assumption \ref{assumptionV}, and also Assumption \ref{flui} holds.
\begin{itemize}
\item We have
\begin{align}
\label{Eqn_DiffThm}
\overrightarrow{\mathbf{\Psi}}_{n}  \Rightarrow \overrightarrow{\mathbf{\Psi}} \stackrel{d.}{=}  \Sigma \overrightarrow{\mathbf{W}}\circ \lambda \mathbf{e}
- \overrightarrow{\bar{V}}v_d  \mathbf{W}_1 \circ \lambda \mathbf{e}  \qquad \text{in} \quad (D^6[0, T], J_1).
\end{align}
Here $\mathbf{W}_1$ is a standard scalar Brownian motion, $v_d$ is given by Eqn.~\eqref{vd}, $\overrightarrow{\mathbf{W}}$ is
a standard six-dimensional Brownian motion independent of $\mathbf{W}_1$,
$\circ$ denotes the composition of functions, and $\Sigma$ is given by $\Sigma \Sigma^T =(a_{jk})$ with
\begin{align} \label{sigma1}
a_{jk}&=\left\{
\begin{aligned}
v_j^2  \qquad & \text{for} \qquad j=k,\\
\rho_{j,k} v_jv_k \qquad &\text{for} \qquad j\neq k,
\end{aligned}
\right.
\end{align}
and
\begin{align}\label{sigma2}
\begin{split}
v_j^2&=\mbox{Var}(V_1^j)+2\sum_{i=2}^{\infty} \mbox{Cov}(V_1^j, V_i^j), \\
\rho_{j,k}&=\frac{1}{v_jv_k}\left( \mbox{Cov}(V_1^j,V_1^k)+\sum_{i=2}^{\infty}
\left( \mbox{Cov}(V_1^j, V_i^k)+\mbox{Cov}(V_1^k, V_i^j) \right) \right).
\end{split}
\end{align}
That is, $\overrightarrow{\mathbf{\Psi}}=(\mathbf{\Psi}^j, 1\le j\le 6)$ is a six-dimensional Brownian motion with zero drift
and variance-covariance matrix $(\lambda \Sigma^T\Sigma+\lambda v_d^2 \overrightarrow{\bar{V}}\cdot \overrightarrow{\bar{V}}^T)$.
\item If $(R_n^b(0), R_n^a(0)) \Rightarrow (q^b, q^a),$
then for any $T>0$,
\begin{align*}
\left(\begin{aligned}
&{\mathbf{R}}_n^b\\
&{\mathbf{R}}_n^a\end{aligned}\right)
 \Rightarrow
\left(\begin{aligned}
&{\mathbf{R}}^b\\
&{\mathbf{R}}^a\end{aligned}\right)
 \qquad \text{in}\quad (D^2[0, T], J_1).
\end{align*}
\end{itemize}
Here, the diffusion limit process $({\mathbf{R}}^b, {\mathbf{R}^a})^T$ up to the first hitting time of the boundary is a
two-dimensional Brownian motion with drift $\overrightarrow{\mu}$ and the variance-covariance matrix as
\begin{align} \label{mu_sigma}
\overrightarrow{\mu} :=(\mu_1, \mu_2)^T = \lambda A \cdot \overrightarrow{\bar{V}}
\qquad \text{and}\qquad \sigma\sigma^T := A\cdot (\lambda \Sigma^T\Sigma+\lambda v_d^2 \overrightarrow{\bar{V}}\cdot \overrightarrow{\bar{V}}^T)\cdot A^T.
\end{align}
\end{theorem}

\begin{proof}
First, define $\mathbf{N}_n$ by
\begin{align*}
N_n(t) = \frac{N(nt) - n \lambda t}{\sqrt{n}}.
\end{align*}
Now recall the FCLT from \cite[Page 197]{billingsley1968convergence}.
For a stationary, ergodic, and mean-zero sequence $(X_{n})_{n\in\mathbb{Z}}$,
that satisfies $\sum_{n\geq 1}\Vert\mathbb{E}[X_{0} \mid \mathcal{F}^{-\infty}_{-n}]\Vert_{2}<\infty$,
we have $\frac{1}{\sqrt{n}}\sum_{i=1}^{\lfloor n\cdot\rfloor}X_{i}\Rightarrow W_1(\cdot)$
on $(D[0,T],J_{1})$ with $v_d^{2}=\mathbb{E}[X_{0}^{2}]+2\sum_{n=1}^{\infty}\mathbb{E}[X_{0}X_{n}]<\infty$, where $\mathbf{W}_1$ is a standard one-dimensional Brownian motion.
Since the sequence $\{N(i,i+1]\}_{i\in\mathbb{Z}}$ satisfies Assumption \ref{assumptionD},
\begin{equation*}
\frac{N_{\lfloor n \cdot \rfloor}-\lambda\lfloor n\cdot \rfloor}{\sqrt{n}}\Rightarrow v_d W_1(\cdot),
\end{equation*}
in $(D[0,T],J_{1})$ as $n\rightarrow\infty$, where
\begin{equation} \label{vd}
v_d^{2}=\mathbb{E}[(N(0,1]-\lambda)^{2}]
+2\sum_{j=1}^{\infty}\mathbb{E}[(N(0,1]-\lambda)(N(j,j+1]-\lambda)]<\infty.
\end{equation}

Next, for any $\epsilon>0$ and $n$ sufficiently large,
\begin{align*}
&\mathbb{P}\left(\sup_{0\leq s\leq T}\left|\frac{N_{\lfloor ns\rfloor}-\lambda\lfloor ns\rfloor}{\sqrt{n}}
-\frac{N_{ns}-\lambda ns}{\sqrt{n}}\right|>\epsilon\right)
\\
%&=\mathbb{P}\left(\sup_{0\leq s\leq T}\left|(N_{\lfloor ns\rfloor}-N_{ns})
%+\lambda(ns-\lfloor ns\rfloor)\right|>\epsilon\sqrt{n}\right)\nonumber
%\\
%&\leq\mathbb{P}\left(\sup_{0\leq s\leq T}\left|N_{\lfloor %ns\rfloor}-N_{ns}\right|+\lambda>\epsilon\sqrt{n}\right)\nonumber
%\\
&\leq\mathbb{P}\left(\max_{0\leq k\leq \lfloor nT\rfloor,k\in\mathbb{Z}}N[k,k+1]>\epsilon\sqrt{n}-\lambda\right)\nonumber
\\
&\leq(\lfloor nT\rfloor+1)\mathbb{P}(N[0,1]>\epsilon\sqrt{n}-\lambda)\nonumber
\\
&\leq\frac{\lfloor nT\rfloor+1}{(\epsilon\sqrt{n}-\lambda)^{2}}
\int_{N[0,1]>\epsilon\sqrt{n}-\lambda}N[0,1]^{2}d\mathbb{P}\rightarrow 0,\nonumber
\end{align*}
as $n \rightarrow\infty$. Hence,  $\mathbf{N}_n \Rightarrow v_d \mathbf{W}_1$
in $(D[0,T],J_{1})$ as $n \rightarrow\infty$.

Moreover, thanks to ~\cite[Theorem 2]{burton1986invariance},
Assumption \ref{assumptionV} implies
\begin{align*}
\overrightarrow{\mathbf{\Phi}}_{n}^{V} \Rightarrow \Sigma \overrightarrow{\mathbf{W}} \qquad \text{in} \qquad (D^6[0, T], J_1),
\end{align*}
where $\overrightarrow{\mathbf{W}}$ is a standard six-dimensional Brownian motion and $\Sigma$ is a $6\times 6$ matrix
representing the covariance scale of the limit process. Furthermore,
the expression of $\Sigma$ by
Eqn.~\eqref{sigma1} and
Eqn.~\eqref{sigma2} can be explicitly computed following ~\cite[Theorem 2]{burton1986invariance}.

Now, by Assumption \ref{flui}, the joint convergence is
guaranteed by \cite[Theorem 11.4.4]{whitt2002stochastic}, i.e.,
\begin{equation*}
(\mathbf{N}_n, \overrightarrow{\mathbf{\Phi}}_n^V) \Rightarrow (v_d \mathbf{W}_{1}, \Sigma \overrightarrow{\mathbf{W}})  \qquad \text{in} \quad (D^7[0,T], J_1).
\end{equation*}
Moreover, by \cite[Corollary 13.3.2]{whitt2002stochastic}, we see
\begin{align*}
\overrightarrow{\mathbf{\Psi}}_{n}  \Rightarrow \overrightarrow{\mathbf{\Psi}} \stackrel{d.}{=}  \Sigma \overrightarrow{\mathbf{W}}\circ \lambda \mathbf{e}
- \overrightarrow{\bar{V}}v_d  \mathbf{W}_1 \circ \lambda \mathbf{e}  \qquad \text{in} \quad (D^6[0, T], J_1).
\end{align*}
To establish the second part of the theorem, it is clear that  the limiting process would satisfy
\begin{equation}\label{QTruncated}
\begin{aligned}
d\left(\begin{aligned}
&{R}^b(t)\\ &{R}^a(t)
\end{aligned}\right)
&=A \mathbb{I}_{t \le \iota}
 \cdot d\left(\overrightarrow{\Psi}(t) + \lambda\overrightarrow{\bar{V}}t \right), \\
({R}^b(0), & {R}^a(0))=(q^b, q^a),
\end{aligned}
\end{equation}
with
\begin{equation} \label{tau}
\iota^a=\inf\{t: R^a(t)\le 0\}, \qquad \iota^b=\inf\{t: R^b(t) \le 0\}, \qquad \iota = \min\{\iota^a, \iota^b \}.
\end{equation}
We now show that
\begin{equation}\label{convergeQn}
(\mathbf{R}_n^b, \mathbf{R}_n^a) \Rightarrow (\mathbf{R}^b, \mathbf{R}^a) \qquad \text{in} \quad (D^2[0,T], J_1).
\end{equation}
%where $(\mathbf{R}_n^b, \mathbf{R}_n^a)$ is defined by
%Eqn.~\eqref{Qn} and $(\mathbf{R}^b, %\mathbf{R}^a)$ is defined by
%Eqn.~\eqref{Q}.
According to the Cram\'{e}r-Wold device, it is equivalent to showing that for any $(\alpha, \beta) \in \mathbb{R}^2$,
\begin{equation} \label{cramer}
\alpha \mathbf{R}_n^b+\beta \mathbf{R}_n^a \Rightarrow \alpha \mathbf{R}^b+ \beta \mathbf{R}^a \qquad \text{in} \quad (D^2[0,T], J_1).
\end{equation}

Since $\overrightarrow{\mathbf{\Psi}}_{n} \Rightarrow \overrightarrow{\mathbf{\Psi}}$ in $(D^2[0, T], J_1)$, by the Cram{\'e}r-Wold device again,
\begin{equation*}
(\alpha, \beta)\cdot A\cdot \overrightarrow{\mathbf{\Psi}}_{n} \Rightarrow (\alpha, \beta)\cdot A\cdot \overrightarrow{\mathbf{\Psi}} \quad\mbox{in} \ \  (D^2[0, T], J_1).
\end{equation*}

By definition, it is easy to see that
\begin{equation*}
\alpha R_n^b(t)+\beta R_n^a(t) = (\alpha, \beta)\cdot A\cdot \left(\overrightarrow{\Psi}_{n}(t\wedge\iota_n)
+\overrightarrow{\bar{V}}(t\wedge\iota_n)\right) + \alpha q^b + \beta q^a.
\end{equation*}
Since the truncation function is continuous, by the continuous-mapping theorem, it asserts that
Eqn.~\eqref{cramer} holds and the desired convergence follows.

Moreover, because $\overrightarrow{\bar{V}}\mathbf{e}$ is deterministic and $\alpha q^b + \beta q^a$ is a constant,
we have the convergence in
Eqn.~\eqref{cramer}, as well as  the convergence in
Eqn.~\eqref{convergeQn}.
Note that $\iota_n, n \ge 1$ and $\iota$ are first passage times, by
\cite[Theorem 13.6.5]{whitt2002stochastic},
\begin{equation*}
(\iota_n, R_n^b(\iota_n-), R_n^a(\iota_n-)) \Rightarrow (\iota, R^b(\iota-), R^a(\iota-)).  \qedhere
\end{equation*}
\end{proof}

%\subsection{Compute some quantities of interest}
%We showed that under some certain conditions, the time scaled and centred net order flow $\overrightarrow{\mathbf{C}}_n^{dif}$
%converges to some limit process $\overrightarrow{\mathbf{C}}^{dif}$, which is a 6-dimensional Brownian motion.
%Since the limit queues $Q^b$ and $Q^a$ are driven by $\overrightarrow{\mathbf{C}}^{dif}$,
%we are able to compute some probabilities of interesting events here.
%Under Assumption \ref{assumptionD}, \ref{flui}, and \ref{assumptionV},
%the instantaneous variance-covariance matrix of the limit Brownian motion process is a 6-by-6 matrix given by
%
%
%\begin{align}
%\sigma^2 &=(\sigma_{jk}), \\
%\sigma_{jk}&=\left\{
%\begin{aligned}
%\lambda v_j^2+ \lambda^3(\bar{V}^j)^2v_d^2 \qquad & \text{if} \qquad j=k,\\
%\lambda \rho_{j,k} v_jv_k \qquad &\text{if} \qquad j\neq k.
%\end{aligned}
%\right.
%\end{align}

\paragraph{Remark.}
 Theorem \ref{DiffThm} holds without Assumption \ref{flui}, % which is only used to show the joint convergence of $(\mathbf{\Phi}_n^D, %\overrightarrow{\mathbf{\Phi}}_n^V)$.
%may be relaxed to allow dependence between the arrival process $\mathbf{N}$ and
%the order size sequence $\{\overrightarrow{V}_i\}_{i \ge 1}$
as long as $(\mathbf{\Phi}_n^D, \overrightarrow{\mathbf{\Phi}}_n^V)$ is guaranteed to converge jointly.

\subsection{Fluctuation analysis of queues and order positions}

Based on the diffusion and fluid limit analysis for the order position and related queues, one may consider fluctuations of
order positions and related queues around their perspective fluid limits.

%First, we consider fluctuation  of $(\mathbf{Q}_n^b, \mathbf{Q}_n^a, \mathbf{Z}_n)$ around its %fluid limit $(\mathbf{Q}^b, \mathbf{Q}^a, \mathbf{Z})$.

%We have proved that $({\mathbf{Q}}_n^b, {\mathbf{Q}}_n^a, {\mathbf{Z}}_n )$ %converges to
%its fluid limit $({\mathbf{Q}}^b, {\mathbf{Q}}^a, {\mathbf{Z}})$.
%Since
%the truncated processes are identical to the original processes up to the hitting time, we
%focus on $(\mathbf{Q}_n^b, \mathbf{Q}_n^a, \mathbf{Z}_n )$ and $(\mathbf{Q}^b, \mathbf{Q}^a, %\mathbf{Z})$.
%Now let us further study the fluctuations of $(\mathbf{Q}_n^b, \mathbf{Q}_n^a, \mathbf{Z}_n )$ %around its fluid limit $(\mathbf{Q}^b, \mathbf{Q}^a, \mathbf{Z})$
%before time $\tau$, with $\tau$  the first that the any of the three processes hits zero.

\begin{theorem}
\label{Thm fluctuation}
Given Assumptions \ref{flui}, \ref{cancelprop},  \ref{assumptionD}, and \ref{assumptionV}, we have
\begin{equation*}
\sqrt{n}
\left(
\begin{array}{c}
\mathbf{Q}^{b}_{n}-\mathbf{Q}^{b}
\\
\mathbf{Q}^{a}_{n}-\mathbf{Q}^{a}
\\
\mathbf{Z}_{n}-\mathbf{Z}
\end{array}
\right)
\Rightarrow
\left(
\begin{array}{c}
\mathbf{\Psi}^{1}-\mathbf{\Psi}^{2}-\mathbf{\Psi}^{3}
\\
\mathbf{\Psi}^{4}-\mathbf{\Psi}^{5}-\mathbf{\Psi}^{6}
\\
\mathbf{Y}
\end{array}
\right)
,\qquad\text{in $(D^{3}[0,\tau),J_{1})$}
\end{equation*}
as $n\rightarrow\infty$. Here $(\mathbf{Q}^b_n, \mathbf{Q}^a_n, \mathbf{Z}_n)$, $(\mathbf{Q}^b, \mathbf{Q}^a, \mathbf{Z})$ are given in
Eqn.~\eqref{scaledq} and Theorem \ref{FluidThm}, $(\mathbf{\Psi}^j, 1\le j\le 6)$ is given in Eqn.~\eqref{Eqn_DiffThm}, and $\mathbf{Y}$ satisfies
\begin{equation}\label{YEqn}
\begin{aligned}
&dY(t)
=\Big(\frac{Z(t)(\Psi^{1}(t)-\Psi^{2}(t)-\Psi^{3}(t))}{Q^{b}(t)}-Y(t)\Big)\frac{\lambda\bar{V}^{3}}{Q^{b}(t)}dt-d\Psi^{2}(t)-\frac{Z(t)}{Q^{b}(t)}d\Psi^{3}(t),
\end{aligned}
\end{equation}
with $Y(0)=0$.
\end{theorem}

\begin{proof}
Given Assumptions \ref{flui}, \ref{cancelprop},  \ref{assumptionD}, and \ref{assumptionV},
 we have from Theorem \ref{DiffThm},
\begin{equation*}
\overrightarrow{\mathbf{\Psi}}_{n}
=\frac{1}{\sqrt{n}}\left(\sum_{i=1}^{N(n\cdot)}\overrightarrow{V}_{i}-\lambda n\overrightarrow{\bar{V}}\mathbf{e}\right)\Rightarrow \overrightarrow{\mathbf{\Psi}},
\qquad \text{in} \quad (D^6[0, \tau), J_1).
\end{equation*}
%where
%\begin{equation}
%\overrightarrow{\mathbf{\Psi}}=\Sigma\overrightarrow{\mathbf{W}}\circ \lambda \mathbf{e}
%-\overrightarrow{\bar{V}}v_{d}\lambda \mathbf{W}_{1}\circ\lambda \mathbf{e}.
%\end{equation}
Hence, we have the following convergence in $(D[0, \tau), J_1)$,
\begin{equation*}
\begin{aligned}
\sqrt{n}(\mathbf{Q}_{n}^{b}-\mathbf{Q}^{b})\Rightarrow\mathbf{\Psi}^{1}-\mathbf{\Psi}^{2}-\mathbf{\Psi}^{3}, \\
\sqrt{n}(\mathbf{Q}_{n}^{a}-\mathbf{Q}^{a})\Rightarrow\mathbf{\Psi}^{4}-\mathbf{\Psi}^{5}-\mathbf{\Psi}^{6}.
\end{aligned}
\end{equation*}

Since Theorem  \ref{FluidThm} holds under Assumptions \ref{assumptionD} and \ref{assumptionV}, we now use the dynamics of $Z_{n}(t)$ in Eqn.~\eqref{scaledII} and $Z(t)$ in Theorem \ref{FluidThm} and get
%\begin{equation}
%dZ_{n}(t)=-dC_{n}^{2}(t)-\frac{Z_{n}(t-)}{Q_{n}^{b}(t-)}dC_{n}^{3}(t),
%\end{equation}
%and
%\begin{equation}
%dZ(t)=-dC^{2}(t)-\frac{Z(t-)}{Q^{b}(t-)}dC^{3}(t).
%\end{equation}
%Therefore, we have
\begin{align*}
d(Z_{n}(t)-Z(t))&=-d(C_{n}^{2}(t)-C^{2}(t))-\frac{Z_{n}(t-)}{Q_{n}^{b}(t-)}dC_{n}^{3}(t)+\frac{Z(t-)}{Q^{b}(t-)}dC^{3}(t)
\\
&=-d(C_{n}^{2}(t)-C^{2}(t))-\frac{Z_{n}(t-)}{Q_{n}^{b}(t-)}d(C_{n}^{3}(t)-C^{3}(t))+\left[\frac{Z(t-)}{Q^{b}(t-)}-\frac{Z_{n}(t-)}{Q_{n}^{b}(t-)}\right]dC^{3}(t).
\nonumber
\end{align*}
We can rewrite this as
\begin{equation*}
d(Z_{n}(t)-Z(t))+\left(\frac{Z_{n}(t-)-Z(t-)}{Q^{b}(t-)}\right)dC^{3}(t)=dX_{n}(t),
\end{equation*}
\begin{equation*}
X_{n}(t)=-(C_{n}^{2}(t)-C^{2}(t))-\int_{0}^{t}\frac{Z_{n}(s-)}{Q_{n}^{b}(s-)}d(C_{n}^{3}(s)-C^{3}(s))
+\int_{0}^{t}\frac{Z_{n}(s-)(Q_{n}^{b}(s-)-Q^{b}(s-))}{Q^{b}(s-)Q_{n}^{b}(s-)}dC^{3}(s).
\end{equation*}
Now,
\begin{equation*}
\sqrt{n}\mathbf{X}_{n}\Rightarrow-\mathbf{\Psi}^{2}-\int_{0}^{\cdot}\frac{Z(s-)}{Q^{b}(s-)}d\Psi^{3}(s)
+\int_{0}^{\cdot}\frac{Z(s-)(\Psi^{1}(s-)-\Psi^{2}(s-)-\Psi^{3}(s-))}{(Q^{b}(s-))^{2}}\lambda\bar{V}^{3}ds
\end{equation*}
As the limit processes $\overrightarrow{\mathbf{\Psi}}$ and $\mathbf{Q}^b$, $\mathbf{Q}^a$ are continuous, this could be changed into
\begin{equation*}
\sqrt{n}\mathbf{X}_{n}\Rightarrow-\mathbf{\Psi}^{2}-\int_{0}^{\cdot}\frac{Z(s)}{Q^{b}(s)}d\Psi^{3}(s)
+\int_{0}^{\cdot}\frac{Z(s)(\Psi^{1}(s)-\Psi^{2}(s)-\Psi^{3}(s))}{(Q^{b}(s))^{2}}\lambda\bar{V}^{3}ds.
\end{equation*}
Hence,
\begin{equation*}
\sqrt{n}(\mathbf{Z}_{n}-\mathbf{Z})\Rightarrow \mathbf{Y},
\end{equation*}
where $\mathbf{Y}$ satisfies Eqn.~\eqref{YEqn}.
%\begin{align}
%&dY(t)+\frac{Y(t)}{Q^{b}(t)}\lambda\bar{V}^{3}dt
%=-d\Psi^{2}(t)-\frac{Z_{t}}{Q^{b}(t)}d\Psi^{3}(t)
%\\
%&\qquad\qquad\qquad\qquad
%+\frac{Z(t)(\Psi^{1}(t)-\Psi^{2}(t)-\Psi^{3}(t))}{(Q^{b}(t))^{2}}\lambda\bar{V}^{3}dt.
%\nonumber
%\end{align}
\end{proof}

\subsection{Large deviations}

In addition to the fluctuation analysis in the previous section, one can further study
 the probability of the rare events that the scaled process $({Q}_{n}^{b}(t),{Q}_{n}^{a}(t))$
deviates away from its fluid limit. Informally, we are interested in the probability
$\mathbb{P}(({Q}_{n}^{b}(t),{Q}_{n}^{a}(t))\simeq(f^{b}(t),f^{a}(t)),0\leq t\leq T)$
as $n\rightarrow\infty$, where $(f^{b}(t),f^{a}(t))$ is a given pair of functions that can be different from
the fluid limit $({Q}^b(t), {Q}^a(t))$.

Recall that a sequence $(P_{n})_{n\in\mathbb{N}}$ of probability measures on a topological space $\mathbb{X}$
satisfies the large deviation principle with rate function $\mathcal{I}:\mathbb{X}\rightarrow\mathbb{R}$ if $\mathcal{I}$ is non-negative,
lower semi-continuous and for any measurable set $A$, we have
\begin{equation*}
-\inf_{x\in A^{o}}\mathcal{I}(x)\leq\liminf_{n\rightarrow\infty}\frac{1}{n}\log P_{n}(A)
\leq\limsup_{n\rightarrow\infty}\frac{1}{n}\log P_{n}(A)\leq-\inf_{x\in\overline{A}}\mathcal{I}(x).
\end{equation*}
The rate function is said to be good if the level set $\{x \mid I(x)\leq\alpha\}$ is compact for any $\alpha\geq 0$.
Here, $A^{o}$ is the interior of $A$ and $\overline{A}$ is its closure.
Finally,  the contraction principle in large deviation says that
if $P_{n}$ satisfies a large deviation principle on $X$ with rate
function $\mathcal{I}(x)$ and $F:X\rightarrow Y$ is a continuous map,
then the probability measures $Q_{n}:=P_{n}F^{-1}$ satisfies
a large deviation principle on $Y$ with rate function $I(y)=\inf_{x \mid  F(x)=y}\mathcal{I}(x)$.
Interested readers are referred to the standard references by Dembo and Zeitouni \cite{dembo2009large} and Varadhan \cite{varadhan1984large}
for the general theory of large deviations and its applications.

Recall that under Assumptions \ref{fluidd}, \ref{fluidv} and \ref{flui}, we had
a FLLN result for $(Q_{n}^{b}(t),Q_{n}^{a}(t))$ and under Assumptions \ref{assumptionD}, \ref{assumptionV}, and \ref{flui}, we had
a FCLT result for $(Q_{n}^{b}(t),Q_{n}^{a}(t))$.  It is natural to replace Assumptions \ref{fluidd}, \ref{fluidv} by some stronger
assumptions to obtain a large deviations result for $(Q_{n}^{b}(t),Q_{n}^{a}(t))$.
We will see that by assuming that $(\overrightarrow{V}_{i})_{i\in\mathbb{N}}$ and  $(N(i)-N(i-1))_{i\in\mathbb{N}}$
satisfy the following Assumptions \ref{LDPAssump} and \ref{LDPAssumpII} in addition to Assumption \ref{flui},
by a large deviation result of Bryc and Dembo \cite{bryc1996large}, we will have the large deviations
for $(Q_{n}^{b}(t),Q_{n}^{a}(t))$.

\begin{assumption}\label{LDPAssump}
Let $(X_{i})_{i\in\mathbb{N}}$ be a sequence of stationary $\mathbb{R}^{K}$-valued random vectors
with the $\sigma$-algebra $\mathcal{F}_{m}^{\ell}$ defined as $\sigma(X_{i},m\leq i\leq\ell)$.
For every $C<\infty$, there is a nondecreasing sequence $\ell(n)\in\mathbb{N}$ with
$\sum_{n=1}^{\infty}\frac{\ell(n)}{n(n+1)}<\infty$ such that
\begin{align*}
&\sup\left\{\mathbb{P}(A)\mathbb{P}(B)-e^{\ell(n)}\mathbb{P}(A\cap B) \mid
A\in\mathcal{F}^{k_{1}}_{0},B\in\mathcal{F}^{k_{1}+k_{2}+\ell(n)}_{k_{1}+\ell(n)},k_{1},k_{2}\in\mathbb{N}\right\}
\leq e^{-Cn},
\\
&\sup\left\{\mathbb{P}(A\cap B)-e^{\ell(n)}\mathbb{P}(A)\mathbb{P}(B) \mid
A\in\mathcal{F}^{k_{1}}_{0},B\in\mathcal{F}^{k_{1}+k_{2}+\ell(n)}_{k_{1}+\ell(n)},k_{1},k_{2}\in\mathbb{N}\right\}
\leq e^{-Cn}.
\end{align*}
\end{assumption}

Assumption \ref{LDPAssump} holds under the hypermixing condition  in \cite[Section 6.4]{dembo2009large},
under the $\psi$-mixing condition  of Bryc \cite[(1.10),(1.12)]{bryc1992large}, and
under the hyperexponential $\alpha$-mixing rate for stationary processes of Bryc and Dembo 
\cite[Proposition 2]{bryc1996large}.
Therefore, if $(\overrightarrow{V}_{i})_{i\in\mathbb{N}}$ and  $(N(i)-N(i-1))_{i\in\mathbb{N}}$
satisfy Assumption \ref{LDPAssump}, then Assumptions \ref{fluidd} and \ref{fluidv} are satisfied.
It is also clear that Assumption \ref{LDPAssump} holds if $X_{i}$'s are $m$-dependent.

In order to have the large deviations result, we also need to assume that $(\overrightarrow{V}_{i})_{i\in\mathbb{N}}$ and  $(N(i)-N(i-1))_{i\in\mathbb{N}}$ satisfy the following condition:

\begin{assumption}\label{LDPAssumpII}
For all $0\leq\gamma,R<\infty$,
\begin{equation*}
g_{R}(\gamma):=\sup_{k,m\in\mathbb{N},k\in[0,Rm]}\frac{1}{m}
\log\mathbb{E}\left[e^{\gamma\Vert\sum_{i=k+1}^{k+m}X_{i}\Vert}\right]<\infty,
\end{equation*}
and $A:=\sup_{\gamma}\limsup_{R\rightarrow\infty}R^{-1}g_{R}(\gamma)<\infty$.
\end{assumption}

Note Assumption \ref{LDPAssumpII} is trivially satisfied if $X_{i}$'s are bounded.
If $X_{i}$'s are i.i.d. random variables,  Assumption \ref{LDPAssumpII} reduces to the finiteness of  the moment generating function of  $X_{i}$,
which is a standard assumption for Mogulskii's theorem (\cite[Theorem 5.1.2]{dembo2009large}).
Therefore, Assumption \ref{LDPAssumpII} is a natural assumption for large deviations.

Under Assumption \ref{LDPAssump} and Assumption \ref{LDPAssumpII}, Dembo and Zajic \cite{dembo1995large} proved a sample path large deviation principle
for $\mathbb{P}(\frac{1}{n}\sum_{i=1}^{\lfloor\cdot n\rfloor}X_{i}\in\cdot)$
 (For ease of reference, we list it in Appendix~\ref{AppendixB} as Theorem \ref{Dembo}).
 From this, we can show the following.
\begin{lemma}\label{LDPLemma}
Let both $(\overrightarrow{V}_{i})_{i\in\mathbb{N}}$ and  $(N(i)-N(i-1))_{i\in\mathbb{N}}$ satisfy Assumption \ref{LDPAssump} and Assumption \ref{LDPAssumpII} and let Assumption \ref{flui} hold.
Then, for any $T>0$, $\mathbb{P}(C_{n}(t)\in\cdot)$ satisfies a large deviation principle on $L_{\infty}[0,T]$
with the good rate function
\begin{equation}\label{RateFunctionII}
\mathcal{I}(f)=\inf_{\substack{
h\in\mathcal{AC}_{0}^{+}[0,T],g\in\mathcal{AC}_{0}[0,\infty)
\\
g(h(t))=f(t),0\leq t\leq T
}}[I_{V}(g)+I_{N}(h)],
\end{equation}
with the convention that $\inf_{\emptyset}=\infty$ and
\begin{equation*}
I_{V}(g)=\int_{0}^{\infty}\Lambda_{V}(g'(x))dx,
\end{equation*}
if $g\in\mathcal{AC}_{0}^{+}[0,\infty)$ and $I_{V}(g)=\infty$ otherwise, where
\begin{equation}
\label{LambdaV}
\Lambda_{V}(x):=\sup_{\theta\in\mathbb{R}^{6}}\left\{\theta\cdot x-\Gamma_{V}(\theta)\right\},
\quad
\Gamma_{V}(\theta):=\lim_{n\rightarrow\infty}\frac{1}{n}\log\mathbb{E}\left[e^{\sum_{i=1}^{n}\theta\cdot\overrightarrow{V}_{i}}\right],
\end{equation}
and
\begin{equation*}
I_{N}(h)=\int_{0}^{T}\Lambda_{N}(h'(x))dx,
\end{equation*}
if $h\in\mathcal{AC}_{0}^{+}[0,T]$ and $I_{N}(h)=\infty$ otherwise, where
\begin{equation} \label{LambdaN}
\Lambda_{N}(x):=\sup_{\theta\in\mathbb{R}^{6}}\left\{\theta\cdot x-\Gamma_{N}(\theta)\right\},
\quad
\Gamma_{N}(\theta):=\lim_{n\rightarrow\infty}\frac{1}{n}\log\mathbb{E}\left[e^{\theta N_{n}}\right].
\end{equation}
\end{lemma}

\begin{proof}
Under Assumption \ref{LDPAssump} and Assumption \ref{LDPAssumpII}, by Theorem \ref{Dembo} in Appendix~\ref{AppendixB},
$\mathbb{P}(\frac{1}{n}\sum_{i=1}^{\lfloor\cdot n\rfloor}\overrightarrow{V}_{i}\in\cdot)$
satisfies a large deviation principle on $L_{\infty}[0,M]$ with the good rate function
\begin{equation*}
I_{V}(f)=\int_{0}^{M}\Lambda_{V}(f'(x))dx,
\end{equation*}
if $f\in\mathcal{AC}_{0}^{+}[0,M]$ and $I_{V}(f)=\infty$ otherwise, where
$\Lambda_{V}(x)$ and $\Gamma_{V}(\theta)$ are given by Eqn.~\eqref{LambdaV}
and $\mathbb{P}(\frac{1}{n}N(n\cdot)\in\cdot)$
satisfies a large deviation principle on $L_{\infty}[0,T]$ with the good
rate function
\begin{equation*}
I_{N}(f)=\int_{0}^{T}\Lambda_{N}(f'(x))dx,
\end{equation*}
if $f\in\mathcal{AC}_{0}^{+}[0,T]$ and $I_{N}(f)=\infty$ otherwise, where
$\Lambda_{N}(x)$ and $\Gamma_{N}(\theta)$ are given by Eqn.~\eqref{LambdaN}.
Since $(\overrightarrow{V}_{i})_{i\in\mathbb{N}}$ and $N_{t}$ are independent,
$\mathbb{P}(\frac{1}{n}\sum_{i=1}^{\lfloor\cdot n\rfloor}\overrightarrow{V}_{i}\in\cdot,
\frac{1}{n}N(n\cdot)\in\cdot)$ satisfies a large deviation principle
on $L_{\infty}[0,M]\times L_{\infty}[0,T]$ with the good rate function $I_{V}(\cdot)+I_{N}(\cdot)$.

We claim that the following superexponential estimate holds:
\begin{equation}\label{SuperExp}
\limsup_{M\rightarrow\infty}\limsup_{n\rightarrow\infty}
\frac{1}{n}\log\mathbb{P}\left(N(n)\geq nM\right)=-\infty.
\end{equation}
Indeed,
for any $\gamma>0$, by Chebychev's inequality,
\begin{equation*}
\mathbb{P}\left(N(n)\geq nM\right)
\leq e^{-\gamma n}\mathbb{E}\left[e^{\gamma N(n)}\right].
\end{equation*}
Therefore,
\begin{equation}\label{gammainfty}
\limsup_{n\rightarrow\infty}\frac{1}{n}\log
\mathbb{P}\left(N(n)\geq nM\right)
\leq-\gamma+\limsup_{n\rightarrow\infty}\frac{1}{n}\log\mathbb{E}\left[e^{\gamma N(n)}\right].
\end{equation}
From Assumption \ref{LDPAssumpII},
$\sup_{\gamma>0}\limsup_{n\rightarrow\infty}\frac{1}{n}\log\mathbb{E}\left[e^{\gamma N(n)}\right]<\infty$.
Hence, by letting $\gamma\rightarrow\infty$ in
Eqn.~\eqref{gammainfty}, we have
Eqn.~\eqref{SuperExp}.

For any closed set $C\in L_{\infty}[0,T]$,
\begin{align}
&\limsup_{n\rightarrow\infty}\frac{1}{n}\log\mathbb{P}
\left(\frac{1}{n}\sum_{i=1}^{N(n\cdot)}\overrightarrow{V}_{i}\in C\right)
\nonumber
\\
&=\limsup_{M\rightarrow\infty}\limsup_{n\rightarrow\infty}\frac{1}{n}\log\mathbb{P}
\left(\frac{1}{n}\sum_{i=1}^{N(n\cdot)}\overrightarrow{V}_{i}\in C,\frac{1}{n}N(nT)\leq M\right)
\label{FirstEq}
\\
&=-\inf_{M\in\mathbb{N}}\inf_{\substack{f\in C
\\
h\in\mathcal{AC}_{0}^{+}[0,T],g\in\mathcal{AC}_{0}[0,M]
\\
g(h(t))=f(t),0\leq t\leq T
\\
h(T)\leq M}}[I_{V}(g)+I_{N}(h)]
\label{SecondEq}
\\
&=-\inf_{f\in C}\inf_{\substack{
h\in\mathcal{AC}_{0}^{+}[0,T],g\in\mathcal{AC}_{0}[0,\infty)
\nonumber
\\
g(h(t))=f(t),0\leq t\leq T
}}[I_{V}(g)+I_{N}(h)],
\end{align}
where
Eqn.~\eqref{FirstEq} follows from
Eqn.~\eqref{SuperExp} and
Eqn.~\eqref{SecondEq}
follows from the contraction principle.
The contraction principle applies here
since for $h(t)=\frac{1}{n}N(nt)$ and $g(t)=\frac{1}{n}\sum_{i=1}^{\lfloor nt\rfloor}\overrightarrow{V}_{i}$
we have $\frac{1}{n}\sum_{i=1}^{N(nt)}\overrightarrow{V}_{i}=g(h(t))$
and moreover, the map $(g,h)\mapsto g\circ h$ is continuous
since for any two functions $F_{n},G_{n}\rightarrow F,G$ in uniform topology and that are absolutely continuous, we have
$\sup_{t}|F_{n}(G_{n}(t))-F(G(t))|
\leq\sup_{t}|F_{n}(G_{n}(t))-F(G_{n}(t))|+\sup_{t}|F(G_{n}(t))-F(G(t))|\rightarrow 0$
as $n\rightarrow\infty$.

For any open set $G\in L_{\infty}[0,T]$,
\begin{align}
&\liminf_{n\rightarrow\infty}\frac{1}{n}\log\mathbb{P}
\left(\frac{1}{n}\sum_{i=1}^{N(n\cdot)}\overrightarrow{V}_{i}^{j}\in G\right)
\nonumber
\\
&\geq
\liminf_{n\rightarrow\infty}\frac{1}{n}\log\mathbb{P}
\left(\frac{1}{n}\sum_{i=1}^{N(n\cdot)}\overrightarrow{V}_{i}\in G,\frac{1}{n}N(nT)\leq M\right)
\nonumber
\\
&=-\inf_{\substack{f\in G
\\
h\in\mathcal{AC}_{0}^{+}[0,T],g\in\mathcal{AC}_{0}[0,M]
\\
g(h(t))=f(t),0\leq t\leq T
\\
h(T)\leq M}}[I_{V}(g)+I_{N}(h)].
\nonumber
\end{align}
Since it holds for any $M\in\mathbb{N}$, the lower bound is proved.
\end{proof}

%We have proved in Lemma \ref{LDPLemma} that
%$\mathbb{P}(\overrightarrow{\mathbf{C}}_{n}(t)\in\cdot)$ satisfies a large deviation principle %on $L^{\infty}[0,\infty)$
%with rate function
%\begin{equation*}
%\mathcal{I}(\phi)=\inf_{\substack{
%h\in\mathcal{AC}_{0}^{+}[0,\infty),g\in\mathcal{AC}_{0}[0,\infty)
%\\
%g(h(t))=\phi(t)
%}}[I_{V}(g)+I_{N}(h)],
%\end{equation*}
%for any $\phi\in\mathcal{AC}_{0}[0,\infty)$, the space of absolutely continuous functions %starting at $0$ and
%$I(\phi)=+\infty$ otherwise.

Moreover, by the contraction principle,
\begin{theorem}\label{LDPThm}
Under the same assumptions as in Lemma \ref{LDPLemma},
$\mathbb{P}(({Q}_{n}^{b}(t),{Q}_{n}^{a}(t))\in\cdot)$ satisfies
a large deviation principle on $L^{\infty}[0,\infty)$ with the rate function
\begin{equation*}
I(f^{b},f^{a})=\inf_{\phi\in\mathcal{G}_{f}}\mathcal{I}(\phi),
\end{equation*}
where $\mathcal{I}(\cdot)$ is defined in Lemma \ref{LDPLemma},
$\mathcal{G}_{f}$ is the set consists of absolutely continuous functions $\phi(t)$ starting at $0$ that satisfy
\begin{equation*}
d(f^b(t), f^a(t))^T=\left(
\begin{array}{cccccc}
1 &-1 &-1 &0 &0 &0 \\
0 &0 &0 &1 &-1 &-1
\end{array}
\right)d\phi(t),
\end{equation*}
with the initial condition $(f^b(0), f^a(0))=(q^b, q^a)$. Otherwise $I(f)=\infty$.
\end{theorem}

\begin{proof}
Since $\mathbb{P}(\overrightarrow{\mathbf{C}}_{n}(t)\in\cdot)$ satisfies a large deviation principle on $L^{\infty}[0,\infty)$ with
the rate function $\mathcal{I}(\phi)$, it follows that
$\mathbb{P}((Q_{n}^{b}(t),Q_{n}^{a}(t))\in\cdot)$ satisfies
a large deviation principle on $L^{\infty}[0,\infty)$ with the rate function
\begin{equation*}
I(f):=I(f^{b},f^{a})=\inf_{\phi\in\mathcal{G}_{f}}\mathcal{I}(\phi),
\end{equation*}
where $\mathcal{G}_{f}$ is the set of absolutely continuous
functions $\phi(t)=(\phi^{j}(t),1\leq j\leq 6)$
starting at $0$ that satisfy
\begin{equation*}
d\left(\begin{aligned}
&f^b(t)\\
&f^a(t)\end{aligned}\right)
=\left(
\begin{array}{cccccc}
1 &-1 &-1 &0 &0 &0 \\
0 &0 &0 &1 &-1 &-1
\end{array}
\right)d\phi(t),
\end{equation*}
with the initial condition $(f^b(0), f^a(0))=(q^b, q^a)$.
It is clear that
\begin{align*}
&f^{b}(t)=q^{b}+\phi^{1}(t)-\phi^{2}(t)-\phi^{3}(t),
\\
&f^{a}(t)=q^{a}+\phi^{4}(t)-\phi^{5}(t)-\phi^{6}(t),
\nonumber
\end{align*}
and the mapping $\phi\mapsto(f^{b},f^{a})$ is continuous, since it is easy to check that if
\begin{equation*}
\phi_{n}(t):=(\phi^{1}_{n}(t),\ldots,\phi^{6}_{n}(t))\rightarrow
\phi(t)=(\phi^{1}(t),\ldots,\phi^{6}(t))
\end{equation*}
in the $L^{\infty}$ norm, then $(f_{n}^{b}(t),f_{n}^{a}(t))\rightarrow(f^{b}(t),f^{a}(t))$ in the $L^{\infty}$ norm.
Since the mapping $\phi\mapsto(f^{b},f^{a})$ is continuous, the large deviation principle follows
from the contraction principle.
\end{proof}
Let us now consider a special case of Theorem \ref{LDPThm}:
\begin{corollary}
Assume that $N(t)$ is a standard Poisson process with intensity $\lambda$ independent of
the i.i.d. random vectors $\overrightarrow{V}_{i}$ in $\mathbb{R}^{6}$ such that $\mathbb{E}[e^{\theta\cdot\overrightarrow{V}_{1}}]<\infty$
for any $\theta\in\mathbb{R}^{6}$. Then, the rate function $I(f)$ in
Eqn.~\eqref{RateFunctionII} in Lemma \ref{LDPLemma} has an alternative expression
\begin{equation}\label{RateFunctionI}
\mathcal{I}(f)=\int_{0}^{\infty}\Lambda(f'(t))dt,
\end{equation}
for any $f\in\mathcal{AC}_{0}[0,\infty)$, the space of absolutely continuous functions starting at $0$ and
$I(\phi)=+\infty$ otherwise, where
\begin{equation*}
\Lambda(x):=\sup_{\theta\in\mathbb{R}^{6}}\left\{\theta\cdot x-\lambda(\mathbb{E}[e^{\theta\cdot\overrightarrow{V}_{1}}]-1)\right\}.
\end{equation*}
\end{corollary}
\begin{proof}
First, notice that when $N_{t}$ is a standard Poisson process with intensity $\lambda$,
independent of i.i.d. random vectors $\overrightarrow{V}_{i}$
then, $N(i)-N(i-1)$ is a sequence of i.i.d. Poisson random variables with parameter $\lambda$
and therefore both $(\overrightarrow{V}_{i})_{i\in\mathbb{N}}$ and  $(N(i)-N(i-1))_{i\in\mathbb{N}}$ satisfy Assumption \ref{LDPAssump}
and Assumption \ref{flui} is also satisfied. Under the assumption, $\mathbb{E}[e^{\theta\cdot\overrightarrow{V}_{1}}]<\infty$
for any $\theta\in\mathbb{R}^{6}$ and moreover, $\mathbb{E}[e^{\theta(N(i)-N(i-1))}]=e^{\lambda(e^{\theta}-1)}<\infty$
for any $\theta\in\mathbb{R}$. Therefore, both $(\overrightarrow{V}_{i})_{i\in\mathbb{N}}$ and  $(N(i)-N(i-1))_{i\in\mathbb{N}}$
satisfy Assumption \ref{LDPAssumpII}.

By Lemma \ref{LDPLemma},
\begin{equation*}
I_{V}(g)+I_{N}(h)
=\int_{0}^{T}\Lambda_{V}(g'(t))dt
+\int_{0}^{\infty}\Lambda_{N}(h'(t))dt,
\end{equation*}
where
\begin{equation*}
\Lambda_{V}(x)=\sup_{\theta\in\mathbb{R}^{6}}\left\{\theta\cdot x-\log\mathbb{E}\left[e^{\theta\cdot\overrightarrow{V}_{1}}\right]\right\},
\end{equation*}
and
\begin{equation*}
\Lambda_{N}(x)=x\log\left(\frac{x}{\lambda}\right)-x+\lambda.
\end{equation*}
Since $f(t)=g(h(t))$, we have $f'(t)=g'(h(t))h'(t)$ and
\begin{equation*}
\int_{0}^{\infty}\Lambda_{V}(g'(t))dt
=\int_{0}^{T}\Lambda_{V}(g'(h(t))h'(t)dt
=\int_{0}^{T}\Lambda_{V}\left(\frac{f'(t)}{h'(t)}\right)h'(t)dt.
\end{equation*}
Therefore,
\begin{align*}
&\inf_{\substack{
h\in\mathcal{AC}_{0}^{+}[0,T],g\in\mathcal{AC}_{0}[0,\infty)
\\
g(h(t))=f(t),0\leq t\leq T
}}(I_{V}(g)+I_{N}(h))
\\
&=\inf_{h\in\mathcal{AC}_{0}^{+}[0,T]}
\int_{0}^{T}\left(\Lambda_{V}\left(\frac{f'(t)}{h'(t)}\right)h'(t)
+h'(t)\log\left(\frac{h'(t)}{\lambda}\right)-h'(t)+\lambda\right)dt.
\nonumber
\end{align*}
Now,
\begin{align*}
&\inf_{y}\left\{\Lambda_{V}\left(\frac{x}{y}\right)y
+y\log\left(\frac{y}{\lambda}\right)-y+\lambda\right\}
\\
&=\inf_{y}\sup_{\theta}\left\{\theta\cdot x-y\log\mathbb{E}[e^{\theta\cdot\overrightarrow{V}_{1}}]
+y\log\left(\frac{y}{\lambda}\right)-y+\lambda\right\}
\nonumber
\\
&=\sup_{\theta}\inf_{y}\left\{\theta\cdot x-y\log\mathbb{E}[e^{\theta\cdot\overrightarrow{V}_{1}}]
+y\log\left(\frac{y}{\lambda}\right)-y+\lambda\right\}
\nonumber
\\
&=\sup_{\theta}\left\{\theta\cdot x-\lambda(\mathbb{E}[e^{\theta\cdot\overrightarrow{V}_{1}}]-1)\right\}.
\nonumber
\end{align*}
Therefore,
Eqn.~\eqref{RateFunctionII} reduces to
Eqn.~\eqref{RateFunctionI}.
%Indeed, one can obtain
%\eqref{RateFunctionI} directly as follows. Under our assumptions,
%$\overrightarrow{\mathbf{C}}_{n}(t)$ has independent time increments and
%thus we can apply the Mogulskii's theorem \ref{Mogulskii} in
%Appendix~\ref{AppendixB},
%to obtain a sample path large deviation principle for %$\mathbb{P}(\overrightarrow{\mathbf{C}}_{n}(t)\in\cdot)$.
%For any $\theta\in\mathbb{R}^{6}$, we can compute that
%\begin{equation}
%\mathbb{E}\left[e^{\theta\cdot\sum_{i=1}^{N(nt)}\overrightarrow{V}_{i}}\right]
%=\exp\left\{\lambda nt(\mathbb{E}[e^{\theta\cdot\overrightarrow{V}_{1}}]-1)\right\}.
%\nonumber
%\end{equation}
%Therefore, for any $\theta\in\mathbb{R}^{6}$,
%\begin{equation}
%\Gamma(\theta):=\lim_{n\rightarrow\infty}\frac{1}{n}\log\mathbb{E}[e^{\theta\cdot\overrightarrow{\mathbf{C}}_{n}(1)}]
%=\lambda(\mathbb{E}[e^{\theta\cdot\overrightarrow{V}_{1}}]-1),
%\end{equation}
%and
%\begin{align}
%\Lambda(x)&:=\sup_{\theta\in\mathbb{R}^{6}}\{\theta\cdot x-\Gamma(\theta)\}
%\\
%&=\sup_{\theta\in\mathbb{R}^{6}}\left\{\theta\cdot %x-\lambda(\mathbb{E}[e^{\theta\cdot\overrightarrow{V}_{1}}]-1)\right\}.
%\nonumber
%\end{align}
%By Mogulskii's theorem \ref{Mogulskii}, %$\mathbb{P}(\overrightarrow{\mathbf{C}}_{n}(t)\in\cdot)$ satisfies a large deviation principle %on $L^{\infty}[0,\infty)$
%with rate function
%\begin{equation}
%\mathcal{I}(f)=\int_{0}^{\infty}\Lambda(f'(t))dt,
%\end{equation}
%for any $f\in\mathcal{AC}_{0}[0,\infty)$, the space of absolutely continuous functions %starting at $0$ and
%$I(f)=+\infty$ otherwise.
\end{proof}

\section{Applications to LOB}
\label{Sec application}

\subsection{Examples} Having established the fluid limit and the fluctuations of the queue lengths and order positions, we will  give some examples
of the order arrival process $N(t)$ that satisfy the assumptions in our analysis.

\begin{example}[Poisson process]
Let $N(t)$ be a Poisson process with intensity $\lambda$. %Then $\{N(i,i+1]\}_{i\in\mathbb{Z}}$ %are i.i.d.
%distributed Poisson random variables with parameter $\lambda$. Obviously, $\frac{N(t)}{t}$ %converges to $\lambda$ a.s. and $v_{d}^{2}=\lambda$. Therefore
Clearly assumptions \ref{fluidd} and \ref{assumptionD} are satisfied.
\end{example}

\begin{example}[Hawkes process]
Let $N(t)$ be a Hawkes process \cite{bremaud1996stability}, i.e.,
a simple point process with intensity
\begin{equation}\label{Hdynamics}
\lambda(t):=\lambda\left(\int_{-\infty}^{t}h(t-s)N(ds)\right),
\end{equation}
at time $t$, where we assume that
$\lambda:\mathbb{R}_{\geq 0}\rightarrow\mathbb{R}^{+}$ is an increasing function,
$\alpha$-Lipschitz, where $\alpha\Vert h\Vert_{L^{1}}<1$ and
$h:\mathbb{R}_{\geq 0}\rightarrow\mathbb{R}^{+}$ is a decreasing function and
$\int_{0}^{\infty}h(t)tdt<\infty$. Under these assumptions, there exists
a stationary and ergodic Hawkes process satisfying the dynamics
Eqn.~\eqref{Hdynamics} (see e.g., Br\'{e}maud and Massouli\'{e} \cite{bremaud1996stability}).
By the Ergodic theorem,
\begin{equation*}
\frac{N(t)}{t}\rightarrow\lambda:=\mathbb{E}[N(0,1]],
\end{equation*}
a.s. as $t\rightarrow\infty$.
Therefore, Assumption
\ref{fluidd} is satisfied. It was proved in
Zhu \cite{zhu2013central}, that
$\{N(i,i+1]\}_{i\in\mathbb{Z}}$ satisfies Assumption \ref{assumptionD} and
hence $\frac{N_{ n \cdot }-\lambda n\cdot }{\sqrt{n}}\Rightarrow v_d W_1(\cdot)$,
on $(D[0,T],J_{1})$ as $n\rightarrow\infty$.

In the special case $\lambda(z)=\nu+z$, Eqn.~\eqref{Hdynamics} becomes
\begin{equation*}
\lambda(t)=\nu+\int_{-\infty}^{t}h(t-s)N(ds),
\end{equation*}
which is the original self-exciting point process proposed by Hawkes \cite{hawkes1971spectra}, where $\nu>0$
and $\Vert h\Vert_{L^{1}}<1$. In this case,
\begin{equation*}
\lambda=\frac{\nu}{1-\Vert h\Vert_{L^{1}}},
\qquad
v_d^{2}=\frac{\nu}{(1-\Vert h\Vert_{L^{1}})^{3}}.
\end{equation*}
\end{example}

\begin{example}[Cox process with shot noise intensity]
Let $N(t)$ be a Cox process with shot noise intensity
(see for example \cite{asmussen2010ruin}).
That is, $N(t)$ is a simple point process with intensity at time $t$ given by
\begin{equation*}
\lambda(t)=\nu+\int_{-\infty}^{t}g(t-s)\bar{N}(ds),
\end{equation*}
where $\bar{N}$ is a Poisson process with intensity $\rho$,
$g:\mathbb{R}_{\geq 0}\rightarrow\mathbb{R}^{+}$
is decreasing,
$\Vert g\Vert_{L^{1}}<\infty$, and $\int_{0}^{\infty}tg(t)dt<\infty$.
$N(t)$ is stationary and ergodic and
\begin{equation*}
\frac{N(t)}{t}\rightarrow\lambda:=\nu+\rho\Vert g\Vert_{L^{1}} \mbox{ a.s.},
\end{equation*}
as $t\rightarrow\infty$.
Therefore,  Assumption \ref{fluidd} is satisfied.
Moreover  one can check that  condition
Eqn.~\eqref{L2Condition} in Assumption \ref{assumptionD} is satisfied. Indeed, by stationarity,
\begin{equation*}
\Vert\mathbb{E}[N(0,1]-\lambda|\mathcal{F}^{-\infty}_{-n}]\Vert_{2}
=\Vert\mathbb{E}[N(n-1,n]-\lambda|\mathcal{F}^{-\infty}_{0}]\Vert_{2}.
\end{equation*}
We have
\begin{equation*}
\mathbb{E}[N(n-1,n]-\lambda \mid \mathcal{F}^{-\infty}_{0}]
=\mathbb{E}\left[\int_{n-1}^{n}\lambda(t)dt-\lambda\bigg|\mathcal{F}^{-\infty}_{0}\right],
\end{equation*}
where
\begin{equation*}
\lambda(t)=\nu+\int_{-\infty}^{0}g(t-s)\bar{N}(ds)
+\int_{0}^{t}g(t-s)\bar{N}(ds),
\end{equation*}
therefore,
\begin{equation*}
\mathbb{E}[N(n-1,n]-\lambda \mid \mathcal{F}^{-\infty}_{0}]
=\int_{n-1}^{n}\int_{-\infty}^{0}g(t-s)\bar{N}(ds)dt
+\rho\int_{n-1}^{n}\int_{0}^{t}g(t-s)dsdt-\rho\Vert g\Vert_{L^{1}}.
\end{equation*}
By Minkowski's inequality,
\begin{align*}
\Vert\mathbb{E}[N(n-1,n]-\lambda \mid \mathcal{F}^{-\infty}_{0}]\Vert_{2}
\leq\bigg\Vert\int_{n-1}^{n}\int_{-\infty}^{0}g(t-s)\bar{N}(ds)dt\bigg\Vert_{2}
%\\ &\qquad\qquad
+\bigg\Vert\rho\int_{n-1}^{n}\int_{0}^{t}g(t-s)dsdt-\rho\Vert g\Vert_{L^{1}}\bigg\Vert_{2}.
\nonumber
\end{align*}
Note that
\begin{equation*}
\bigg\Vert\rho\int_{n-1}^{n}\int_{0}^{t}g(t-s)dsdt-\rho\Vert g\Vert_{L^{1}}\bigg\Vert_{2}
=\rho\int_{n-1}^{n}\int_{t}^{\infty}g(s)dsdt,
\end{equation*}
therefore,
\begin{equation*}
\sum_{n=1}^{\infty}\bigg\Vert\rho\int_{n-1}^{n}\int_{0}^{t}g(t-s)dsdt-\rho\Vert g\Vert_{L^{1}}\bigg\Vert_{2}
=\int_{0}^{\infty}\int_{t}^{\infty}g(s)dsdt
=\int_{0}^{\infty}tg(t)dt.
\end{equation*}
Furthermore,
\begin{align*}
\sum_{n=1}^{\infty}\bigg\Vert\int_{n-1}^{n}\int_{-\infty}^{0}g(t-s)\bar{N}(ds)dt\bigg\Vert_{2}
&\leq\sum_{n=1}^{\infty}\bigg\Vert\int_{-\infty}^{0}g(n-1-s)\bar{N}(ds)\bigg\Vert_{2}
\\
&=\sum_{n=1}^{\infty}\sqrt{\int_{-\infty}^{0}g^{2}(n-1-s)\rho ds
+\rho^{2}\left(\int_{-\infty}^{0}g(n-1-s)ds\right)^{2}}
\nonumber
\\
&\leq\sum_{n=1}^{\infty}\sqrt{\int_{-\infty}^{0}g^{2}(n-1-s)\rho ds}
+\sum_{n=1}^{\infty}\rho\int_{-\infty}^{0}g(n-1-s)ds
\nonumber
\\
&\leq\sqrt{\rho}\sum_{n=1}^{\infty}\sqrt{g(n-1)}\sqrt{\int_{-\infty}^{0}g(n-1-s)ds}
+\rho\int_{0}^{\infty}tg(t)dt
\nonumber
\\
&\leq
\frac{\sqrt{\rho}}{4}\left[\sum_{n=1}^{\infty}g(n-1)+\sum_{n=1}^{\infty}\int_{-\infty}^{0}g(n-1-s)ds\right]
+\rho\int_{0}^{\infty}tg(t)dt
\nonumber
\\
&\leq
\frac{\sqrt{\rho}}{4}\left[g(0)+\Vert g\Vert_{L^{1}}+\int_{0}^{\infty}tg(t)dt\right]
+\rho\int_{0}^{\infty}tg(t)dt<\infty.
\nonumber
\end{align*}
Hence Assumption \ref{assumptionD} is satisfied. $\frac{N_{ n \cdot }-\lambda n\cdot }{\sqrt{n}}\Rightarrow v_d W_1(\cdot)$
in $(D[0,T],J_{1})$ as $n\rightarrow\infty$, where
\begin{equation*}
v_d^{2}=\nu+\rho\Vert g\Vert_{L^{1}}+\rho\Vert g^{2}\Vert_{L^{1}}.
\end{equation*}
\end{example}
\subsection{Probability of price increase and hitting times}
Given the diffusion limit to the queue lengths for the best bid and ask, we can also compute
 the distribution of the first hitting time $\iota$ (defined in
 Eqn.~\eqref{tau}) and the probability of price increase/decrease.
Our results generalize those in \cite{cont2013price} which correspond to the special case of zero drift.

Given Theorem \ref{DiffThm}, let us first parameterize $\sigma$ by
\begin{align*}
\sigma = \left(\begin{array}{cc}
\sigma_1\sqrt{1-\rho^2} &\sigma_1\rho\\
0 &\sigma_2
\end{array}
\right),
\end{align*}
and assume that $-1<\rho<1$.
Next, denote
$I_\nu$  the modified Bessel function of the first kind of order
$\nu$ and $\nu_n := n\pi/\alpha$, and define
\begin{align*}
&\alpha :=
\begin{cases}
\pi+\tan^{-1}\left(-\frac{\sqrt{1-\rho^2}}{\rho}\right) &\qquad\rho>0, \\
\frac{\pi}{2} &\qquad\rho=0, \\
\tan^{-1}\left(-\frac{\sqrt{1-\rho^2}}{\rho}\right) &\qquad\rho<0,
\end{cases}
\\
&
r_0 :=\sqrt{\frac{(q^b/\sigma_1)^2+(q^a/\sigma_2)^2-2\rho (q^b/\sigma_1)( q^a/\sigma_2)}{1-\rho^2}},
\\
&\theta_0 :=
\begin{cases}
\pi+\tan^{-1}\left(\frac{q^a/\sigma_2\sqrt{1-\rho^2}}{q^b/\sigma_1-\rho q^a/\sigma_2}\right) &\qquad q^b/\sigma_1<\rho q^a/\sigma_2,\\
\frac{\pi}{2} &\qquad q^b/\sigma_1=\rho q^a/\sigma_2,\\
\tan^{-1}\left(\frac{q^a/\sigma_2\sqrt{1-\rho^2}}{q^b/\sigma_1-\rho q^a/\sigma_2}\right) &\qquad q^b/\sigma_1>\rho q^a/\sigma_2.
\end{cases}
\end{align*}
Then according to Zhou~\cite{zhou2001analysis}, we have
\begin{corollary}
Given Theorem \ref{DiffThm} and the initial state $(q^b, q^a)$, the distribution of the first hitting time $\iota$
\begin{equation*}
\mathbb{P}_{\overrightarrow{\mu}}(\iota>t) = \frac{2}{\alpha t} e^{l_1 q^b+l_2 q^a+l_3 t}\sum\limits_{n=1}^{\infty} \sin\left(\frac{n\pi \theta_0}{\alpha}\right)
e^{-\frac{r_0^2}{2t}} \int_{0}^\alpha \sin\left(\frac{n\pi\theta}{\alpha}\right)g_n(\theta)d\theta,
\end{equation*}
where
\begin{align*}
g_n(\theta) &:= \int_0^\infty re^{-\frac{r^2}{2t}}e^{l_4r\sin(\theta-\alpha)-l_5r\cos(\theta-\alpha)}I_{\frac{n\pi}{\alpha}}\left(\frac{rr_0}{t}\right)dr,\\
l_1 &:= \frac{-\mu_1\sigma_2+\rho\mu_2\sigma_1}{(1-\rho^2)\sigma_1^2\sigma_2},
\enspace
l_2 := \frac{\rho\mu_1\sigma_2-\mu_2\sigma_1}{(1-\rho^2)\sigma_2^2\sigma_1},
\enspace
l_3 := \frac{l_1^2\sigma_1^2}{2}+\rho l_1 l_2 \sigma_1\sigma_2+\frac{l_2^2\sigma_2^2}{2}+l_1\mu_1+l_2\mu_2, \\
l_4 &:= l_1\sigma_1+\rho l_2\sigma_2,
\enspace
l_5 := l_2\sigma_2\sqrt{1-\rho^2}.
\end{align*}
\end{corollary}
Note that when $\overrightarrow{\mu}>0$, it is possible to have $\mathbb{P}_{\overrightarrow{\mu}}(\iota = \infty) > 0$, meaning the measure $\mathbb{P}_{\overrightarrow{\mu}}$  might be a sub-probability measure, depending on the value of $\overrightarrow{\mu}$. In this case, $\mathbb{P}_{\overrightarrow{\mu}}(\iota>t)$ actually includes $\mathbb{P}_{\overrightarrow{\mu}}(\iota = \infty)$.

Moreover, based on the results in Iyengar~\cite{iyengar1985hitting}
and Metzler~\cite{metzler2010first},
\begin{corollary}
 Given Theorem \ref{DiffThm} and the initial state $(q^b, q^a)$,  the probability of price decrease is given by
\begin{equation*}
\mathbb{P}_{\overrightarrow{\mu}}(\iota^b < \iota^a) = \int_0^\infty \int_0^\infty
\exp(\kappa^{b}(r\cos\alpha-z^{b})+\kappa^{a}(r\sin\alpha-z^{a})-|\overrightarrow{\kappa}|^2t/2){g}(t, r) dr dt,
\end{equation*}
where
\begin{equation*}
{g}(t, r)= \frac{\pi}{\alpha^2tr}e^{-(r^2+r_0^2)/2t}\sum\limits_{n=1}^{\infty} n\sin\left(\frac{n\pi(\alpha-\theta_0)}{\alpha}\right)I_{n\pi/\alpha}\left(\frac{r r_0}{t}\right),
\end{equation*}
and $\overrightarrow{\kappa}=(\kappa^{b},\kappa^{a})^{T}=\sigma^{-1}(\mu_{1},\mu_{2})^{T}$
and $(z^{b},z^{a})=\sigma^{-1}(q^{b},q^{a})^{T}$. That is,
\begin{equation*}
\left(
\begin{array}{c}
\kappa^{b}
\\
\kappa^{a}
\end{array}
\right)
=
\left(
\begin{array}{c}
\sigma_{2}\mu_{1}
\\
-\sigma_{1}\rho\mu_{1}+\sigma_{1}\sqrt{1-\rho^{2}}\mu_{2}
\end{array}
\right),
\qquad
\left(
\begin{array}{c}
z^{b}
\\
z^{a}
\end{array}
\right)
=
\left(
\begin{array}{c}
\sigma_{2}q^{b}
\\
-\sigma_{1}\rho q^{b}+\sigma_{1}\sqrt{1-\rho^{2}}q^{a}
\end{array}
\right).
\end{equation*}
\end{corollary}
Similarly, when $\overrightarrow{\mu} > 0$, with positive probability, we might have $\iota^b=\infty$ and $\iota^a=\infty$. Therefore $\mathbb{P}_{\overrightarrow{\mu}}(\iota^b < \iota^a)$ we compute here implicitly refers
to $\mathbb{P}_{\overrightarrow{\mu}}(\iota^b < \iota^a, \iota^b<\infty)$ in that case.
%\begin{proof}
%When $\overrightarrow{\mu} = \overrightarrow{0}$, according to \cite{metzler2010first}, we have
%\begin{align}
%\mathbb{P}(\iota \in dt, R(\iota) \in dr, \Theta(\iota)=\alpha) = \frac{\pi}{\alpha^2tr}e^{-(r^2+r_0^2)/2t}\sum\limits_{n=1}^{\infty} n\sin\big(\frac{n\pi(\pi-\theta_0)}{\alpha}\big)I_{n\pi/\alpha}\big(\frac{r r_0}{t}\big) drdt
%\end{align}
%Then for any $\overrightarrow{\mu}$,
%\begin{align}
%\mathbb{P}_{\overrightarrow{\mu}}(\iota \in dt, R(\iota) \in dr, \Theta(\iota)=\alpha) &= \exp(\mu_1(r\cos\alpha-a_1)+\mu_2(r\sin\alpha-a_2)-|\overrightarrow{\mu}|^2t/2)\mathbb{P}(\iota \in dt, R(\iota) \in dr, \Theta(\iota)=\alpha)
%\end{align}
%Then the probability of price decrease could be computed by
%\begin{align}
%\mathbb{P}_{\overrightarrow{\mu}}(\Theta(\iota)=\alpha) = \int_0^\infty \int_0^\infty \exp(\mu_1(r\cos\alpha-a_1)+\mu_2(r\sin\alpha-a_2)-|\overrightarrow{\mu}|^2t/2)\mathbb{P}(\iota \in dt, R(\iota) \in dr, \Theta(\iota)=\alpha) dr dt
%\end{align}
%\end{proof}

Note that both expressions for $\iota$ and the probability of price decrease are semi-analytic. However, in the special case of
$\overrightarrow{\mu} = \overrightarrow{0}$,
i.e., when $\bar{V}_1= \bar{V}_2+\bar{V}_3$ and $\bar{V}_4= \bar{V}_5+\bar{V}_6$, they become analytic.

\begin{corollary}
Given Theorem \ref{DiffThm} and the initial state $(q^b, q^a)$, if $\overrightarrow{\mu} = \overrightarrow{0}$, then
\begin{equation*}
\mathbb{P}(\iota> t) = \frac{2r_0}{\sqrt{2\pi t}}e^{-r_0^2/4t} \sum_{n:\text{ odd}}
\frac{1}{n} \sin\frac{n\pi \theta_0}{\alpha} \left( I_{(\nu_n-1)/2}(r_0^2/4t)+I_{(\nu_n+1)/2}(r_0^2/4t) \right).
\end{equation*}
\end{corollary}
\begin{corollary}
Given Theorem \ref{DiffThm} and the initial state $(q^b, q^a)$,
if $\overrightarrow{\mu} = \overrightarrow{0}$, then
the probability that the price decreases is
%$\mathbb{P}(\iota^b<\iota^a) =
$\frac{\theta_0}{\alpha}.$
\end{corollary}
\begin{proof}
%Thanks to Metzler~\cite{metzler2010first},
\begin{align*}
\mathbb{P}(\iota^b<\iota^a) &= \int_{0}^{\infty}
\frac{(r/r_0)^{(\pi/\alpha)-1}\sin(\pi \theta_0/\alpha)}{\sin^2(\pi\theta_0/\alpha)+((r/r_0)^{\pi/\alpha}+\cos(\pi\theta_0/\alpha))^2}\frac{dr}{\alpha r_0}\\
&=\int_{0}^{\infty}  \frac{\sin(\pi \theta_0/\alpha)}{\sin^2(\pi\theta_0/\alpha)+((r/r_0)^{\pi/\alpha}+\cos(\pi\theta_0/\alpha))^2}\frac{d (r/r_0)^{\pi/\alpha}}{\pi}\\
&=\int_{0}^{\infty} \frac{\sin(\pi \theta_0/\alpha)}{\sin^2(\pi\theta_0/\alpha)+(x+\cos(\pi\theta_0/\alpha))^2} \frac{dx}{\pi}
%&=\int_{0}^{\infty} \frac{\sin(\pi \theta_0/\alpha)}{x^2+2\cos(\pi\theta_0/\alpha)x+1} %\frac{dx}{\pi}
=\frac{\theta_0}{\alpha}.
\qedhere
\end{align*}
\end{proof}

\subsection{Fluctuations of execution and hitting times}
In addition, we can  study the fluctuations of the
execution time $\tau_{n}^{z}$.
\begin{prop}
Given Theorem \ref{Thm fluctuation}, for any $x$ (say, $x<0$),
\begin{equation*}
\lim_{n\rightarrow\infty}\mathbb{P}(\sqrt{n}(\tau_{n}^{z}-\tau^{z})\geq x)
=\mathbb{P}(Y(\tau^{z})>ax)=1-\Phi\left(\frac{ax}{\sigma_{Y}(\tau^{z})}\right),
\end{equation*}
where $\Phi(x):=\int_{-\infty}^{x}\frac{e^{-y^{2}/2}}{\sqrt{2\pi}}dy$ is the cumulative probability
distribution function of a standard Gaussian random variable, with $\sigma_Y$  the variance of the Gaussian process $Y(t)$.
\end{prop}
\begin{proof}
For any $x<0$,
\begin{align*}
&\mathbb{P}(\sqrt{n}(\tau_{n}^{z}-\tau^{z})\geq x)
\\
&=\mathbb{P}\left(Z_{n}\left(\tau^{z}+\frac{x}{\sqrt{n}}\right)>0\right)
\nonumber
\\
&=\mathbb{P}\left(\sqrt{n}\left(Z_{n}\left(\tau^{z}+\frac{x}{\sqrt{n}}\right)
-Z\left(\tau^{z}+\frac{x}{\sqrt{n}}\right)\right)>-\sqrt{n}Z\left(\tau^{z}+\frac{x}{\sqrt{n}}\right)\right).
\nonumber
\end{align*}
Note that
\begin{equation*}
\lim_{n\rightarrow\infty}\sqrt{n}Z\left(\tau^{z}+\frac{x}{\sqrt{n}}\right)
=xZ'(\tau^{z}),
\end{equation*}
and for any $t>0$, $c\neq -1$, Eqn. \eqref{Z} and
Eqn.~\eqref{tauz} lead to
%\begin{equation}
%Z'(t)=-\frac{ac}{1+c}-\left(z+\frac{ab}{1+c}\right)b^{\frac{1}{c}}(b+ct)^{-\frac{1+c}{c}},
%\end{equation}
%and $\tau^{z}=[\frac{(1+c)z}{a}+b]^{\frac{c}{c+1}}b^{\frac{1}{c+1}}c^{-1}-\frac{b}{c}$. %Therefore,
\begin{equation*}
Z'(\tau^{z})=-\frac{ac}{1+c}-\left(z+\frac{ab}{1+c}\right)b^{\frac{1}{c}}
\left(\left(\frac{(1+c)z}{a}+b\right)^{\frac{c}{c+1}}b^{\frac{1}{c+1}}\right)^{-\frac{1+c}{c}}=-a.
\end{equation*}
Similarly, when $c=-1$, we have
$Z(t)=(a\log(b-t)+\frac{z}{b}-a\log b)(b-t)$. Thus
$Z'(t)=-a-\left(a\log(b-t)+\frac{z}{b}-a\log b\right)$,
and
$Z'(\tau^{z})=-a-\left(a\log(be^{-\frac{z}{ab}})+\frac{z}{b}-a\log b\right)=-a$.
Finally, recall that $\sqrt{n}(Z_{n}(t)-Z(t))\rightarrow Y(t)$ on $(D[0,\tau^{z}),J_{1})$ as $n\rightarrow\infty$, hence the first equation.

The second equation follows from $Y(t)$ being  a
Gaussian process with zero mean and variance $\sigma_{Y}^{2}$, the latter of which can be computed explicitly, albeit in a messy form as in
Appendix~\ref{AppendixC}.
\end{proof}

%\begin{corollary}
%Given Theorem \ref{Thm fluctuation},  for any $x$ (say $x>0$),
%\begin{equation*}
%\lim_{n\rightarrow\infty}\mathbb{P}(\sqrt{n}(\tau_{n}^{z}-\tau^{z})\geq x)
%=1-\Phi\left(\frac{ax}{\sigma_{Y}(\tau^{z})}\right),
%\end{equation*}
%where $\Phi(x):=\int_{-\infty}^{x}\frac{e^{-y^{2}/2}}{\sqrt{2\pi}}dy$ is the cumulative probability
%distribution function of a standard Gaussian random variable.
%\end{corollary}

\begin{proposition}
Given Theorem \ref{DiffThm}, with $v^b, v^a>0$, for any $x$ (say $x<0$),

(i) \begin{equation}\label{taubfluct}
\lim_{n\rightarrow\infty}\mathbb{P}(\sqrt{n}(\tau_{n}^{b}-\tau^{b})\geq x)
=1-\Phi\left(\sqrt{\frac{q^{b}\lambda v^b}{\psi_{11}+\psi_{22}+\psi_{33}-2\psi_{12}-2\psi_{13}+2\psi_{23}}}x\right),
\end{equation}
where $\Phi(x):=\frac{1}{\sqrt{2\pi}}\int_{-\infty}^{x}e^{-y^{2}/2}dy$ is the cumulative probability distribution function of a standard Gaussian random
variable, and
\begin{equation}
\label{psiij}
\psi_{ij}:=\sum_{k=1}^{6}\Sigma_{ik}\Sigma_{jk}\lambda
+\bar{V}^{i}\bar{V}^{j}v_{d}^{2}\lambda^{3},\qquad 1\leq i,j\leq 6.
\end{equation}

(ii)
\begin{equation}\label{tauafluct}
\lim_{n\rightarrow\infty}\mathbb{P}(\sqrt{n}(\tau_{n}^{a}-\tau^{a})\geq x)
=1-\Phi\left(\sqrt{\frac{q^{a}\lambda v^a}{\psi_{44}+\psi_{55}+\psi_{66}-2\psi_{45}-2\psi_{46}+2\psi_{56}}}x\right).
\end{equation}
\end{proposition}

\begin{proof}
Similar to the proof of the fluctuation of the execution time $\tau_{n}^{z}$, we can show that, for any $x<0$,
\begin{equation*}
\lim_{n\rightarrow\infty}\mathbb{P}(\sqrt{n}(\tau_{n}^{z}-\tau^{z})\geq x)
=\mathbb{P}((\mathbf{\Psi}^{1}-\mathbf{\Psi}^{2}-\mathbf{\Psi}^{3})(\tau^{b})>-(Q^{b})'(\tau^{b})x),
\end{equation*}
From the expression of $Q^b, \tau^b$ in Eqns. (\ref{Qb}), (\ref{taub}), and  (\ref{Eqn_DiffThm}),
%\begin{equation}
%Q^{b}(t)=q^{b}+\lambda(\bar{V}^{1}-\bar{V}^{2}-\bar{V}^{3})t,
%\qquad
%\tau^{b}=\frac{-q^{b}}{\lambda(\bar{V}^{1}-\bar{V}^{2}-\bar{V}^{3})},
%\end{equation}
%and
%\begin{equation}
%\overrightarrow{\mathbf{\Psi}}=\Sigma\overrightarrow{\mathbf{W}}\circ\lambda\mathbf{e}
%-\overrightarrow{\bar{V}}v_{d}\lambda\mathbf{W}_{1}\circ\lambda\mathbf{e}.
%\end{equation}
 it is clear that $(Q^{b})'(\tau^{b})=-q^{b}$ and the mean of $(\mathbf{\Psi}^{1}-\mathbf{\Psi}^{2}-\mathbf{\Psi}^{3})(t)$ is zero
and the variance is
\begin{equation*}
(\psi_{11}+\psi_{22}+\psi_{33}-2\psi_{12}-2\psi_{13}+2\psi_{23})t.
\end{equation*}
%where %$\psi_{ij}=\sum_{k=1}^{6}\Sigma_{ik}\Sigma_{jk}\lambda+\bar{V}^{i}\bar{V}^{j}v_{d}^{2}\lambda^{3}$.
Therefore,
\begin{equation*}
\lim_{n\rightarrow\infty}\mathbb{P}(\sqrt{n}(\tau_{n}^{b}-\tau^{b})\geq x)
=1-\Phi\left(\sqrt{\frac{q^{b}\lambda v^b}{\psi_{11}+\psi_{22}+\psi_{33}-2\psi_{12}-2\psi_{13}+2\psi_{23}}}x\right).
\end{equation*}
Similarly, we can show that
Eqn.~\eqref{tauafluct} holds.
\end{proof}

%\subsection{Large deviation for the tails of the hitting time}

Finally, we have the large deviations for the tails of the hitting time.
Indeed, given Assumptions \ref{fluidd}, \ref{fluidv} and \ref{flui}, we have the fluid limit in Theorem \ref{FluidThm},
%\begin{align*}
%&{Q}^{b}(t)=q^{b}-\lambda v^b t\wedge\tau,
%\\
%&{Q}^{a}(t)=q^{a}-\lambda v^a t\wedge\tau,
%\end{align*}
we see that $\tau_{n}\rightarrow\tau:=\tau^{b}\wedge\tau^{a}$.
More generally, by replacing Assumptions \ref{fluidd} and \ref{fluidv} by the stronger
Assumption \ref{LDPAssump}, we can use the large deviations result to study the tail probabilities of the hitting time $\tau_{n}$
as $n$ goes to $\infty$. Note that for any $t>\tau$,
\begin{equation*}
\mathbb{P}(\tau_{n}\geq t)=\mathbb{P}\left(Q_{n}^{b}(s)>0,Q_{n}^{a}(s)>0,0\leq s< t\right)
=\mathbb{P}\left({Q}_{n}^{b}(s)>0,{Q}_{n}^{a}(s)>0,0\leq s< t\right).
\end{equation*}
And for any $t<\tau$,
\begin{align*}
\mathbb{P}(\tau_{n}\leq t)
&=\mathbb{P}\left(Q_{n}^{b}(s)\le 0\text{ or }Q_{n}^{a}(s)\le 0, \text{for some }0\leq s\leq t\right),
\\
&=\mathbb{P}\left({Q}_{n}^{b}(s)\le 0\text{ or }{Q}_{n}^{a}(s)\le 0, \text{for some }0\leq s\leq t\right).
\end{align*}

Now recall  the large deviation principle for $\mathbb{P}({Q}_{n}^{b}(\cdot)\in\cdot,{Q}_{n}^{a}(\cdot)\in\cdot)$,
i.e., Theorem \ref{LDPThm}, and recall that
%for any $t>\tau$ we have
%\begin{equation*}
%\lim_{n\rightarrow\infty}\frac{1}{n}\log\mathbb{P}(\tau_{n}\geq t)
%=-\inf_{\substack{f^{b}(s)\geq 0,
%\\
%f^{a}(s)\geq 0,
%\\
%\text{for any }0\leq s\leq t}}I(f^{b},f^{a})
%=-\inf_{\substack{f^{b}(s)\geq 0,
%\\
%f^{a}(s)\geq 0,
%\\
%\text{for any }0\leq s\leq t}}\inf_{\phi\in\mathcal{G}_{f}}\mathcal{I}(\phi).
%\end{equation*}
%Similarly, for any $t<\tau$,
%\begin{equation*}
%\lim_{n\rightarrow\infty}\frac{1}{n}\log\mathbb{P}(\tau_{n}\leq t)
%=-\inf_{\substack{f^{b}(s)\leq 0\text{ for some $0\leq s\leq t$}
%\\
%\text{or }f^{a}(s)\leq 0\text{ for some $0\leq s\leq t$}}}I(f^{b},f^{a})
%=-\inf_{\substack{f^{b}(s)\leq 0\text{ for some $0\leq s\leq t$}
%\\
%\text{or }f^{a}(s)\leq 0\text{ for some $0\leq s\leq t$}}}\inf_{\phi\in\mathcal{G}_{f}}\mathcal{I}(\phi).
%\end{equation*}
%Recall that $\mathcal{G}_{f}$ consists of the functions
%$\phi=(\phi^{j}(t),1\leq j\leq 6)\in\mathcal{AC}_{0}[0,\infty)$ and
$
f^{b}(t)=q^{b}+\phi^{1}(t)-\phi^{2}(t)-\phi^{3}(t),
f^{a}(t)=q^{a}+\phi^{4}(t)-\phi^{5}(t)-\phi^{6}(t).
$
Therefore, we have the following,
\begin{corollary}
Given Theorem \ref{LDPThm}, for any $t>\tau$,
\begin{equation*}
\lim_{n\rightarrow\infty}\frac{1}{n}\log\mathbb{P}(\tau_{n}\geq t)
=-\inf_{\substack{q^{b}+\phi^{1}(s)-\phi^{2}(s)-\phi^{3}(s)\geq 0,
\\
q^{a}+\phi^{4}(s)-\phi^{5}(s)-\phi^{6}(s)\geq 0,
\\
\text{for any }0\leq s\leq t
\\
\phi\in\mathcal{AC}_{0}[0,\infty)}}\mathcal{I}(\phi).
\end{equation*}
Similarly, for any $t<\tau$,
\begin{equation*}
\lim_{n\rightarrow\infty}\frac{1}{n}\log\mathbb{P}(\tau_{n}\leq t)
=-\inf_{\substack{q^{b}+\phi^{1}(s)-\phi^{2}(s)-\phi^{3}(s)\leq 0\text{ for some $0\leq s\leq t$}
\\
\text{or }q^{a}+\phi^{4}(s)-\phi^{5}(s)-\phi^{6}(s)\leq 0\text{ for some $0\leq s\leq t$}
\\
\phi\in\mathcal{AC}_{0}[0,\infty)}}\mathcal{I}(\phi).
\end{equation*}
\end{corollary}

\section{Extensions and discussions}
\label{Sec discussion}

\subsection{General assumptions for cancellation} In the previous
section, we have derived the fluid limit and fluctuations for the order positions
under the simple assumption that cancellation is uniform on the queue. This
assumption can be easily relaxed and the analysis can be modified
fairly easily.

For instance, one may assume (more realistically) that
the closer the order to the queue head, the less likely it is
cancelled. More generally, one may replace the term
$\frac{{Z}_n(t-)}{{Q}_n^b(t-)}$ in Eqn.~\eqref{scaledq} with
$\Upsilon\left(\frac{{Z}_n(t-)}{{Q}_n^b(t-)}\right)$ where $\Upsilon$
is a Lipschitz-continuous increasing function from $[0, 1]$ to $[0,
1]$ with $\Upsilon(0)=0$ and $\Upsilon(1)=1$. Now,
%\begin{assumption} \label{general_cancel}
%With order position $z(t)$ and best bid queue length $q^b(t)$,
%the proportion of cancellation before the particular order is determined by %$\Upsilon\left(\frac{z(t)}{q^b(t)}\right)$,
%where $\Upsilon$ is a Lipschitz continuous increasing function from $[0, 1]$ to $[0, 1]$ with %$\Upsilon(0)=0$ and $\Upsilon(1)=1$.
%\end{assumption}
the dynamics of the scaled processes are described as
\begin{align}\label{generalizedscaled}
d\left(\begin{aligned}
&{Q}_n^b(t)\\
&{Q}_n^a(t)\\
&{Z}_n(t)\end{aligned}\right)
&=\left(
\begin{array}{cccccc}
1 &-1 &-1 &0 &0 &0 \\
0 &0 &0 &1 &-1 &-1  \\
0 &-1 & -\Upsilon\left(\frac{{Z}_n(t-)}{{Q}_n^b(t-)}\right) &0 &0 &0
\end{array}
\right)\mathbb{I}_{{Q}_n^a(t-)>0, {Q}_n^b(t-)>0, {Z}_n(t-)>0}
 \cdot d\overrightarrow{C}_n(t).
%({Q}_n^b(0), {Q}_n^a(0), {Z}_n(0)) &= (Q_n^b(0), Q_n^a(0), Z_n(0)).
%\nonumber
\end{align}
Then the limit processes would follow
\begin{align}\label{generalizedlimit}
d\left(\begin{aligned}
&{Q}^b(t)\\
&{Q}^a(t)\\
&{Z}(t)\end{aligned}\right)
&=\left(
\begin{array}{cccccc}
1 &-1 &-1 &0 &0 &0 \\
0 &0 &0 &1 &-1 &-1  \\
0 &-1 & -\Upsilon\left(\frac{{Z}(t-)}{{Q}^b(t-)}\right) &0 &0 &0
\end{array}
\right)\mathbb{I}_{{Q}^a(t-)>0, {Q}^b(t-)>0, {Z}(t-)>0}
 \cdot d\overrightarrow{C}(t)
%({Q}_n^b(0), {Q}_n^a(0), {Z}_n(0)) &= (Q_n^b(0), Q_n^a(0), Z_n(0)).
%\nonumber
\end{align}
%The $\Upsilon$ function describes how the relative position $Z(t)/Q^b(t)$ determines the %cancellation rate. If the orders are more likely to cancel when they get closer to the queue %head, then $\Upsilon$ is a convex function, and vice versa.
\begin{theorem}\label{FluidThmGeneralized}
Given Assumptions \ref{fluidd}, \ref{fluidv}, and \ref{flui}, and
 the scaled processes $(\mathbf{Q}_n^b, \mathbf{Q}_n^a, \mathbf{Z}_n)$ defined by ~ Eqn.~\eqref{generalizedscaled}. If there exist constants $q^b$, $q^a$, and $z$ such that
\begin{align*}
(Q_n^b(0), Q_n^a(0), Z_n(0)) \Rightarrow (q^b, q^a, z),
\end{align*}
then for any $T>0$, %the scaled sequence of process $\tilde{\mathbf{Q}}_n^b$, $\tilde{\mathbf{Q}}_n^a$, and %$\tilde{\mathbf{Z}}_n$ defined in
Eqn.~\eqref{scaledII}
%converges to a limit process  $D^3[0,\infty)$ with $J_1$ topology. That is,
\begin{align*}
&({\mathbf{Q}}_n^b, {\mathbf{Q}}_n^a, {\mathbf{Z}}_n )  \Rightarrow ({\mathbf{Q}}^b, {\mathbf{Q}}^a, {\mathbf{Z}}) \qquad\text{in} \quad (D^3[0,T], J_1),
\end{align*}
where $(\mathbf{Q}^b, \mathbf{Q}^a, \mathbf{Z})$ is defined by Eqn.~\eqref{generalizedlimit} and
\begin{align}\label{initial}
(Q^b(0), Q^a(0), Z(0)) = (q^b, q^a, z).
\end{align}
\end{theorem}
\begin{proof}
First, let us extend the definition of $\Upsilon$ from $[0, 1]$ to $\mathbb{R}$ by
\begin{equation*}
\Upsilon(x) = \Upsilon(x)\mathbb{I}_{0 \le x \le 1} +\mathbb{I}_{1 <x}.
\end{equation*}
Then $\Upsilon$ is (still) Lipschitz-continuous and increasing on $\mathbb{R}$. That is, there exists $K>0$, such that for any $z_{1},z_{2}\in\mathbb{R}$,
$\left|\Upsilon(z_1) - \Upsilon(z_2)\right| \le K|z_1-z_2|$.
Next, define $\tau = \min\{\tau^b, \tau^a, \tau^z\}$ with
$\tau^b = \inf\{t \mid Q^b(t)\le 0\}$, $\tau^a = \inf\{t \mid Q^a(t)\le 0\}$,
and $\tau^z = \inf\{t \mid Z(t)\le 0\}$.
Similar to the argument for Lemma \ref{zlessthanqb},
$\Upsilon \in [0, 1]$ and $z, q^b>0$
imply that $Z_n(t) \le Q_n^b(t)$ and $Z(t) \le Q^b(t)$ for any time before hitting zero. Thus  $\tau^z \le \tau^b$.
Now the remaining part of the proof is similar to that of Theorem \ref{FluidThm} except for
the global existence and local uniqueness of the solution to Eqn.~\eqref{generalizedlimit}, with
\begin{align}
%&{Q}^{b}(t)=q^{b}-\lambda v^b(t\wedge \tau),\\
%&{Q}^{a}(t)=q^{a}-\lambda v^a(t\wedge \tau),\\
&\frac{d{Z}(t)}{dt}=-\lambda\left(\bar{V}^{2}+\bar{V}^{3}\Upsilon\left(\frac{{Z}(t-)}{{Q}^{b}(t-)}\right)\right)\mathbb{I}_{t \le \tau}. \label{generalizedzt}
\end{align}
 Denote the right hand side of Eqn.~\eqref{generalizedzt} by $\vartheta(Z, t)$, and define $\vartheta(Z, q^b/(\lambda v^b)) = 1$.
Let $\{T_i\}_{i \ge 1}$ be an increasing positive sequence with $\lim_{i \to \infty} T_i =\tau$.
Then for any $z_1, z_2 \ge 0$ and $0 \le t \le T_i$,
\begin{align*}
\left|\vartheta(z_1, t)-\vartheta(z_2, t)\right| &= \lambda \bar{V}^3 \left| \Upsilon\left(\frac{z_1}{q^b-\lambda v^b t} \right)
-\Upsilon\left(\frac{z_2}{q^b-\lambda v^b t}\right)\right| \\
&\le \lambda \bar{V}^3 K \left|\frac{z_1}{q^b-\lambda v^b t}-\frac{z_2}{q^b-\lambda v^b t}\right|\\
&\le \frac{\lambda \bar{V}^3 K}{q^b-\lambda v^b T_i}|z_1-z_2|.
\end{align*}
Therefore $\vartheta(Z, t)$ is Lipschitz-continuous in $Z$ and continuous in $t$ for any $t < T_i$ and $Z>0$.
 By the Picard's existence theorem,  there exists a unique solution to Eqn.~\eqref{generalizedzt} with the initial condition $Z(0) =z$ on $[0, T_i]$.
Now letting $i \to \infty$,  the unique solution exists in $[0, \tau)$. Moreover, by the boundedness of  $\vartheta(Z, \tau)$ and the continuity of
 $Z(t)$  at $\tau$, the unique solution also exists at $t=\tau$. For $t> \tau$, $\vartheta(Z, 0)=0$ and $Z(t) =Z(\tau)$.
Hence there exists a unique solution $Z(t)$ for $t\ge 0$.
Note that $\tau^a=\infty$ (resp. $\tau^b= \infty$) when $v^a < 0$ (resp. $v^b<0$).
However, since the right hand side of Eqn.~\eqref{generalizedzt} is less than or equal to $-\lambda \bar{V}^2$,
it follows that $Z(t)$ is decreasing in $t$ and hits $0$ in finite time. Therefore $\tau$ is well defined.
%We can borrow all other argument in the proof of Theorem \ref{FluidThm} and claim the desired %convergence by Theorem \ref{kurtz}.
\end{proof}

%We can actually push the result a little bit further. Let $L(t) = \frac{Z(t)}{q^b-\lambda v^b %%t}$, then Eqn.~\eqref{generalizedzt} is turned into
%\begin{equation}\label{LEqn}
%\frac{d L(t)}{\lambda (v^{b}L(t)-\bar{V}^3 \Upsilon(L(t))-\bar{V}^2)} = \frac{dt}{q^b-\lambda %v^bt},
%\end{equation}
%with the initial condition is $L(0) = z/q^b$. Integrating
%Eqn.~\eqref{LEqn} from $0$ to $t$, $L(t)$ %and hence the limiting
%order position $Z(t)$ is determined
%implicitly by the following equation:
%\begin{equation}
%\int_{\frac{z}{q^{b}}}^{\frac{Z(t)}{q^{b}-\lambda v^{b}t}}\frac{d L}{\lambda (v^{b}L-\bar{V}^3 %\Upsilon(L)-\bar{V}^2)}
%=\frac{1}{\lambda v^{b}}\log\left(\frac{q^{b}}{q^{b}-\lambda v^{b}t}\right),
%\qquad t<\tau.
%\end{equation}

\subsection{Linear dependence between the order arrival and the trading volume}
One may also replace Assumptions \ref{fluidd} and \ref{fluidv} by the assumption that order arrival rate is linearly  correlated with trading volumes.
The fluid limit can be analyzed in a similar way with few modifications.

\begin{assumption}
\label{assumption linear}
 $N(nt)$ is a simple point process
with an intensity $
n\lambda+\alpha n{Q}^{a}_{n}\left(t-\right)+\beta n{Q}^{b}_{n}\left(t-\right)$
at time $t$, where $\alpha,\beta$ are positive constants.
\end{assumption}
\begin{assumption} \label{strongerV}
For any $1 \le j \le 6$, $\{V_{i}^{j}\}_{i\geq 1}$ is a sequence of stationary, ergodic, and uniformly bounded sequence. Moreover, for any $i \ge 2$ and $\mathcal{G}_i= \sigma (\{\overrightarrow{V}_k\}_{1 \le k \le i})$,
\begin{align*}
\mathbb{E}[\overrightarrow{V}_i \mid \mathcal{G}_{i-1}] = \overrightarrow{\bar{V}}.
\end{align*}
\end{assumption}

\begin{theorem}
\label{Thm lineararrival}
Given Assumptions  \ref{cancelprop}, \ref{assumption linear}, and \ref{strongerV}, then Theorem \ref{FluidThm}  holds except that the limit processes
will be replaced by
\begin{align}
\label{Eqn Qb}
&{Q}^{b}(t)=-\frac{\alpha q^{a}v^{b}-\alpha q^{b}v^{a}+\lambda v^{b}}{v^{a}\alpha+v^{b}\beta}
+\frac{v^{b}(\beta q^{b}+\alpha q^{a}+\lambda)}{\beta v^{b}+\alpha v^{a}}e^{-(v^{b}\beta+v^{a}\alpha)t\wedge\tau},
\end{align}
\begin{align}
\label{Eqn Qa}
&{Q}^{a}(t)=-\frac{\beta q^{b}v^{a}-\beta q^{a}v^{b}+\lambda v^{a}}{v^{a}\alpha+v^{b}\beta}
+\frac{v^{a}(\beta q^{b}+\alpha q^{a}+\lambda)}{\beta v^{b}+\alpha v^{a}}e^{-(v^{b}\beta+v^{a}\alpha)t\wedge\tau},
\end{align}
and
\begin{align}
\label{Eqn Z}
{Z}(t)&=ze^{-\int_{0}^{t\wedge\tau}\bar{V}_{3}\left[\frac{\lambda}{{Q}^{b}(s)}+\beta+\frac{\alpha {Q}^{a}(s)}{{Q}^{b}(s)}\right]ds}
\\
&\qquad\qquad
-\int_{0}^{t\wedge\tau}\bar{V}_{2}[\lambda+\beta {Q}^{b}(s)+\alpha {Q}^{a}(s)]
e^{-\int_{s}^{t\wedge\tau}\bar{V}_{3}\left[\frac{\lambda}{{Q}^{b}(u)}+\beta+\frac{\alpha {Q}^{a}(u)}{{Q}^{b}(u)}\right]du}ds.
\nonumber
\end{align}
\end{theorem}

%We also assume that
%\begin{equation}
%\mathbb{E}[\overrightarrow{V}_{i} \mid \mathcal{G}_{i-1}]=\overrightarrow{\bar{V}}.
%\end{equation}
\begin{proof}
Recall that before $t \le \tau$, with Assumption \ref{assumption linear},
\begin{align*}
d\left(\begin{aligned}
&{Q}_n^b(t)\\
&{Q}_n^a(t)\\
&{Z}_n(t)\end{aligned}\right)
&=\left(
\begin{array}{cccccc}
1 &-1 &-1 &0 &0 &0 \\
0 &0 &0 &1 &-1 &-1  \\
0 &-1 & -\frac{Z_n(t-)}{{Q}_n^b(t-)} &0 &0 &0
\end{array}
\right)\mathbb{I}_{{Q}^{a}_{n}(t-)>0, {Z}_{n}(t-)>0}
 \cdot d \overrightarrow{C}_n(t),
\end{align*}%
where
\begin{equation*}
\overrightarrow{C}_{n}(t)=\frac{1}{n}\sum_{i=1}^{N(nt)}\overrightarrow{V}_{i}
=M_{n}(t)+\int_{0}^{t}(\lambda+\beta {Q}_{n}^{b}(s-)+\alpha {Q}_{n}^{a}(s-))ds\overrightarrow{\bar{V}}.
\end{equation*}
Here
\begin{align*}
\overrightarrow{M}_{n}(t)
&=\frac{1}{n}\sum_{i=1}^{N(nt)}[\overrightarrow{V}_{i}-\overrightarrow{\bar{V}}]
+\frac{1}{n}\overrightarrow{\bar{V}}\left[N(nt)-n\int_{0}^{t}(\lambda+\beta {Q}_{n}^{b}(s-)+\alpha {Q}_{n}^{a}(s-))ds\right]
%\\
%&=\frac{1}{n}\sum_{i=1}^{N(nt)}[\overrightarrow{V}_{i}-\overrightarrow{\bar{V}}]
%+\frac{1}{n}\overrightarrow{\bar{V}}\left[N(nt)-\int_{0}^{t}\left(n\lambda+n\beta {Q}_{n}^{b}\left(\frac{ns-}{n}\right)+n\alpha %{Q}_{n}^{a}\left(\frac{ns-}{n}\right)\right)ds\right]
%\nonumber
%\\
%&=\frac{1}{n}\sum_{i=1}^{N(nt)}[\overrightarrow{V}_{i}-\overrightarrow{\bar{V}}]
%+\frac{1}{n}\overrightarrow{\bar{V}}\left[N(nt)-\int_{0}^{nt}\left(\lambda+\beta n{Q}_{n}^{b}\left(\frac{s-}{n}\right)+\alpha %n{Q}_{n}^{a}\left(\frac{s-}{n}\right)\right)ds\right]
%\nonumber
\end{align*}
is a martingale. Similar to the arguments  before, we can show that
$(\mathbf{{Q}}_{n}^{b},\mathbf{{Q}}^{a}_n,\mathbf{{Z}}_{n})\Rightarrow(\mathbf{{Q}}^{b},\mathbf{{Q}}^{a},\mathbf{Z})$,
where $(\mathbf{{Q}}^{b},\mathbf{{Q}}^{a},\mathbf{{Z}})$ satisfies the ODE:
\begin{equation*}
d\left(\begin{aligned}
&{Q}^b(t)\\
&{Q}^a(t)\\
&{Z}(t)\end{aligned}\right)
=\left(
\begin{array}{cccccc}
1 &-1 &-1 &0 &0 &0 \\
0 &0 &0 &1 &-1 &-1  \\
0 &-1 & -\frac{Z_n(t-)}{{Q}_n^b(t-)} &0 &0 &0
\end{array}
\right)\mathbb{I}_{{Q}^a(t-)>0, {Z}(t-)>0}
 \cdot  (\lambda+\beta {Q}^{b}(t-)+\alpha {Q}^{a}(t-))\overrightarrow{\bar{V}}dt,
\nonumber
\end{equation*}
with the initial condition $({Q}^{b}(0),{Q}^{a}(0),{Z}(0))=(q^{b},q^{a},z)$.
The equations for ${Q}^{b}(t)$ and ${Q}^{a}(t)$ can be written down more explicitly as
\begin{align*}
&d{Q}^{b}(t)=(\lambda+\beta {Q}^{b}(t-)+\alpha {Q}^{a}(t-))(\bar{V}_{1}-\bar{V}_{2}-\bar{V}_{3})dt,
\\
&d{Q}^{a}(t)=(\lambda+\beta {Q}^{b}(t-)+\alpha {Q}^{a}(t-))(\bar{V}_{4}-\bar{V}_{5}-\bar{V}_{6})dt,
\end{align*}
which can be further simplified as
\begin{equation*}
d
\left(
\begin{array}{c}
{Q}^{b}(t)
\\
{Q}^{a}(t)
\end{array}
\right)
=\left(
\begin{array}{cc}
-v^{b}\beta & -v^{b}\alpha
\\
-v^{a}\beta & -v^{a}\alpha
\end{array}
\right)
\left(
\begin{array}{c}
{Q}^{b}(t)
\\
{Q}^{a}(t)
\end{array}
\right)
-
\left(
\begin{array}{c}
\lambda v^{b}
\\
\lambda v^{a}
\end{array}
\right).
\end{equation*}
Hence, for $t\leq\tau$, we get
\begin{equation*}
\left(
\begin{array}{c}
{Q}^{b}(t)
\\
{Q}^{a}(t)
\end{array}
\right)
=c_{1}
\left(
\begin{array}{c}
\alpha
\\
-\beta
\end{array}
\right)
+c_{2}e^{-(v^{b}\beta+v^{a}\alpha)t}
\left(
\begin{array}{c}
v^{b}
\\
v^{a}
\end{array}
\right)
-
\left(
\begin{array}{c}
\frac{\lambda}{\beta}
\\
0
\end{array}
\right),
\end{equation*}
where $c_{1}$, $c_{2}$ are constants that can be determined from the initial condition,
\begin{equation*}
c_{1}=-\frac{q^{a}v^{b}-\frac{\lambda v^{a}}{\beta}-q^{b}v^{a}}{v^{a}\alpha+v^{b}\beta},
\qquad
c_{2}=\frac{\beta q^{b}+\alpha q^{a}+\lambda}{\beta v^{b}+\alpha v^{a}}.
\end{equation*}
Hence Eqns (\ref{Eqn Qb}) and  (\ref{Eqn Qa}) follow.

Finally, ${Z}(t)$ satisfies the first-order ODE
\begin{equation*}
d{Z}(t)+{Z}(t)\bar{V}_{3}\left(\frac{\lambda}{{Q}^{b}(t)}+\beta+\frac{\alpha {Q}^{a}(t)}{{Q}^{b}(t)}\right)dt
=-\bar{V}_{2} (\lambda+\beta {Q}^{b}(t)+\alpha {Q}^{a}(t))dt,
\end{equation*}
whose solution is given by Eqn.~\eqref{Eqn Z}.
\end{proof}

%Let us assume the following:
%\begin{align}
%&v^{b}\beta+v^{a}\alpha>0,
%\\
%&c_{1}\alpha-\frac{\lambda}{\beta}<0,
%\qquad
%c_{1}>0,
%\end{align}
%which is equivalent to the following condition:
%\begin{assumption}\label{NegAssump}
\begin{cor}
Given Assumptions  \ref{cancelprop}, \ref{assumption linear},
and \ref{strongerV}, assume further that $v^{b}\beta+v^{a}\alpha>0$ and $-\frac{\lambda v^{b}}{\alpha}<q^{a}v^{b}-q^{b}v^{a}<\frac{\lambda v^{a}}{\beta}.$
%\end{assumption}
Then ${Q}^{b}(t)$ and ${Q}^{a}(t)$ will hit zero at some finite times $\tau^{b}$ and $\tau^{a}$ respectively. Moreover,
\begin{align*}
&\tau^{b}=-\frac{1}{v^{b}\beta+v^{a}\alpha}
\log\left(\frac{v^{b}\lambda+q^{a}v^{b}\alpha-q^{b}v^{a}\alpha}{v^{b}\beta q^{b}+v^{b}\alpha q^{a}+\lambda v^{b}}\right),
\\
&\tau^{a}=-\frac{1}{v^{b}\beta+v^{a}\alpha}
\log\left(\frac{-q^{a}v^{b}\beta+q^{b}v^{a}\beta+\lambda v^{a}}
{\beta q^{b}v^{a}+\alpha q^{a}v^{a}+\lambda v^{a}}\right),
\end{align*}
and $\tau^{z}$ is determined via the equation
\begin{equation*}
z=\int_{0}^{\tau^{z}}\bar{V}_{2}(\lambda+\beta {Q}^{b}(s)+\alpha {Q}^{a}(s))
e^{\int_{0}^{s}\bar{V}_{3}\left(\frac{\lambda}{{Q}^{b}(u)}+\beta+\frac{\alpha {Q}^{a}(u)}{{Q}^{b}(u)}\right)du}ds.
\end{equation*}
\end{cor}

\subsection{Various forms of diffusion limits}
%\label{Remark consistent}
There is more than one possible alternative set of assumptions under which appropriate forms of diffusion limits may be derived.
For instance, one may impose a weaker condition than Assumption \ref{assumptionD} for $\{D_i\}_{i \ge 1}$.
\begin{assumption}\label{weakD}
For any time $t$,
\begin{align*}
\lim_{n \to \infty} \frac{N(nt)}{n} = \lambda t,  \quad\text{a.s.}
\end{align*}
Moreover, there exists $K > 0$, such that $\mathbb{E}[N(t)] \le Kt$, for
any $t$.
\end{assumption}
This assumption holds, for example, if the point process
$N(t)$ is stationary and ergodic with finite mean.
To compensate for the weakened Assumption \ref{weakD}, one may need a stronger condition on $\{\overrightarrow{V}_i\}_{i \ge 1}$, for instance, Assumption \ref{strongerV}.
%\begin{assumption}\label{strongerV}
%Assume that for any $1 \le j \le 6$, $\{V_{i}^{j}\}_{i\geq 1}$ is a sequence of stationary, ergodic and uniformly bounded sequence. Moreover, for any $i \ge 2$,
%\begin{align}
%\mathbb{E}[\overrightarrow{V}_i|\mathcal{G}_{i-1}] = \overrightarrow{\bar{V}},
%\end{align}
%where $\mathcal{G}$ is the filtration generated by $\{\overrightarrow{V}_{i}\}_{i\geq 1}$.
%\end{assumption}

Note that under this alternative set of assumptions, the resulting limit process will  in fact be simpler than Theorem \ref{DiffThm}.
 This is because  Assumption \ref{strongerV} implies that
$V_i^j$ is actually uncorrelated to $V_{i'}^j$ for any $i \neq i'$ and $1 \le j \le 6$.
Hence the  covariance of $V_1^j$ and $V_i^j$, $i \ge 2$ in the limit process may vanish.  We illustrate this in some detail below.

Making Assumptions \ref{strongerV} and \ref{weakD},
define a modified version of the scaled net order flow process
$\overrightarrow{\mathbf{\Psi}}^*_{n}$ by
\begin{equation} \label{CnDfn2}
\overrightarrow{\Psi}^*_n(t)=\frac{1}{\sqrt{n}}  \sum_{i=1}^{N(nt)} \big(\overrightarrow{V}_i - \overrightarrow{\bar{V}}\big),
\end{equation}
 while the scaled processes $R_n^b(t)$, $R_n^a(t)$ still follows
 Eqn.~\eqref{Qn},  the first hitting time the same as in Eqn.~\eqref{taun},
and  the corresponding limit processes in Eqn.~\eqref{tau} and Eqn.~\eqref{QTruncated}.  Then we have the following.
\begin{theorem} \label{DiffThm2}
Given Assumptions \ref{flui}, \ref{strongerV},  and \ref{weakD},
for any $T>0$,

(i) $\overrightarrow{\mathbf{\Psi}}^*_{n} \Rightarrow \overrightarrow{\mathbf{\Psi}}^{*}$
where $\overrightarrow{\mathbf{\Psi}}^*=(\sigma_{j}W_{j}, 1\le j\le 6)$,
where $(W_j, 1\le j\le 6)$ is a standard six-dimensional Brownian
motion and $\sigma_j^2 = \lambda\mbox{Var}(V_1^j)$.

(ii)
$(\mathbf{{R}}_n^b, \mathbf{{R}}_n^a) \Rightarrow (\mathbf{{R}}^b, \mathbf{{R}}^a)  \qquad \text{in} \quad (D^2[0, T], J_1).$
\end{theorem}

\begin{proof}
Under Assumption \ref{strongerV}, it is clear that
\begin{equation*}
\overrightarrow{\Psi}^*_{n}(t)=\frac{1}{\sqrt{n}}  \sum_{i=1}^{N(nt)}
\big(\overrightarrow{V}_i - \mathbb{E}[\overrightarrow{V}_{i} \mid \mathcal{G}_{i-1}]\big)
\end{equation*}
is a martingale. Now define for $j=1,2,\ldots,6$,
\begin{equation*}
M^{j}_{nt}:=\sum_{i=1}^{N(nt)}\left(V^{j}_{i}-\mathbb{E}[V^{j}_{i}
  \mid \mathcal{G}_{i-1}]\right)
=\sum_{i=1}^{N(nt)}(V^{j}_{i}-\bar{V}^{j}).
\end{equation*}
First, the jump size of $M^{j}_{nt}$ is uniformly bounded since $N(nt)$ is a simple point process
and by Assumption \ref{strongerV}, $V^{j}_{i}$'s are uniformly bounded.
Next, the quadratic variation of $M^{j}_{nt}$ is given by
\begin{equation*}
[M^{j}]_{nt}=\sum_{i=1}^{N^{j}(nt)}(V^{j}_{i}-\bar{V}^{j})^{2}.
\end{equation*}
By Assumptions \ref{strongerV} and \ref{weakD} and the Ergodic theorem, as $t\rightarrow\infty$,
\begin{equation*}
\frac{[M^{j}]_{t}}{t}\rightarrow\lambda\mbox{Var}[V^{j}], \quad \mbox{a.s}.
\end{equation*}
 Moreover, since $M^{j}$ and $M^{k}$ have no common jumps for $j\neq k$,
\begin{equation*}
[M^{j},M^{k}]_{t}\equiv 0.
\end{equation*}
Therefore, applying the FCLT for martingales~\cite[Theorem VIII-3.11]{jacod1987limit},
for any $T>0$, we have
\begin{equation*}
\overrightarrow{\mathbf{\Psi}}^*_{n} \Rightarrow \overrightarrow{\mathbf{\Psi}}^*,\qquad\text{in $(D^6[0,T],J_{1})$},
\end{equation*}
To see the second part of the claim, first note that
by Assumption \ref{weakD},
\begin{equation*}
\frac{1}{n}\sum_{i=1}^{N(n\cdot)}\overrightarrow{\bar{V}}
\Rightarrow\lambda\overrightarrow{\bar{V}}\mathbf{e}, \qquad\text{in $(D[0,T],J_{1})$} \quad \mbox{a.s.}
\end{equation*}
as $n\rightarrow\infty$. The remaining of the proof is to check the conditions for Theorem \ref{kurtz} as in  the proof of Theorem \ref{FluidThm}.
The quadratic variance of $M_{nt}:=(M^{j}_{nt})_{1\leq j\leq 6}$ is given by
\begin{align*}
\mathbb{E}\left[\left[\frac{1}{\sqrt{n}}M\right]_{nt}\right]
=\frac{1}{n}\sum_{1\leq j\leq 6}\mathbb{E}[N(nt)]
\mathbb{E}\left[\left(V_{i}^{j}-\mathbb{E}\left[V_{i}^{j} \mid \mathcal{F}_{T_{i}^{j}-}\right]\right)^{2}\right]
%\nonumber \\
\leq Kt\sum_{1\leq j\leq 6}\mathbb{E}\left[\left(V_{1}^{j}\right)^{2}\right],
%\nonumber
\end{align*}
which is uniformly bounded in $n$.
The total variation of $A_{n}:=\frac{1}{n}\sum_{i=1}^{N(nt)}\overrightarrow{\bar{V}}$ satisfies
\begin{equation*}
\mathbb{E}[[T(A_{n})]_{t}]
\leq\sum_{1\leq j\leq 6}\mathbb{E}\left[\frac{1}{n}\sum_{i=1}^{N(nt)}|\bar{V}^{j}|\right]
\leq\sum_{1\leq j\leq 6}Kt\mathbb{E}[|\bar{V}^{j}|],
\end{equation*}
which is uniformly bounded in $n$.  The proof is complete.
\end{proof}

\bibliographystyle{abbrv}
\bibliography{optimal_dynamics}

\appendix

\section{Some large deviations results}
\label{AppendixB}

%According to \cite[Theorem 5.1.2]{dembo2009large}, we have

%\begin{theorem}[Mogulskii's Theorem]
%\label{Mogulskii}
%Assume $(X_{i})_{i\geq 1}$ are i.i.d. random vectors in $\mathbb{R}^{K}$.
%If $\Gamma(\theta):=\log\mathbb{E}[e^{\theta\cdot X_{1}}]<\infty$
%for any $\theta\in\mathbb{R}^{K}$ and let
%\begin{equation}
%\Lambda(x):=\sup_{\theta\in\mathbb{R}^{K}}\left\{\theta\cdot x-\Gamma(\theta)\right\},
%\end{equation}
%then $\mathbb{P}(\frac{1}{n}\sum_{i=1}^{\lfloor nt\rfloor}X_{i}\in\cdot)$ follows
%a large deviation principle on $L^{\infty}[0,\infty)$ with the rate function
%\begin{equation}
%\mathcal{I}(\phi)=\int_{0}^{T}\Lambda(\phi'(t))dt,
%\end{equation}
%for any $\phi\in\mathcal{AC}_{0}[0,\infty)$, the space of absolutely continuous functions starting at $0$ and
%$\mathcal{I}(\phi)=+\infty$ otherwise.
%\end{theorem}

According to \cite[Theorem 2]{dembo1995large}, we have

\begin{theorem}\label{Dembo}
Let $(X_{i})_{i\in\mathbb{N}}$ be a sequence of stationary $\mathbb{R}^{K}$-valued random vectors satisfying Assumption \ref{LDPAssump}
and Assumption \ref{LDPAssumpII}. Then, the empirical mean process $S_{n}(t):=\frac{1}{n}\sum_{i=1}^{\lfloor nt\rfloor}X_{i}$, $0\leq t\leq T$,
satisfies a large deviations principle on $D[0,T]$ equipped with the topology of uniform convergence with the convex good rate function
\begin{equation}
I(\phi):=\int_{0}^{T}\Lambda(\phi'(t))dt,
\end{equation}
for any $\phi\in\mathcal{AC}_{0}[0,\infty)$, the space of absolutely continuous functions starting at $0$ and
$\mathcal{I}(\phi)=+\infty$ otherwise, where
\begin{equation}
\Lambda(x):=\sup_{\theta\in\mathbb{R}^{K}}\{\theta\cdot x-\Gamma(\theta)\},
\end{equation}
with $\Gamma(\theta):=\lim_{n\rightarrow\infty}\frac{1}{n}\log\mathbb{E}[e^{\sum_{i=1}^{n}\theta\cdot X_{i}}]$.
\end{theorem}

\paragraph{Remark.}
Note that the original statement in \cite[Theorem 2]{dembo1995large} applies to Banach space-valued $(X_{i})_{i\in\mathbb{N}}$. For the purpose
in our paper, we only need to consider $\mathbb{R}^{K}$-valued $(X_{i})_{i\in\mathbb{N}}$.

\section{$Y(t)$ process}
\label{AppendixC}

%\begin{equation*}
%\psi_{ij}:=\sum_{k=1}^{6}\Sigma_{ik}\Sigma_{jk}\lambda
%+\bar{V}^{i}\bar{V}^{j}v_{d}^{2}\lambda^{3},\qquad 1\leq i,j\leq 6.
%\end{equation*}

\begin{proposition}\label{YVarProp}
$Y(t)$ defined in Eqn (\ref{YEqn}) is a Gaussian process for $t<\tau^{z}$,
%Y(t)&=-\int_{0}^{t}e^{-\int_{s}^{t}\frac{\lambda\bar{V}^{3}}{Q^{b}(u)}du}d\mathbf{\Psi}^{2}(s)
%-\int_{0}^{t}e^{-\int_{s}^{t}\frac{\lambda\bar{V}^{3}}{Q^{b}(u)}du}\frac{Z(s)}{Q^{b}(s)}d\mathbf{\Psi}^{3}(s)
%\\
%&\qquad\qquad
%+\int_{0}^{t}e^{-\int_{s}^{t}\frac{\lambda\bar{V}^{3}}{Q^{b}(u)}du}
%\frac{Z(s)(\mathbf{\Psi}^{1}(s)-\mathbf{\Psi}^{2}(s)-\mathbf{\Psi}^{3}(s))}{(Q^{b}(s))^{2}}\lambda\bar{V}^{3}ds,
%\nonumber
%\end{align}
with mean $0$ and variance $\sigma_{Y}^{2}(t)$. In particular, when $c<0$ and $c\neq-1$,
\begin{align}
\sigma_{Y}^{2}(t)&:=
\frac{(b+ct)^{\frac{2}{c}+1}-b^{\frac{2}{c}+1}}{(2+c)(b+ct)^{\frac{2}{c}}}\sum_{j=1}^{6}
\left(\lambda\left(\Sigma_{2j}-\frac{\Sigma_{3j}a}{(1+c)\lambda\bar{V}^{3}}\right)^{2}
+\frac{\lambda^{3}v_{d}^{2}}{6}\left(\frac{c}{1+c}\bar{V}^{2}\right)^{2}\right)
\nonumber
\\
&\qquad
+\frac{b^{\frac{1}{c}}}{\lambda\bar{V}^{3}}\frac{(b+ct)^{\frac{1}{c}}-b^{\frac{1}{c}}}{(b+ct)^{\frac{2}{c}}}
\left(z+\frac{ab}{1+c}\right)\sum_{j=1}^{6}
2\left(\lambda\left(\Sigma_{2j}-\frac{\Sigma_{3j}a}{(1+c)\lambda\bar{V}^{3}}\right)\Sigma_{3j}+\frac{\lambda^{3}v_{d}^{2}}{6}\frac{c}{1+c}\bar{V}^{2}\bar{V}^{3}\right)
\nonumber
\\
&\qquad\qquad
+\frac{t}{(b+ct)^{\frac{2}{c}+1}}\frac{b^{\frac{2}{c}-1}}{\lambda^{2}(\bar{V}^{3})^{2}}
\sum_{j=1}^{6}\left(\lambda(\Sigma_{3j})^{2}+\frac{\lambda^{3}v_{d}^{2}}{6}(\bar{V}^{3})^{2}\right)
\left(z+\frac{ab}{1+c}\right)^{2}
\nonumber
\\
&-\frac{2a}{(b+ct)^{\frac{2}{c}}(1+c)\lambda\bar{V}^{3}}
\cdot
\bigg(\hat{\alpha}\frac{(b+ct)^{\frac{2}{c}+1}-b^{\frac{2}{c}+1}}{2+c}
+(\hat{\beta}-\hat{\gamma}c)((b+ct)^{\frac{1}{c}}-b^{\frac{1}{c}})
\nonumber
\\
&\qquad
+\hat{\gamma}\left((b+ct)^{\frac{1}{c}}\log(b+ct)-b^{\frac{1}{c}}\log(b)\right)
+\frac{\hat{\delta}}{2}((b+ct)^{\frac{2}{c}}-b^{\frac{2}{c}})
+\frac{\hat{\eta}}{1-c}((b+ct)^{\frac{1}{c}-1}-b^{\frac{1}{c}-1})\bigg)
\nonumber
\\
&+\frac{2}{(b+ct)^{\frac{2}{c}}}\left(z+\frac{ab}{1+c}\right)\frac{b^{\frac{1}{c}}}{\lambda\bar{V}^{3}}
\cdot\bigg(\hat{\alpha}((b+ct)^{\frac{1}{c}}-b^{\frac{1}{c}})
+(\hat{\beta}+\hat{\gamma})\frac{t}{b(b+ct)}
\nonumber
\\
&\qquad
+\frac{\hat{\gamma}}{c}\left(\frac{\log b}{b}-\frac{\log(b+ct)}{b+ct}\right)
+\frac{\hat{\delta}}{1-c}((b+ct)^{\frac{1}{c}-1}-b^{\frac{1}{c}-1})+\frac{\hat{\eta}}{2c}(b^{-2}-(b+ct)^{-2})\bigg).
\nonumber
\end{align}
Here \begin{equation} \label{hatalphaDefn}
\hat{\alpha}=\frac{\alpha}{c+1},
\qquad
\hat{\beta}=-\frac{b^{\frac{1}{c}+1}}{1+c}
-\gamma b^{\frac{1}{c}}
+\frac{\delta}{bc}-\frac{\beta\log b}{c},
\hat{\gamma}=\frac{\beta}{c},
\qquad
\hat{\delta}=\gamma,
\qquad
\hat{\eta}=-\frac{\delta}{c},
\end{equation}
with
\begin{align}\label{alphaDefn}
&\alpha:=-(\psi_{12}-\psi_{22}-\psi_{32})+(\psi_{13}-\psi_{23}-\psi_{33})\frac{a}{(1+c)\lambda\bar{V}^{3}}
-\frac{a\varphi}{c(1+c)\lambda\bar{V}^{3}},
\nonumber
\\
&\beta:=-(\psi_{13}-\psi_{23}-\psi_{33})\left(z+\frac{ab}{1+c}\right)\frac{b^{\frac{1}{c}}}{\lambda\bar{V}^{3}}
+\left(z+\frac{ab}{1+c}\right)\frac{\varphi b^{\frac{1}{c}}}{c\lambda\bar{V}^{3}},
\nonumber
\\
&\gamma:=\frac{ab\varphi}{c(1+c)\lambda\bar{V}^{3}}, \ \ \delta:=-\varphi\left(z+\frac{ab}{1+c}\right)
\frac{b^{\frac{1}{c}+1}}{c\lambda\bar{V}^{3}},
\nonumber \\
&\varphi:=\psi_{11} + \psi_{22} + \psi_{33} - \psi_{12} - \psi_{13} - \psi_{21} - \psi_{31} + \psi_{23}+\psi_{32}.
\end{align}
\end{proposition}

\paragraph{Remark.}
Proposition \ref{YVarProp} only gives the formula for the variance of $Y(t)$ for the case $c\neq-1$, $c<0$.
The variance $\sigma_{Y}^{2}(t)$ for the case $c=-1$ can be taken as a continuum limit as $c\rightarrow-1$.

\begin{proof}[Proof of Proposition \ref{YVarProp}]
%Recall that $Y(t)$ satisfies
%\begin{align}\label{YEqn}
%&dY(t)+\frac{Y(t)}{Q^{b}(t)}\lambda\bar{V}^{3}dt
%=-d\mathbf{\Psi}^{2}(t)-\frac{Z(t)}{Q^{b}(t)}d\mathbf{\Psi}^{3}(t)
%\\
%&\qquad\qquad\qquad\qquad
%+\frac{Z(t)(\mathbf{\Psi}^{1}(t)-\mathbf{\Psi}^{2}(t)-\mathbf{\Psi}^{3}(t))}{(Q^{b}(t))^{2}}\lambda\bar{V}^{3}dt,
%\nonumber
%\end{align}
%with $Y(0)=0$.
By multiplying Eqn.~\eqref{YEqn} by the integrating factor $e^{\int_{0}^{t}\frac{\lambda\bar{V}^{3}}{Q^{b}(s)}ds}$
and integrating from $0$ to $t$, and finally dividing the integrating factor, we get
\begin{align}\label{YEqnII}
Y(t)&=-\int_{0}^{t}e^{-\int_{s}^{t}\frac{\lambda\bar{V}^{3}}{Q^{b}(u)}du}d\Psi^{2}(s)
-\int_{0}^{t}e^{-\int_{s}^{t}\frac{\lambda\bar{V}^{3}}{Q^{b}(u)}du}\frac{Z(s)}{Q^{b}(s)}d\Psi^{3}(s)
\\
&\qquad\qquad
+\int_{0}^{t}e^{-\int_{s}^{t}\frac{\lambda\bar{V}^{3}}{Q^{b}(u)}du}
\frac{Z(s)(\Psi^{1}(s)-\Psi^{2}(s)-\Psi^{3}(s))}{(Q^{b}(s))^{2}}\lambda\bar{V}^{3}ds,
\nonumber
\end{align}
which implies that $Y(t)$ is a Gaussian process since $\overrightarrow{\mathbf{\Psi}}$ is a Gaussian process.
Since $\overrightarrow{\mathbf{\Psi}}$ is centered, i.e., with mean zero, it is easy to see that $Y(t)$ is also centered.
Next, let us determine the variance of $Y(t)$.
By It\^{o}'s formula, we have
\begin{align}\label{YSquare}
d(Y(t)^{2})
&=2Y(t)dY(t)+d\langle Y\rangle_{t}
\\
&=d\langle Y\rangle_{t}
-2Y(t)\frac{Y(t)}{Q^{b}(t)}\lambda\bar{V}^{3}dt
-2Y(t)d\Psi^{2}(t)-2Y(t)\frac{Z(t)}{Q^{b}(t)}d\Psi^{3}(t)
\nonumber
\\
&\qquad\qquad
+2Y(t)\frac{Z(t)(\Psi^{1}(t)-\Psi^{2}(t)-\Psi^{3}(t))}{(Q^{b}(t))^{2}}\lambda\bar{V}^{3}dt
\nonumber.
\end{align}
From Eqn.~\eqref{YEqn}, we get
\begin{equation}\label{YQuadratic}
d\langle Y\rangle_{t}
=d\langle\Psi^{2}\rangle_{t}
+\frac{Z(t)^{2}}{Q^{b}(t)^{2}}d\langle\Psi^{3}\rangle_{t}
+\frac{2Z(t)}{Q^{b}(t)}d\langle\Psi^{2},\Psi^{3}\rangle_{t}.
\end{equation}
Plugging
Eqn.~\eqref{YQuadratic} into
Eqn.~\eqref{YSquare}, and taking expectations
on the both hand sides of the equation, we get
\begin{align*}
d\mathbb{E}[Y(t)^{2}]
&=d\langle\Psi^{2}\rangle_{t}
+\frac{Z(t)^{2}}{Q^{b}(t)^{2}}d\langle\Psi^{3}\rangle_{t}
+\frac{2Z(t)}{Q^{b}(t)}d\langle\Psi^{2},\Psi^{3}\rangle_{t}
\\
&\qquad
-2\mathbb{E}[Y(t)^{2}]\frac{1}{Q^{b}(t)}\lambda\bar{V}^{3}dt
\nonumber
\\
&\qquad
+2\frac{Z(t)(\mathbb{E}[Y(t)\Psi^{1}(t)]-\mathbb{E}[Y(t)\Psi^{2}(t)]
-\mathbb{E}[Y(t)\Psi^{3}(t)])}{(Q^{b}(t))^{2}}\lambda\bar{V}^{3}dt.
\nonumber
\end{align*}
By using the integrating factor $e^{\int_{0}^{t}\frac{2\lambda\bar{V}^{3}}{Q^{b}(s)}ds}$,
we conclude that
\begin{align}\label{YVariance}
&\mathbb{E}[Y(t)^{2}]
\\
&=\int_{0}^{t}e^{-\int_{s}^{t}\frac{2\lambda\bar{V}^{3}}{Q^{b}(u)}du}d\langle\Psi^{2}\rangle_{s}
+\int_{0}^{t}e^{-\int_{s}^{t}\frac{2\lambda\bar{V}^{3}}{Q^{b}(u)}du}
\frac{Z(s)^{2}}{Q^{b}(s)^{2}}d\langle\Psi^{3}\rangle_{s}
+\int_{0}^{t}e^{-\int_{s}^{t}\frac{2\lambda\bar{V}^{3}}{Q^{b}(u)}du}
\frac{2Z(s)}{Q^{b}(s)}d\langle\Psi^{2},\Psi^{3}\rangle_{s}
\nonumber
\\
&\qquad
+\int_{0}^{t}e^{-\int_{s}^{t}\frac{2\lambda\bar{V}^{3}}{Q^{b}(u)}du}
2\frac{Z(s)}{(Q^{b}(s))^{2}}\lambda\bar{V}^{3}
(\mathbb{E}[Y(s)\Psi^{1}(s)]-\mathbb{E}[Y(s)\Psi^{2}(s)]
-\mathbb{E}[Y(s)\Psi^{3}(s)])ds,
\nonumber
\end{align}

Let us recall that
\begin{equation*}
\overrightarrow{\mathbf{\Psi}}
=\Sigma\overrightarrow{\mathbf{W}}\circ\lambda\mathbf{e}
-\overrightarrow{\bar{V}}v_{d}\lambda\mathbf{W}_{1}\circ\lambda\mathbf{e}.
\end{equation*}

We also recall that $(\psi_{ij})_{1\leq i,j\leq 6}$ is a symmetric matrix defined as
\begin{equation*}
\psi_{ij}:=\sum_{k=1}^{6}\Sigma_{ik}\Sigma_{jk}\lambda
+\bar{V}^{i}\bar{V}^{j}v_{d}^{2}\lambda^{3},\qquad 1\leq i,j\leq 6.
\end{equation*}

Therefore, we have
\begin{equation}\label{IngredientI}
\langle\Psi^{2}\rangle_{t}
%=\sum_{j=1}^{6}\Sigma_{2j}^{2}\lambda t
%+(\bar{V}^{2}v_{d})^{2}\lambda^{3}t
=\psi_{22}t, \ \ \langle\Psi^{3}\rangle_{t}
%=\sum_{j=1}^{6}\Sigma_{3j}^{2}\lambda t
%+(\bar{V}^{3}v_{d})^{2}\lambda^{3}t
=\psi_{33}t,\ \ \langle\Psi^{2},\Psi^{3}\rangle_{t}
%=\sum_{j=1}^{6}\Sigma_{2j}\Sigma_{3j}\lambda t
%+\bar{V}^{2}\bar{V}^{3}v_{d}^{2}\lambda^{3}t
=\psi_{23}t.
\end{equation}
For any $i,j$ and $t>s$,
\begin{equation}\label{IngredientII}
\mathbb{E}[\Psi^{i}(t)\Psi^{j}(s)]
=\sum_{k=1}^{6}\Sigma_{ik}\Sigma_{jk}\lambda s
+\bar{V}^{i}\bar{V}^{j}v_{d}^{2}\lambda^{3}s=\psi_{ij}s.
\end{equation}

For any $i=1,2,3$, from
Eqn.~\eqref{YEqnII}, we can compute $\mathbb{E}[Y(t)\Psi^{i}(t)]$ as
\begin{align}\label{IngredientIII}
&\mathbb{E}[Y(t)\Psi^{i}(t)]
\\
&=-\int_{0}^{t}e^{-\int_{s}^{t}\frac{\lambda\bar{V}^{3}}{Q^{b}(u)}du}
d\mathbb{E}[\Psi^{i}(t)\Psi^{2}(s)]
-\int_{0}^{t}e^{-\int_{s}^{t}\frac{\lambda\bar{V}^{3}}{Q^{b}(u)}du}\frac{Z(s)}{Q^{b}(s)}
d\mathbb{E}[\Psi^{i}(t)\Psi^{3}(s)]
\nonumber
\\
&\qquad\qquad
+\int_{0}^{t}e^{-\int_{s}^{t}\frac{\lambda\bar{V}^{3}}{Q^{b}(u)}du}
\frac{Z(s)(\mathbb{E}[\Psi^{i}(t)\Psi^{1}(s)]
-\mathbb{E}[\Psi^{i}(t)\Psi^{2}(s)]
-\mathbb{E}[\Psi^{i}(t)\Psi^{3}(s)])}{(Q^{b}(s))^{2}}\lambda\bar{V}^{3}ds.
\nonumber
\end{align}

Next, combining Eqns. \eqref{IngredientI}, \eqref{IngredientII}, \eqref{IngredientIII},
\eqref{Z}, and \eqref{tauz}, and after some calculations, we get
%let us recall that
%$a=\lambda\bar{V}^{2}$, $b=\frac{q^{b}}{\lambda\bar{V}^{3}}$, %$c=\frac{\bar{V}^{1}-\bar{V}^{2}-\bar{V}^{3}}{\bar{V}^{3}}$, and
%\begin{align}
%&Z(t)=-\frac{a}{1+c}(b+ct)+\left(z+\frac{ab}{1+c}\right)\left(\frac{b}{b+ct}\right)^{\frac{1}{c}},
%\\
%&Q^{b}(t)=\lambda\bar{V}^{3}(b+ct).
%\end{align}
%with $a, b, c$ given in Eqn.~\eqref{abc}.
%With some computation, we see
\begin{align}\label{SubstituteI}
&\int_{0}^{t}e^{-\int_{s}^{t}\frac{2\lambda\bar{V}^{3}}{Q^{b}(u)}du}d\langle\Psi^{2}\rangle_{s}
+\int_{0}^{t}e^{-\int_{s}^{t}\frac{2\lambda\bar{V}^{3}}{Q^{b}(u)}du}
\frac{Z(s)^{2}}{Q^{b}(s)^{2}}d\langle\Psi^{3}\rangle_{s}
+\int_{0}^{t}e^{-\int_{s}^{t}\frac{2\lambda\bar{V}^{3}}{Q^{b}(u)}du}
\frac{2Z(s)}{Q^{b}(s)}d\langle\Psi^{2},\Psi^{3}\rangle_{s} \nonumber
\\
&=\lambda\sum_{j=1}^{6}\int_{0}^{t}e^{-\int_{s}^{t}\frac{2\lambda\bar{V}^{3}}{Q^{b}(u)}du}\left(\Sigma_{2j}+\frac{Z(s)}{Q^{b}(s)}\Sigma_{3j}\right)^{2}ds
+\lambda^{3}v_{d}^{2}
\int_{0}^{t}e^{-\int_{s}^{t}\frac{2\lambda\bar{V}^{3}}{Q^{b}(u)}du}\left(\bar{V}^{2}+\frac{Z(s)}{Q^{b}(s)}\bar{V}^{3}\right)^{2}ds
\nonumber
%&=\lambda\sum_{j=1}^{6}\int_{0}^{t}e^{-\int_{s}^{t}\frac{2}{b+cu}du}
%\left[\left(\Sigma_{2j}-\frac{\Sigma_{3j}a}{(1+c)\lambda\bar{V}^{3}}\right)
%+\Sigma_{3j}\left(z+\frac{ab}{1+c}\right)\frac{b^{\frac{1}{c}}}{\lambda\bar{V}^{3}}
%(b+cs)^{-\frac{1}{c}-1}\right]^{2}ds
%\nonumber
%\\
%&\qquad
%+\frac{\lambda^{3}v_{d}^{2}}{6}\sum_{j=1}^{6}\int_{0}^{t}e^{-\int_{s}^{t}\frac{2}{b+cu}du}
%\left[\left(\frac{c}{1+c}\bar{V}^{2}\right)
%+\bar{V}^{3}\left(z+\frac{ab}{1+c}\right)\frac{b^{\frac{1}{c}}}{\lambda\bar{V}^{3}}
%(b+cs)^{-\frac{1}{c}-1}\right]^{2}ds
%\nonumber
%\\
%&=\frac{1}{(b+ct)^{\frac{2}{c}}}\sum_{j=1}^{6}
%\left[\lambda\left(\Sigma_{2j}-\frac{\Sigma_{3j}a}{(1+c)\lambda\bar{V}^{3}}\right)^{2}
%+\frac{\lambda^{3}v_{d}^{2}}{6}\left(\frac{c}{1+c}\bar{V}^{2}\right)^{2}\right]
%\int_{0}^{t}(b+cs)^{\frac{2}{c}}ds
%\nonumber
%\\
%&\qquad
%+\frac{1}{(b+ct)^{\frac{2}{c}}}\sum_{j=1}^{6}
%2\left[\lambda\left(\Sigma_{2j}-\frac{\Sigma_{3j}a}{(1+c)\lambda\bar{V}^{3}}\right)\Sigma_{3j}+\frac{\lambda^{3}v_{d}^{2}}{6}\frac{c}{1+c}\bar{V}^{2}\bar{V}^{3}\right]
%\nonumber
%\\
%&\qquad\qquad\cdot
%\left(z+\frac{ab}{1+c}\right)
%\int_{0}^{t}\frac{b^{\frac{1}{c}}}{\lambda\bar{V}^{3}}(b+cs)^{\frac{1}{c}-1}ds
%\nonumber
%\\
%&\qquad\qquad
%+\frac{1}{(b+ct)^{\frac{2}{c}}}\sum_{j=1}^{6}\left[\lambda(\Sigma_{3j})^{2}+\frac{\lambda^{3}v_{d}^{2}}{6}(\bar{V}^{3})^{2}\right]
%\left(z+\frac{ab}{1+c}\right)^{2}
%\int_{0}^{t}\frac{b^{\frac{2}{c}}}{\lambda^{2}(\bar{V}^{3})^{2}}(b+cs)^{-2}ds
%\nonumber
\\
&=\frac{(b+ct)^{\frac{2}{c}+1}-b^{\frac{2}{c}+1}}{(2+c)(b+ct)^{\frac{2}{c}}}\sum_{j=1}^{6}
\left(\lambda\left(\Sigma_{2j}-\frac{\Sigma_{3j}a}{(1+c)\lambda\bar{V}^{3}}\right)^{2}
+\frac{\lambda^{3}v_{d}^{2}}{6}\left(\frac{c}{1+c}\bar{V}^{2}\right)^{2}\right)
\nonumber
\\
&\qquad
+\frac{b^{\frac{1}{c}}}{\lambda\bar{V}^{3}}\frac{(b+ct)^{\frac{1}{c}}-b^{\frac{1}{c}}}{(b+ct)^{\frac{2}{c}}}
\left(z+\frac{ab}{1+c}\right)\sum_{j=1}^{6}
2\left(\lambda\left(\Sigma_{2j}-\frac{\Sigma_{3j}a}{(1+c)\lambda\bar{V}^{3}}\right)\Sigma_{3j}+\frac{\lambda^{3}v_{d}^{2}}{6}\frac{c}{1+c}\bar{V}^{2}\bar{V}^{3}\right)
\nonumber
\\
&\qquad\qquad
+\frac{t}{(b+ct)^{\frac{2}{c}+1}}\frac{b^{\frac{2}{c}-1}}{\lambda^{2}(\bar{V}^{3})^{2}}
\sum_{j=1}^{6}\left(\lambda(\Sigma_{3j})^{2}+\frac{\lambda^{3}v_{d}^{2}}{6}(\bar{V}^{3})^{2}\right)
\left(z+\frac{ab}{1+c}\right)^{2},
\end{align}
and
\begin{align}
&\mathbb{E}[Y(t)(\Psi^{1}(t)-\Psi^{2}(t)-\Psi^{3}(t))] \nonumber
\\
&=-(\psi_{12}-\psi_{22}-\psi_{32})\int_{0}^{t}e^{-\int_{s}^{t}\frac{\lambda\bar{V}^{3}}{Q^{b}(u)}du}ds
-(\psi_{13}-\psi_{23}-\psi_{33})\int_{0}^{t}e^{-\int_{s}^{t}\frac{\lambda\bar{V}^{3}}{Q^{b}(u)}du}\frac{Z(s)}{Q^{b}(s)}ds
\nonumber
\\
&\qquad\qquad
+(\psi_{11} + \psi_{22} + \psi_{33} - \psi_{12} - \psi_{13} - \psi_{21} - \psi_{31} + \psi_{23}+\psi_{32})\int_{0}^{t}e^{-\int_{s}^{t}\frac{\lambda\bar{V}^{3}}{Q^{b}(u)}du}
\frac{Z(s)s}{(Q^{b}(s))^{2}}\lambda\bar{V}^{3}ds
\nonumber
\\
%&=-(\psi_{12}-\psi_{22}-\psi_{32})\int_{0}^{t}\left(\frac{b+cs}{b+ct}\right)^{\frac{1}{c}}ds
%\nonumber
%\\
%&\qquad
%-(\psi_{13}-\psi_{23}-\psi_{33})\int_{0}^{t}\left(\frac{b+cs}{b+ct}\right)^{\frac{1}{c}}
%\left[-\frac{a}{(1+c)\lambda\bar{V}^{3}}+\left(z+\frac{ab}{1+c}\right)\frac{b^{\frac{1}{c}}}{\lambda\bar{V}^{3}}
%(b+cs)^{-\frac{1}{c}-1}\right]ds
%\nonumber
%\\
%&
%+(\psi_{11} + \psi_{22} + \psi_{33} - \psi_{12} - \psi_{13} - \psi_{21} - \psi_{31} + %\psi_{23}+\psi_{32})
%\nonumber
%\\
%&\qquad\qquad\qquad\qquad
%\cdot\int_{0}^{t}\left(\frac{b+cs}{b+ct}\right)^{\frac{1}{c}}
%\left[-\frac{a}{(1+c)\lambda\bar{V}^{3}}+\left(z+\frac{ab}{1+c}\right)
%\frac{b^{\frac{1}{c}}}{\lambda\bar{V}^{3}}
%(b+cs)^{-\frac{1}{c}-1}
%\right]\frac{s}{b+cs}ds
%\nonumber
%\\
%&=\frac{\alpha}{(b+ct)^{\frac{1}{c}}}\int_{0}^{t}(b+cs)^{\frac{1}{c}}ds
%+\frac{\beta}{(b+ct)^{\frac{1}{c}}}\int_{0}^{t}(b+cs)^{-1}ds
%\nonumber
%\\
%&\qquad
%+\frac{\gamma}{(b+ct)^{\frac{1}{c}}}\int_{0}^{t}(b+cs)^{\frac{1}{c}-1}ds
%+\frac{\delta}{(b+ct)^{\frac{1}{c}}}\int_{0}^{t}(b+cs)^{-2}ds
%\nonumber
%\\
%&=\frac{\alpha}{(b+ct)^{\frac{1}{c}}}\frac{1}{c+1}\left[(b+ct)^{\frac{1}{c}+1}-b^{\frac{1}{c}+1}\right]
%+\frac{\beta}{(b+ct)^{\frac{1}{c}}}\frac{1}{c}\log\left(1+\frac{ct}{b}\right)
%\nonumber
%\\
%&\qquad
%+\frac{\gamma}{(b+ct)^{\frac{1}{c}}}\left[(b+ct)^{\frac{1}{c}}-b^{\frac{1}{c}}\right]
%+\frac{\delta}{(b+ct)^{\frac{1}{c}}}\frac{t}{b(b+ct)}
%\nonumber
%\\
&=\hat{\alpha}(b+ct)+\hat{\beta}(b+ct)^{-\frac{1}{c}}
+\hat{\gamma}\frac{\log(b+ct)}{(b+ct)^{\frac{1}{c}}}
+\hat{\delta}
+\hat{\eta}(b+ct)^{-\frac{1}{c}-1},
\end{align}
where $\alpha,\beta,\gamma,\delta$ are defined in
Eqn.~\eqref{alphaDefn} and $\hat{\alpha},\hat{\beta},\hat{\gamma},\hat{\delta},\hat{\eta}$
are defined in
Eqn.~\eqref{hatalphaDefn}. Therefore,
\begin{align}\label{SubstituteII}
&\int_{0}^{t}e^{-\int_{s}^{t}\frac{2\lambda\bar{V}^{3}}{Q^{b}(u)}du}
\frac{2Z(s)}{(Q^{b}(s))^{2}}\lambda\bar{V}^{3}
\mathbb{E}[Y(s)(\Psi^{1}(s)-\Psi^{2}(s)-\Psi^{3}(s))]ds \nonumber
\\
&=\frac{2}{(b+ct)^{\frac{2}{c}}}\int_{0}^{t}
(b+cs)^{\frac{2}{c}-1}
\left(-\frac{a}{(1+c)\lambda\bar{V}^{3}}
+\left(z+\frac{ab}{1+c}\right)\frac{b^{\frac{1}{c}}}{\lambda\bar{V}^{3}}(b+cs)^{-\frac{1}{c}-1}\right)
\nonumber
\\
&\qquad\qquad\qquad
\cdot\left(\hat{\alpha}(b+cs)+\hat{\beta}(b+cs)^{-\frac{1}{c}}
+\hat{\gamma}\frac{\log(b+cs)}{(b+cs)^{\frac{1}{c}}}
+\hat{\delta}
+\hat{\eta}(b+cs)^{-\frac{1}{c}-1}\right)ds
%&=-\frac{2a}{(b+ct)^{\frac{2}{c}}(1+c)\lambda\bar{V}^{3}}
%\nonumber
%\\
%&\qquad\cdot
%\int_{0}^{t}\left[\hat{\alpha}(b+cs)^{\frac{2}{c}}
%+\hat{\beta}(b+cs)^{\frac{1}{c}-1}
%+\hat{\gamma}\frac{\log(b+cs)}{(b+cs)^{1-\frac{1}{c}}}
%+\hat{\delta}(b+cs)^{\frac{2}{c}-1}
%+\hat{\eta}(b+cs)^{\frac{1}{c}-2}\right]ds
%\nonumber
%\\
%&+\frac{2}{(b+ct)^{\frac{2}{c}}}\left(z+\frac{ab}{1+c}\right)\frac{b^{\frac{1}{c}}}{\lambda\bar{V}^{3}}
%\nonumber
%\\
%&\qquad
%\cdot\int_{0}^{t}\left[\hat{\alpha}(b+cs)^{\frac{1}{c}-1}
%+\hat{\beta}(b+cs)^{-2}+\hat{\gamma}\frac{\log(b+cs)}{(b+cs)^{2}}
%+\hat{\delta}(b+cs)^{\frac{1}{c}-2}+\hat{\eta}(b+cs)^{-3}\right]ds.
%\nonumber
%\\
%&=-\frac{2a}{(b+ct)^{\frac{2}{c}}(1+c)\lambda\bar{V}^{3}}
%\cdot
%\bigg[\hat{\alpha}\frac{(b+ct)^{\frac{2}{c}+1}-b^{\frac{2}{c}+1}}{2+c}
%+[\hat{\beta}-\hat{\gamma}c][(b+ct)^{\frac{1}{c}}-b^{\frac{1}{c}}]
%\nonumber
%\\
%&\qquad
%+\hat{\gamma}\left[(b+ct)^{\frac{1}{c}}\log(b+ct)-b^{\frac{1}{c}}\log(b)\right]
%+\frac{\hat{\delta}}{2}[(b+ct)^{\frac{2}{c}}-b^{\frac{2}{c}}]
%+\frac{\hat{\eta}}{1-c}[(b+ct)^{\frac{1}{c}-1}-b^{\frac{1}{c}-1}]\bigg]
%\nonumber
%\\
%&+\frac{2}{(b+ct)^{\frac{2}{c}}}\left(z+\frac{ab}{1+c}\right)\frac{b^{\frac{1}{c}}}{\lambda\bar{V}^{3}}
%\cdot\bigg[\hat{\alpha}[(b+ct)^{\frac{1}{c}}-b^{\frac{1}{c}}]
%+[\hat{\beta}+\hat{\gamma}]\frac{t}{b(b+ct)}
%\nonumber
%\\
%&\qquad
%+\frac{\hat{\gamma}}{c}\left[\frac{\log b}{b}-\frac{\log(b+ct)}{b+ct}\right]
%+\frac{\hat{\delta}}{1-c}[(b+ct)^{\frac{1}{c}-1}-b^{\frac{1}{c}-1}]+\frac{\hat{\eta}}{2c}[b^{-2}-(b+ct)^{-2}]\bigg].
%\nonumber
\end{align}
Hence, we get the desired result by substituting
Eqn.~\eqref{SubstituteI} and
Eqn.~\eqref{SubstituteII} into
Eqn.~\eqref{YVariance}.
\end{proof}

\end{document}